\documentclass[12pt]{article}
\usepackage{amsmath}
 \usepackage{enumerate}
\usepackage{natbib} 
 \usepackage{amssymb}
 \usepackage{fullpage}	
\usepackage{amsthm}
\usepackage{bbm}
\usepackage{bm}
\usepackage{graphicx} 
\usepackage{algpseudocode}
\algnewcommand\algorithmicinput{\textbf{INPUT:}}
\algnewcommand\INPUT{\item[\algorithmicinput]}
\algnewcommand\algorithmicoutput{\textbf{OUTPUT:}}
\algnewcommand\OUTPUT{\item[\algorithmicoutput]}
\usepackage{hyperref}[]
\hypersetup{
    colorlinks=true,
    linkcolor=blue,
    filecolor=magenta,      
    urlcolor=cyan,
    citecolor=blue,
}
\usepackage{cleveref}
\usepackage{subcaption}
\usepackage{xcolor}
\usepackage{float}
\usepackage{xr}
\newtheorem{lemma}{Lemma}
\newtheorem{proposition}{Proposition}
\newtheorem{corollary}{Corollary}
\newtheorem{definition}{Definition}
\newtheorem{remark}{Remark}
\newtheorem{theorem}{Theorem}

\newtheorem{assumption}{Assumption}

\newcommand{\ou}{\"{o}}

\newcommand{\op}{\text{op}}

\newcommand{\s}{\text{span}}
\newcommand{\row}{\text{Row}}
\newcommand{\col}{\text{Col}}
\newcommand{\lb}{ \langle}
\newcommand{\rb}{ \rangle}

\newcommand{\mk}{\mathbb K}
 
\newcommand{\h}{\mathcal H}
\newcommand{\K}{\mathcal K}
\newcommand{\bu}{\mathbf u}

\newcommand{\lt}{ {\mathcal L ^2}}

\DeclareMathOperator*{\argmin}{arg\,min}

\usepackage{setspace}
 \newcommand{\blind}{1}

\begin{document}

 \setlength{\abovedisplayskip}{4pt}
\setlength{\belowdisplayskip}{4pt}
\setlength{\abovedisplayshortskip}{2pt}
\setlength{\belowdisplayshortskip}{2pt}

\if1\blind
{
	\title{Functional Autoregressive Processes in Reproducing Kernel Hilbert Spaces}
 	\author{
 		Daren Wang\\
	University of Chicago
	 		\and
 		Zifeng Zhao\\
 		University of Notre Dame
 		\and
 	    Rebecca Willett\\
 	University of Chicago
 		\and
 	Chun Yip Yau\\
 	Chinese University of Hong Kong
 	}
	\date{}	
	\maketitle
} \fi

\if0\blind
{
	\title{Functional Autoregressive Processes in Reproducing Kernel Hilbert Spaces}
	\author{}
	\date{}
	\maketitle
} \fi

\maketitle
\begin{abstract}
	We study the estimation and prediction of functional autoregressive~(FAR) processes, a statistical tool for modeling functional time series data. Due to the infinite-dimensional nature of FAR processes, the existing literature addresses its inference via dimension reduction and theoretical results therein require the (unrealistic) assumption of fully observed functional time series. We propose an alternative inference framework  based on Reproducing Kernel Hilbert Spaces~(RKHS). Specifically, a nuclear norm regularization method  is proposed for estimating the transition operators of the FAR process directly from discrete samples of the functional time series. We derive a representer theorem for the FAR process, which enables infinite-dimensional inference without dimension reduction. Sharp theoretical guarantees are established under the (more realistic) assumption that  we only have finite discrete samples of the FAR process. Extensive numerical experiments and a real data application of energy consumption prediction are further conducted to illustrate the promising performance of the proposed approach compared to the state-of-the-art methods in the literature.
\end{abstract}


\clearpage

\section{Introduction}\label{sec:intro}
Functional Data Analysis (FDA) has emerged as an important area of modern statistics as it provides effective tools for analyzing complex data. 
As described in the excellent monographs by \cite{Ramsay2005}, \cite{Ferraty2006} and \cite{Horvath2012}, FDA considers the analysis and theory of data that can be viewed in the form of functions, offering a natural and parsimonious solution to a variety of problems that are difficult to cast into classical statistical frameworks designed for scalar and vector valued data. 

An important type of functional data is functional time series \citep{Hoermann2010, aue2015prediction}, where the functional observations are collected in a sequential manner. A typical scheme is a continuous-time record that can be partitioned into natural consecutive time intervals, such as hours, days or years, where similar behavior is expected across intervals. Common examples of functional time series include the intraday return curves or volatility curves of a stock market index and the daily or annual patterns of meteorological and environmental data such as temperature or precipitation.

Formally speaking, a functional time series takes the form $\{X_t(s), s\in[a,b]\}_{ t\in \mathbb{Z}}$, where each observation $X_t(s)$ is a (random) function defined for $s$ taking values in some compact interval $[a,b]$. By rescaling if needed, throughout the paper, without loss of generality, we assume $[a,b]=[0,1].$ \textcolor{black}{Due to the intrinsic high dimensionality of functional observations, classical univariate and multivariate time series methods, such as (vector) autoregressive models, may fail to track the dynamics of functional time series and thus are unable to provide accurate prediction. See more discussions in, for example, \cite{bosq2000linear}, \cite{Hyndman2007} and \cite{Shang2013}. Thus, a key task for functional time series analysis is the design and estimation of a reliable statistical model tailored for the functional nature of the data, which serves as the foundation for understanding the behavior of the data and providing accurate prediction.}

The most widely-used functional time series model proven to work well in practice is the functional autoregressive~(FAR) process in $\lt$~\citep{bosq2000linear}, where $\lt$ denotes the Hilbert space of square-integrable functions on $[0,1]$ equipped with the inner product $ \lb f,g \rb_{\lt}  = \int_0^1 f (r)g(r) dr$. Generally speaking, a functional time series $\{X_t\}$ follows an FAR process~(in $\lt$) of order $D$ if
\begin{align}\label{eq:farL2_classic}
X_t(\cdot)=\mu(\cdot) + \sum_{d=1}^D \Psi_d(X_{t-d})(\cdot) + \epsilon_{t}(\cdot),
\end{align}
where $\mu(\cdot)\in \lt$ is a deterministic function, $\epsilon_{t}(\cdot)\in \lt$ are \textit{i.i.d.}\ zero-mean noise functions, and $\{\Psi_d\}_{d=1}^D$ are bounded linear operators mapping $\lt\to \lt$.  Conditions for the existence of a stationary and causal solution of \eqref{eq:farL2_classic} in $\lt$ and other theoretical properties of FAR in $\lt$ are studied extensively in \cite{bosq2000linear}. Besides FAR, other types of models and prediction approaches, such as nonlinear kernel-distance based methods, for functional times series are considered in \cite{Bosq1998}, \cite{Besse2000}, \cite{Antoniadis2006}, \cite{Kokoszka2017} and \cite{BuenoLarraz2019}, among others. In this paper, we focus on the FAR process.

Due to the infinite-dimensional nature of the functional space $\lt$, existing literature addresses the inference of FAR mainly via dimension reduction, where a projection onto a finite basis is conducted to facilitate the estimation of $\{\Psi_d\}_{d=1}^D$ and the prediction of future realizations. Most literature reduces the dimension via functional principal component analysis~(FPCA), where the finite basis is chosen as the leading $p$ functional principal components~(FPC) of an estimated covariance operator of the underlying FAR process~(e.g.\ the sample covariance operator based on the observations $\{X_t\}_{t=1}^T$). See \cite{Besse1996}, \cite{bosq2000linear}, \cite{Besse2000}, \cite{Hyndman2009}, \cite{Didericksen2012} and \cite{aue2015prediction} for influential works on FPCA-based approaches. Notable methods based on other basis such as wavelets or predictive factors include \cite{Antoniadis2003} and \cite{Kargin2008}. 

Theoretical justification for dimension reduction based methods can be found in, for example, \cite{bosq2000linear} and \cite{Kargin2008}, where results such as the consistency of $\{\widehat{\Psi}_d\}_{d=1}^D$ are provided. For these results to hold, a typical condition is that the number of basis  elements $p$ grows with the sample size $T$ at a rate that implicitly depends on the intricate interrelation of eigenvalues and spectral gaps of the true covariance operator of the underlying FAR process, making the convergence rate derived therein rather opaque and case-specific. As a result, there seems to be no clear guidance on the selection of $p$ in practice, with most literature using a heuristic threshold (e.g. 80\%) on the cumulative variance of FPCs, see for example \cite{Didericksen2012}. An exception is \cite{aue2015prediction}, where a novel functional prediction error criterion is developed for the selection of $p.$

A notable limitation of the current FAR literature is that existing estimation methods and theoretical results require fully observed functional time series $\{X_t(s), s\in[0,1]\}_{t=1}^T$. However, this is an unrealistic assumption as, in reality, the FAR process is measured discretely and observations instead take the form $\{X_t(s_i), 1\leq i\leq n\}_{t=1}^T$, where $\{s_i\}_{i=1}^n$ denotes $n$ discrete grid points in $[0,1]$. In practice, the aforementioned methods typically rely on an extra smoothing step to convert discrete measurements $\{X_t(s_i), 1\leq i\leq n\}_{t=1}^T$ into (estimated) fully functional data $\{\widetilde{X}_t(s), s\in[0,1]\}_{t=1}^T$, and the statistical analysis is performed on $\{\widetilde{X}_t(s), s\in[0,1]\}_{t=1}^T$. Intuitively, the smoothing step may have substantial impact on the inference of FAR~(as is illustrated via simulation studies in Section \ref{subsec:simu_far1}). However, to establish theoretical guarantees, existing literature commonly ignores the smoothing error and assumes the analysis is conducted on the true functional time series $\{X_t(s), s\in[0,1]\}_{t=1}^T$, possibly due to technical difficulties. One ramification is that the derived convergence rate therein typically only involves $T$ but not $n$.

In this paper, we study the inference of FAR processes through the lens of Reproducing Kernel Hilbert Spaces~(RKHS, \cite{Wahba1990}) and propose new estimation and prediction procedures for FAR without dimension reduction. Specifically, we consider a refined FAR process in RKHS (see detailed definition in Section \ref{subsec:FARinRKHS}),
\begin{align}\label{eq:farL2}
X_t(\cdot)=\mu(\cdot) + \sum_{d=1}^D \int_0^1 A_d(\cdot, s) X_{t-d}(s)ds + \epsilon_{t}(\cdot),
\end{align}
where the bounded linear operators $\{\Psi_d\}_{d=1}^D$ take the explicit form of integral operators with bivariate kernels $\{A_d(r,s):[0,1]\times[0,1] \to \mathbb R\}_{d=1}^D$. Note that statistically speaking, \eqref{eq:farL2} is essentially equivalent to \eqref{eq:farL2_classic}, as consistent estimation of $\{\Psi_d\}_{d=1}^D$ requires them to be Hilbert-Schmidt operators~\citep[e.g.][]{bosq2000linear,Kargin2008}, which indeed implies $\{\Psi_d\}_{d=1}^D$ can be written as integral operators with square-integrable kernels~\citep{Heil2018}.

By viewing $\{A_d(r,s)\}_{d=1}^D$ as compact linear operators in RKHS, we first derive its reproducing property, which facilitates its consistent estimation directly based on discrete measurements without smoothing. We then propose a nuclear norm regularization method for the estimation of $\{A_d(r,s)\}_{d=1}^D$ and further derive the representer theorem, which enables the (infinite-dimensional) inference without dimension reduction. For efficient implementation, we reformulate the regularization of functional operators into the well-studied trace norm minimization in the machine learning literature, which can be readily solved via the accelerated gradient method~\citep{Ji2009}. The consistency and explicit convergence rate (incorporating both $T$ and $n$) of the proposed procedure are provided. To our best knowledge, this is the first rigorous theoretical guarantee for estimation and prediction of FAR processes based on discrete observations of functional time series.


The rest of the paper is organized as follows. Section \ref{sec:RKHS} gives a brief review of RKHS and defines the FAR process in RKHS. Section \ref{sec:FARD} proposes the penalized nuclear norm estimator for FAR and studies its theoretical properties. The promising performance of the proposed method over existing procedures is demonstrated via extensive numerical experiments in Section \ref{sec:simulation} and a real data application of energy consumption prediction in Section \ref{sec:realdata}. Section \ref{sec:conclusion} concludes with a discussion. Some notations used throughout the paper are defined as follows. Denote $\| f\|_\lt ^2 =\lb f,f\rb_\lt $ and $\|f\|_{\infty}:=\sup_{s \in [0,1]} |f(s)|$. For a matrix $W$, denote $\|W\|_F$ as its Frobenius norm and $\|W\|_*$ as its trace norm. We omit $[0,1]$ in the integral whenever the domain of functions is clear.

\section{Functional Autoregressive Processes in RKHS}\label{sec:RKHS}

\subsection{RKHS and  compact linear operators}\label{subsec:RKHS_overview}
In this subsection, we briefly review the Reproducing Kernel Hilbert Spaces (RKHS) and introduce the class of compact linear operator, which is later used for defining the FAR process in RKHS.

Let $\mathbb K :[0,1]\times [0,1]  \to \mathbb R^+$ be a reproducing kernel and $\mathcal H \subset \lt$ be the corresponding reproducing kernel Hilbert space. Denote $\lb \cdot, \cdot \rb_\h$ as the inner product for $\h$ and define $\|f\|^2_\h=\lb f,f\rb_\h$ as the RKHS norm. The eigen-expansion of $\mathbb K$ has the form
\begin{align} \label{eq:eigen expansion of kernel}
\mathbb K (r,  s) =  \sum_{k=1}^\infty   \mu_k \phi _k (r) \phi_k(s),
\end{align}
where $\{\phi_k\}_{k=1}^\infty$  is an orthonormal basis of $\lt $ such that $ \| \phi_k\|_\lt ^2 =1$ and $\| \phi_k\|_\mathcal H ^2 = {1}/{\mu_k} $. 
Thus for $f=\sum_{k=1}^\infty a_k \phi_k $ and  $g=\sum_{k=1}^\infty b_k \phi_k$, we have $  \lb f ,g\rb_\mathcal H  = \sum_{k=1}^\infty  {a_kb_k}/{\mu_k }$. In particular, $f(r) = \lb f, \mathbb K (\cdot, r) \rb_\mathcal H$ for $f\in\h $, which is known as the reproducing property of RKHS. In \Cref{assume:rkhs}, we impose some mild regularity conditions on the $\h$ that we study in this paper.

\begin{assumption}\label{assume:rkhs}
\
\\
{\bf a.} There exists an absolute constant $C_\h$ such that for any $f, g\in \h$, it holds that
\begin{align} \label{eq:regularity of norm in  h}
\|f g\|_{\h }\le C_\h \|f\|_\h \|g\|_\h. 
\end{align}   
{\bf b.} There exists a constant $C_\mathbb K $ such that  $\sup_{0\le r\le  1} \mathbb K(r,r) \le C_\mathbb K$.
\end{assumption}

\Cref{assume:rkhs}{\bf a} is a mild regularity condition on $\h$ and is mainly made for technical simplicity in the proof. A wide class of RKHS satisfies \Cref{assume:rkhs}{\bf a}. For a concrete example, consider the commonly used Sobolev space $\mathcal H = W^{\alpha,2}$ on $[0,1]$ where $W^{\alpha ,2}  : = \{ f :  \| f\|_{ W^{\alpha ,2} } ^2   =   \|f \|_\lt ^2 +  \sum_{k =1} ^\alpha \|f ^{(k) } \|_{\lt } ^2 < \infty  \}.$ From the definition of $W^{\alpha,2}$, it   is straightforward to show that there exists a constant $C_\alpha$ such that $\|fg\|_{W^{\alpha, 2}}\le C_\alpha \|f\|_{W^{\alpha, 2}}  \|g\|_{W^{\alpha, 2}}.$ For illustration, letting $\alpha=1$, we have $ \| (fg )'    \|_\lt  ^2  \le 2    \int  ( f '   (s) g (s)  )^2 + ( f(s) g' (s) )^2     ds 	\le 2  \|f'\|_{ \lt } ^2 \|g\|_{\infty }^2    +2  \|f\|_{\infty}^2  \| g'\|_{ \lt } ^2  \le 4 \|f \|_{  W^{1,2} } ^2 \|g\|_{  W^{1,2} }  ^2$, therefore it suffices to take $C_\h =\sqrt 5 $ in \eqref{eq:regularity of norm in  h} for $W^{1,2}$. We refer to \cite{brezis2011functional} for a comprehensive introduction to Sobolev spaces. Throughout the paper, we assume $C_\h =1$ for notational simplicity as the theoretical analysis holds for any constant $C_\h$.

\Cref{assume:rkhs}{\bf b} is a widely-used assumption on the kernel function $\mathbb K$ and is satisfied by most commonly used kernels. Note that \Cref{assume:rkhs}{\bf b} implies that for any $s\in [0,1]$ and any $f\in \h$,
$$f(s ) = \lb f  , \mathbb K_s(\cdot)\rb_\h \le \|f\|_\h \|\mathbb K_s(\cdot)\|_\h  \le \|f\|_\h  \sqrt {C_\mathbb K }  . $$
As a result, $\|f\|_\lt \le \|f\|_{\infty }\le \sqrt {C_\mathbb K }   \|f\|_\h $.
Note that for any positive constant $\beta$,  the two kernel functions $ \mathbb K $ and $\beta\mathbb K$ generate the same function space. Thus with rescaling if necessary, we assume without loss of generality that $C_\mathbb K=1$ throughout the paper.  

We now introduce the compact linear operator mapping $\h \to \h$, which is used to regulate the transition kernels $\{A_d(r,s):[0,1]\times [0,1]\to \mathbb{R}\}_{d=1}^D$ of the FAR process in $\h$~(see \eqref{eq:farL2} in Section \ref{sec:intro}). Denote $A(r,s)$ as a function from $[0,1]  \times [0,1] \to \mathbb R$   such that $A(\cdot, s)\in \h$ for any $s\in[0,1]$ and $A(r,\cdot)\in \h$ for any $r\in[0,1]$. Thus, $A(r,s)$ induces a linear operator on $\h$ via
\begin{align}\label{eq:bivariate_kernel}
 A[v] (r) : = \langle A(r, \cdot),v(\cdot) \rangle_\h, \ r\in[0,1], \text{ for any } v\in \h.
\end{align}
If we further have $A[v] \in \h $ for all $v\in \h$, the bivariate function $A(r,s)$ can be viewed as a linear operator $A$ mapping $\h\to \h$ in the light of \eqref{eq:bivariate_kernel}. To utilize the smoothness of the RKHS, we focus on the class of linear operators $A : \mathcal H \to  \mathcal H$ that are compact, as defined in the following definition. 
\begin{definition} \label{definition:compact}
	A linear operator $A:\mathcal H \to \mathcal H$  is said to be  compact   if the image of any bounded set in $\mathcal H$ is (relatively) compact.
\end{definition}
The space of compact operators on  a Hilbert space  is the closure of the space of finite rank operators. It is well known that in the classical functional analysis~(see e.g. \cite{brezis2011functional}), the compact operators share many desirable properties with matrices  such as the existence of the singular value decomposition, which facilitates its theoretical analysis.

Denote $\mathcal C$ as the space of compact linear operators on $\h$ and denote $\Phi_ k = \sqrt {\mu_k } \phi_k$. For a compact linear operator $A\in \mathcal{C}$, define $a_{ij} : = A[\Phi_i, \Phi_j] : =  \lb  A[   \Phi_j] ,   \Phi_i\rb_\h$. We have $A[f,g] : = \lb A [g],f \rb_\h  =  \sum_{i,j=1}^\infty a_{ij}    \lb \Phi_i, f\rb_\mathcal H  \cdot  \lb   \Phi_j,g\rb_\mathcal H$, which implies the useful decomposition 
\begin{align}\label{eq:expression of A}
A(r,s) =   A [\mathbb K (\cdot, r) ,\mathbb K (\cdot, s) ] = \sum_{i,j=1}^\infty a_{ij}  
\Phi_i(r) \Phi_j(s),
\end{align}
which is essentially the reproducing property of $A$ and is used to derive the Representer theorem later in \Cref{lemma:representer}. Note that by \eqref{eq:expression of A}, any compact linear operator $A \in \mathcal{C}$ can be viewed as a bivariate function $A(r,s)$ such that $A[f](r)=\langle A(r, \cdot),f(\cdot) \rangle_\h, r\in[0,1],$ for any $f\in \h$.

We now define norms of the operator $A:\h \to \h$ that are later used to regulate its smoothness. Denote $\|A\|_{ \h, *} $ as the nuclear norm of $A$ such that $\| A\|_{\h, *}  = \sum_{i=1}^ \infty  \sqrt {\lambda_i},$ where $\lambda_i$ is the $i$-th eigenvalue of $A^\top A$ and  $A^\top $ is the adjoint operator of $A$ such that $\lb A[u],v\rb_\h=\lb u, A^\top [v]\rb_\h$. In addition, define $\text{rank}(A) =  \sum_{i=1}^\infty \mathbbm 1_{ \{ \lambda_i \not = 0 \} }$. 
An operator $A$ is said to be a bounded operator if its operator norm  $\|A\|_{\h , \op}  $ is finite, where 
\begin{align}\label{eq:bounded bivariate operator}
\| A\|_{\h , \op} := \sup_{\|u\|_\h\leq 1,\|v\|_\h\leq 1} \lb A[u],v\rb_\h =  \sup_{ \sum_{i=1}^\infty u_i^2  \le 1 ,\sum_{j=1}^\infty v_j^2   \le 1 } \sum_{i,j=1}^\infty   a_{ij}u_iv_j.
\end{align}

  
\subsection{FAR processes in an RKHS}\label{subsec:FARinRKHS}
In this subsection, based on the compact linear operators discussed in the previous section, we define the FAR process in an RKHS and further study its probabilistic properties.

\begin{definition} \label{definition:AR model}
	For an RKHS $\mathcal{H}$, a functional time series $\{X_t\}_{t=1}^T \subset \mathcal  H$ is said to follow a functional autoregressive process of order $D$ in $\mathcal{H}$, hereafter FAR($D$), if
	\begin{align}\label{eq:AR model}
	X_{t} (r)  = \sum_{d=1}^D   \int A_d ^*(r,s) X_{t-d} (s) ds     +\epsilon_{t} (r), \text{ for } r\in [0,1],
	\end{align} 
	where $\{ \epsilon_t \}_{t=1}^T \subset \mathcal H $ is a collection of i.i.d. functional noise and the transition operators $\{A _d ^* \}_{d=1}^D \subset \mathcal C$ are compact linear operators on $\h$.
\end{definition}
\noindent Note that compared to \eqref{eq:farL2}, for ease of presentation, \Cref{definition:AR model} does not include the deterministic function $\mu(\cdot)$, as $\mu(\cdot)$ can be easily removed by centering $X_t$ via $X_t-E(X_t)$ for stationary $\{X_t\}_{t=1}^T$.

Definition \ref{definition:AR model} requires that the functional time series $\{ X _t\}_{t=1}^T$ resides in $\h$. This is an intuitive and necessary condition which allows us to estimate the transition operators  $\{A_d^*\}_{d=1}^D $ from discrete measurements of $\{X_t\}_{t=1}^T$. As discussed in the introduction, most existing FAR literature assumes that $\{X_t\}_{t=1}^T $ are $\lt $ functions with no additional regularity assumptions. We remark that in this latter setting, it is theoretically impossible to  recover the transition operators  $\{A_d^*\}_{d=1}^D $ from discrete measurements  of $ \{X_t\}_{t=1}^T $. In fact, in this case we cannot even consistently estimate one function $X_t$ without extra regularity assumptions, as suggested by  existing information theoretical lower bounds discussed in \cite{mendelson2002geometric} and  \cite{raskutti2012minimax}.

In \Cref{assume:regularity of FARD}, we introduce regularity conditions on the transition operators $\{A_d^*\}_{d=1}^D$ and the i.i.d. functional noise $\{\epsilon_{t}\}_{t=1}^T $ of the FAR($D$) process.
\begin{assumption} \label{assume:regularity of FARD}
	~\\
	{\bf a.} The transition operators $\{A_d^*\}_{d=1}^D \subset \mathcal C$ is a collection of compact linear operators with finite nuclear norm such that $\max_{1\leq d\leq D}\| A^*_d\|_{\h, * } < C$ for some constant $C$.
	\\
	{\bf b.} The functional noise is zero-mean with $E\epsilon _t(s)=0$ for all $s\in [0,1]$. In addition, there exist positive constants $C_\epsilon$ and $\kappa_\epsilon$ such that 
	\begin{align}
	&P( \| \epsilon_t\| _\h \le C_\epsilon ) =1, \label{eq:variance of the functional noise1}\\ 
	&E\left( \int    v(s)   \epsilon_t  (s)ds  \right) ^2   \ge \kappa_\epsilon  \|v\|_{\lt}^2  \ \text{ for all } v\in \h. \label{eq:variance of the functional noise2}
	\end{align}
\end{assumption}
 
\Cref{assume:regularity of FARD}{\bf a} essentially  requires that the transition operator $A^*_d(r,s)$ is a smooth function on  $[0,1]^2$ and implies that for any $s\in [0,1]$, both $A_d^*(\cdot, s)$  and $A_d^*(s,\cdot)$ are functions in $\h$~(see \Cref{lemma:bound of A 1}). This ensures that the reproducing property holds for both arguments of $A_d^*(r,s)$ and therefore allows us to estimate $A_d^*$ from discrete measurements of $\{X_t\}_{t=1}^T$. \Cref{assume:regularity of FARD}{\bf a} is similar to the commonly used assumption in functional principle component analysis~(FPCA) literature that the covariance operator (and therefore the transition operators) of the FAR process can be well approximated by a finite number of eigenfunctions.



\Cref{assume:regularity of FARD}{\bf b} is a commonly used condition in functional analysis literature. Since $E\epsilon _t(s)=0$, condition \eqref{eq:variance of the functional noise2} simply implies that the covariance operator $ \Sigma_\epsilon (s,r):=E (\epsilon_t(s)\epsilon_t (r ) )$ of the noise function is positive definite. Condition \eqref{eq:variance of the functional noise1} can be relaxed to a sub-Gaussian condition where $P( \| \epsilon_t\| _\h > \tau) \le \exp(-c \tau^2).$ In this case, all of our theoretical results still hold and the convergence rate will only be slower by a log factor of $T$.
 
\Cref{assume:X} imposes regularity conditions directly on the FAR($D$) process $\{X_t\}_{t=1}^T$.
\begin{assumption}\label{assume:X}
	The functional time series $\{X_t\}_{t=1}^T$ is stationary and there exist positive constants $C_X$ and $\kappa_X$ such that $P( \| X_t \|_\h \le C_X ) = 1$ and
	\begin{align}\label{eq:FAR restricted eigenvalue condition}
	E \left ( \int  \sum_{d=1}^D    v_d(s) X_{t-d}  (s) ds\right)^2  \ge \kappa_X \sum_{d=1}^D\|v_d\|_{\lt}^2  \ \text{ for all }   \{v_d \}_{d=1}^D\subset \h.
	\end{align}
\end{assumption}
\noindent \Cref{assume:X} is a high-level assumption made for explicitness. Note that when $D=1$, condition \eqref{eq:FAR restricted eigenvalue condition} reduces to $   \iint      v (s) \Sigma_X     (r,s  ) v (r )ds dr   \ge \kappa_X  \|v \|_{\lt}^2 \ \text{ for all } v \in \h, $ where $\Sigma_X (r,s ): =  E(X_t(r)X_t(s) )$. Therefore condition \eqref{eq:FAR restricted eigenvalue condition} can be thought of as the restricted eigenvalue condition for FAR($D$), which is a frequently  used condition in the high-dimensional time series literature, see for example \cite{basu2015regularized}.

\Cref{example:farD} shows that \Cref{assume:X} holds for a large family of FAR processes in RKHS.
\begin{proposition}  \label{example:farD} 
Given \Cref{assume:rkhs} and \Cref{assume:regularity of FARD}{\bf b}, for the FAR($D$) process in \Cref{definition:AR model}, if the transition operators $\{ A_d^* \}_{d=1}^D$ satisfy
\begin{align} \label{eq:example 2 condition}
\sup _{|z| \le 1 ,z \in \mathbb C   } \left \|  \sum_{d=1}^D z^d  A_d^*   \right   \|_{\h, \op}  =\gamma_A < 1,
\end{align} 
then there exists a unique stationary solution $\{X_t\}_{t=-\infty}^\infty$ to \eqref{eq:AR model} and there exists $C_X$ depending only on  $C_\epsilon$   and $\gamma_A$ such that  $P( \| X_t \|_\h \le C_X ) = 1$. In addition, if 
$\max_{1\leq d\leq D}\text{rank}(A_d^*)<\infty$,  then  \eqref{eq:FAR restricted eigenvalue condition} holds with $\kappa_X$ depending only on $\kappa_\epsilon$ and $\gamma_A$.
\end{proposition}

The stationarity result in \Cref{example:farD} is similar to that in Theorem 5.1 of \cite{bosq2000linear}, which gives the stationarity condition of an FAR process in $\lt.$ For $D=1$, condition \eqref{eq:example 2 condition} in \Cref{example:farD} reduces to $\|A_1^*\|_{\h, \op}<1$, which is intuitive and resembles the stationarity condition for the classical AR(1) process~(\cite{Brockwell1991}). For a general $D$, condition \eqref{eq:example 2 condition} resembles the stability condition of the VAR($D$) process~(\cite{Luetkepohl2005}).

\section{Estimation Methodology and Main Results}\label{sec:FARD}
In this section, we propose a penalized nuclear norm estimator for the transition operators of the FAR process in RKHS~(\Cref{definition:AR model}) and further study its consistency.

\Cref{subsec:noise free FARD} proposes the RKHS-based penalized estimation procedure for the transition operators $\{A_d^*\}_{d=1}^D$ with discrete realizations of $\{X_t\}_{t=1}^T$. \Cref{subsec:theory} establishes the consistency and the sharp convergence rate for the proposed estimator. \Cref{subsec:optimization} formulates the penalized estimation as a trace norm minimization problem and discusses its numerical implementation.

\subsection{Penalized estimation and Representer theorem} \label{subsec:noise free FARD}
As discussed before, in almost all real applications, instead of fully observed functional time series $\{X_t(s), s\in[0,1]\}_{t=1}^T$, the available data are typically discrete measurements $\{X_t(s_i)\}_{1\le t\le T,1\le i\le n}$, where $\{s_i\}_{i=1}^n$ denotes the collection of sampling points. Following the standard RKHS literature, we assume $ \{s_i\}_{i=1}^n$ to be a collection of random designs uniformly sampled from the domain $[0,1]$. Given $\{X_t(s_i)\}_{1\le t\le T,1\le i\le n}$, our interest is the $D$ unknown transition operators $\{A_d^*\}_{d=1}^D$, as the estimation of $\{A_d^*\}_{d=1}^D$ facilitates important inference tasks such as prediction.

We remark that for mathematical brevity, in this paper we only consider the case that $ \{s_i\}_{i=1}^n$ are uniformly sampled from the domain $[0,1]$. As a common feature in the RKHS literature (see e.g., \cite{koltchinskii2010sparsity}, \cite{raskutti2012minimax} and reference therein), all the results presented in the paper continue to hold under the more general setting where the random designs $ \{s_i\}_{i=1}^n$ are i.i.d sampled from a common continuous distribution on $[0,1]$ with density $p$ such that 
$\inf_{r\in [0,1]} p(r)>0$, if we redefine  $\lt $ with the inner product $ \lb f,g\rb_\lt := \int_0^1  f(s)g(s) p(s)ds$ and adjust the definition of $\h$ accordingly. 


Utilizing the finiteness of the nuclear norm of $\{A^*_d\}_{d=1}^D$ imposed in \Cref{assume:regularity of FARD}, we construct its estimator $\{\widehat A_d\}_{d=1}^D$ via a constrained nuclear norm optimization such that
\begin{align}\label{eq:approx A 2}
 \{ \widehat A_d\}_{d=1}^D   = \argmin_{\{ A_d \}_{d=1}^D \in  \mathcal C_{\bm\tau}}   \frac{1}{Tn} \sum_{t=D+1}^T    \sum_{i=1}^ n   \left  (X_{t} (s_i) -  \sum_{d=1}^D  \frac{1}{n }\sum_{j=1} ^n A_d  (s_i, s_j)  X_{t-d } (s_j) \right  )^2    
\end{align}   
where $\bm\tau=(\tau_1,\cdots,\tau_D)$ is the tuning parameter and $\mathcal C_{\bm\tau} :=\{(A_1,\cdots,A_D): A_d\in \mathcal C \text{ and } \| A_d\|_{\h, * }\le \tau_d,~ d=1,\cdots, D\}$ is the constraint space. We name $\{\widehat{A}_d\}_{d=1}^D$ in \eqref{eq:approx A 2} the penalized/constrained nuclear norm estimator for transition operators of FAR. (We use the term constrained and penalized nuclear norm estimator exchangeably due to the equivalence between constrained and penalized optimization. See \Cref{subsec:optimization} for more detail.)

To motivate the formulation of \eqref{eq:approx A 2}, consider the (ideal yet infeasible) scenario in which $\{ X_t\}_{t=1}^T$ are fully observed in the entire domain $[0,1]$, thus we can solve 
\begin{align*}
 \{ \widetilde A_d\}_{d=1}^D  = \argmin_{\{ A_d \}_{d=1}^D \in  \mathcal  C_{\bm\tau}  } \frac{1}{T }   \sum_{t=D+1}^T \int  \left  (X_{t} (r) - \sum_{d=1}^D \int   A_d (r,s)X_{t-d} (s) ds \right  )^2 dr. 
\end{align*}
 However, since only discrete measurements $ \{X_t(s_i)\}_{1\le t\le T,1\le i\le n}$ are observed, we instead solve \eqref{eq:approx A 2} where we use the integral approximation  
$ \int  A_d (s_i, r ) X_{t-d}  (r)dr \approx   \frac{1}{n }\sum_{j=1} ^n A_d(s_i, s_j)  X_{t-d}  (s_j) $.


Observe that \eqref{eq:approx A 2} is an optimization problem in an infinite dimensional Hilbert space due to the nature of functional time series. As discussed in Section \ref{sec:intro}, existing literature~\citep[e.g.][]{bosq2000linear,Didericksen2012,aue2015prediction} uses dimension reduction to bypass such difficulty. Instead, we derive the Representer theorem for the penalized estimator in RKHS, which reduces the infinite dimensional optimization problem in \eqref{eq:approx A 2} to finite dimension without dimension reduction.
\begin{proposition}[Representer theorem] \label{lemma:representer}
There exists a minimizer $\{ \widehat A_d \}_{d=1}^D$ of the constrained nuclear norm optimization \eqref{eq:approx A 2} such that for any $(r, s)\in [0,1] \times [0,1]$, 
\begin{align}
\label{eq:representer}
\widehat A_d (r, s) = \sum_{1\le i,j \le n}  \widehat  a_{d,ij }\mk (r, s_i)\mk (s, s_j), \text{ for } d=1,2,\cdots, D. 
\end{align}
\end{proposition} 

We note that while \Cref{lemma:representer} implies that a minimizer of \eqref{eq:approx A 2} lives in the space spanned by the reproducing kernel $\{\mk ( s_i,\cdot )\}_{i=1}^n$, it does not rule out the possibility that there is a different solution $  \{ \widehat A_d'\}_{d=1}^D $ that lives in a different subspace of higher dimensions. However, this does not affect the later theoretical analysis of consistency, which holds for any minimizer of \eqref{eq:approx A 2}. We remark that uniqueness of optimum is not a necessary condition for consistency in the RKHS literature.  See for instance, \cite{raskutti2012minimax} and \cite{koltchinskii2010sparsity}.

Given the estimated transition operators $\{\widehat A_d\}_{d=1}^D$, the one-step ahead prediction of $X_{T+1} $ can be readily calculated as
\begin{align}\label{eq:prediction}
\widehat X_{T+1}(r)=\sum_{d=1}^D\frac{1}{n}\sum_{j=1}^n\widehat{A}_d(r,s_j)X_{T+1-d}(s_j) \text{ for } r \in [0,1].
\end{align}

 
\subsection{Consistency}\label{subsec:theory}
In this section, we investigate the theoretical properties of the penalized nuclear norm estimator $\{\widehat A_d\}_{d=1}^D$ and establish its consistency.

Given Assumptions \ref{assume:rkhs}-\ref{assume:X}, \Cref{theorem:far rate} establishes the consistency result of $\{\widehat A_d\}_{d=1}^D$ and further provides the explicit convergence rate. We first introduce some notations before stating the theorem. Denote $\|f\|_n^2=\frac{1}{n}\sum_{i=1}^n f(s_i)^2$. We define
\begin{align}\label{eq:gamma 1}
&\gamma_n' : =  \inf\left\{ \gamma :  \left |  \int  f(s) ds  - \frac{1}{n} \sum_{i=1}^ n  f(s_i)   \right |    \le  \gamma   \|f\|_\lt    +    \gamma ^2   
\text{ for all } f \text{ such that }  \| f\|_{\h }\le 1     \right\},
\\ \label{eq:gamma 2}
&\gamma_n'' : = \inf\left \{ \gamma :   \|f \|_\lt ^2  \le  2 \|f \|_n^2 +      \gamma ^2   \text{ and } 
 \|f \|_n^2  \le  2 \|f \|_\lt ^2  +      \gamma ^2   \text{ for all } f \text{ such that } \| f\|_\h \le  1   \right\},
 \\ \label{eq:gamma}
 & \gamma_n = \max  \{ \gamma_n', \gamma_n''\}.
\end{align}
Intuitively speaking, $\gamma_n$ (uniformly) quantifies how well we know a function $f\in \h$~(in our case $f=X_t$) given its measurements on $n$ sample points: if the number of measurements $n$ increases, we have more knowledge of $X_t$ and $\gamma_n$ decreases. We note that $\gamma_n^2$ is the optimal mean squared error bound of estimating a single function in $\mathcal H$ given $n$ discrete measurements. See \cite{mendelson2002geometric},
\cite{koltchinskii2010sparsity} and \cite{raskutti2012minimax} for more details. For $\mathcal H  =  W^{ \alpha,2 }$, \Cref{corollary:explicit rate} of the Appendix establishes that $\gamma_n  = O_p(  n^{-\alpha /(2\alpha+1) } ) $, where $W^{ \alpha,2 }$ denotes the commonly used Sobolev space on $[0,1]$~(see Section \ref{subsec:RKHS_overview}).

In addition, we define 
\begin{align} 
  \nonumber 
&\delta '_T  : =       \inf \Bigg \{ \delta:      \left| 
\frac{1}{T  }\sum_{t=1}^T  
\left(   \sum_{d=1}^D    \int v_d(  r) X_{t-d} (r) dr   
 \right)  ^2    -
 E \left(   \sum_{d=1}^D    \int v_d(  r) X_{t-d} (r) dr   
 \right)  ^2  
\right|  \le     \delta   \sqrt { \sum_{d=1}^D \|v_d \|_\lt^2 }   
\\
&  \quad \quad     
\text{for all }  \{ v_d\}_{d=1}^D \text{ such that } \sup_{1\le d \le D} \|v_d\|_\h \le 1  \Bigg \},  \label{eq:delta 2}
\\
&\delta_T ''  :=  
\sup_{   1\le d \le D,  r,s \in [0,1] }  \left | \frac{1}{T} \sum_{t=1}^T    X_{t-d}(r ) \epsilon_{t} (s)    \right  |, \label{eq:delta 1} 
\\
&\delta_ T = \max \{ \delta_T', \delta_T'' \}. \label{eq:delta} 
\end{align} 
Intuitively speaking, $ \delta_T$ characterizes the convergence rate of $\{ \widetilde A_d\}_{d=1}^D $ when the functional time series $\{X_t\}_{t=1}^T$ is fully observed.  \Cref{corollary:delta explicit rate} of the Appendix shows that, for $\h=W^{\alpha,2}$, there exists constants $c'_w, C'_w$ such that  
\begin{align*}  
   P \left(  \delta_T '  \ge    C _w'   T ^{  \frac{-\alpha}{2\alpha +1 } }   \right ) \le 2T^2\exp \left (-c_w'   T ^{\frac{1}{2\alpha + 1}}     \right  )   \quad 
 \text{and}    \quad    
   P \left(  \delta_T'' \ge 3 C_XC_\epsilon   \sqrt { \frac{\log(T )}{ T}}    \right ) \le   T^{-3} . 
\end{align*}

We now state the main theoretical result of the paper, which quantifies the convergence rate of the penalized estimators  $\{\widehat A_d \}_{d=1}^D$ through the $\lt$ norm. The $\lt$ norm of any bivariate function  $A(r,s)$  is defined as $\| A\|_\lt ^2 : = \iint A^2(r,s) dr ds.$
\begin{theorem} \label{theorem:far rate}
Suppose Assumptions \ref{assume:rkhs}-\ref{assume:X} hold. Let $\{ \widehat A_d\}_{d=1}^D$ be the solution of \eqref{eq:approx A 2}. If $\kappa_X\ge 64 D\gamma_n^2$ and the tuning parameter $\bm\tau$ satisfies $ \|A_d^*\|_{\h,*}\leq \tau_d <C_A$, $d=1,\cdots, D$ for some constant $C_A$, then we have
\begin{align}\label{eq:convergent rate in far}
\sum_{d=1}^D  \|\widehat A_d -A^*_d  \|_\lt ^2 \le  C_1    \left (  \gamma_n^2 + \delta_T^2   \right)  ,
\end{align}
for some constant $C_1$ independent of $n$ and $T$.
\end{theorem} 

The bound in \Cref{theorem:far rate} has two components: $\gamma_n^2$ quantifies how well we can estimate a single function in $\mathcal H$ based on $n$ discrete measurements and $\delta_T^2$ is the rate of estimating the transition operators given fully observed $\{X_t\}_{t=1}^T$. Note that unlike existing FAR literature, the consistency result in \Cref{theorem:far rate} does not require the (unrealistic) assumption of fully observed functional time series. To our best knowledge, this is the first result in the FAR literature providing theoretical guarantees for the estimation of transition operators based on discrete measurements.

An immediate result of \Cref{theorem:far rate} is the explicit convergence rate of the penalized estimator $\{\widehat A_d\}_{d=1}^D$ for FAR($D$) in the Sobolev space $W^{\alpha,2}$, which is given in the following \Cref{corollary:farD}.
\begin{corollary}\label{corollary:farD}
Suppose the conditions in \Cref{theorem:far rate} hold. For $\h=W^{\alpha,2}$, with probability at least $1-1/n^4 -1/T^3 -2T^2\exp \left (-c_w'   T ^{\frac{1}{2\alpha + 1}}     \right  )$, it holds that
	\begin{align*}
	\sum_{d=1}^D  \|\widehat A_d -A^*_d  \|_\lt^2 \le  C_1'    \left (n^{\frac{-2\alpha}{2\alpha+ 1}} + T^{\frac{-2\alpha}{2\alpha+ 1}}   \right),
	\end{align*}
for some constants $c_w', C_1'$ independent of $n$ and $T$.
\end{corollary}
As discussed in Section \ref{sec:intro}, the consistency result of the dimension reduction based estimation methods~\citep[e.g.][]{bosq2000linear,Kargin2008} typically require the number of basis $p$ grows with the sample size $T$ at a rate that implicitly depends on intricate interrelation of eigenvalues and spectral gaps of the true covariance operator of the functional time series $\{X_t\}$, making the derived convergence rate rather opaque and case-specific. In contrast, our RKHS-based estimation method does not require dimension reduction, making the convergence rate in \Cref{corollary:farD} explicit as there is no dimension reduction incurred errors.

We further provide a simple argument to show that the error bound given in \Cref{theorem:far rate} is intuitive.  Note that estimating  $A_d^*$ is harder than estimating a single function in $\mathcal H$.  Thus  the consistency rate is lower bounded by
 $ \gamma_n^2,$
since as mentioned before, $\gamma_n^2$ is the well known optimal rate of estimating a single function from its $n$ discrete realizations in RKHS. We also note that FAR is an extension and generalization  to the Function to Function Regression (FFR) model. When the functions are fully observed,  the optimal rate  of excess risk in the    FFR setting is $\delta_T^2$ (see e.g. \cite{sun2018optimal}).
Based on the above discussion,   the error bound  we  established in \Cref{theorem:far rate} is sharp.

The estimation error bound in \Cref{theorem:far rate} naturally implies an error bound on the one-step ahead prediction given in \eqref{eq:prediction}. \Cref{prop:prediction rate} quantifies the prediction risk of $X_{T+1}$ given $\{X_t\}_{t=1}^T$.
\begin{proposition} \label{prop:prediction rate}
Let $\h=W^{\alpha,2}$ and $ \widehat X_{T+1}(r)   $ be defined as in \eqref{eq:prediction}. Define the oracle one-step ahead prediction of $X_{T+1}(r)$ as $E(X_{T+1}(r) \vert \{X_t\}_{t=1}^T)=\int\sum_{d=1}^D A^*_{d}(r,s)X_{T+1-d}(s)ds$. With probability at least $1-1/n^4 -1/T^3 -2T^2\exp \left (-c_w'   T ^{\frac{1}{2\alpha + 1}}     \right  )$, it holds that
\begin{align*}
&\| E(X_{T+1 } \vert \{X_t\}_{t=1}^T  ) -  \widehat X_{T+1}  \|_{\lt }^2  \le C_1'' \left (n^{\frac{-2\alpha}{2\alpha+ 1}} + T^{\frac{-2\alpha}{2\alpha+ 1}}   \right ),\\
&\frac{1}{n}\sum_{j=1}^n \left(E(X_{T+1}(s_j) \vert \{X_t\}_{t=1}^T) - \widehat X_{T+1}(s_j)\right)^2\le C_1'' \left (n^{\frac{-2\alpha}{2\alpha+ 1}} + T^{\frac{-2\alpha}{2\alpha+ 1}}   \right ),
\end{align*}
for some constants $c_w', C_1''$ independent of $n$ and $T$.
\end{proposition}

\subsection{Optimization via Accelerated Gradient Method} \label{subsec:optimization}
In this section, we discuss the numerical implementation of the proposed RKHS-based penalized estimator by reformulating the constrained optimization in \eqref{eq:approx A 2} into a standard trace norm minimization problem, which is well-studied in the machine learning literature~(\cite{Ji2009}).

We first introduce some notations. Denote the estimator $A_d (r, s) = \sum_{1\le i,j \le n} a_{d,ij }\mk (r, s_i)\mk (s, s_j)$, where $a_{d,ij}$s are the coefficients to be estimated. Define the coefficient matrix $R_{d}\in\mathbb{R}^{n\times n}$ with $R_{d,ij}=a_{d,ij}$. Define the kernel vector $k_i=(\mk(s_1,s_i),\mk(s_2,s_i),\cdots,\mk(s_n,s_i))^\top$ and the kernel matrix $K=[k_1,k_2,\cdots,k_n]$. Note that the kernel matrix $K$ is symmetric such that $K=K^\top$. Denote the observation of the functional time series at time $t$ as $X_t=(X_t(s_1),X_t(s_2),\cdots, X_t(s_n))^\top$. Define the observation matrix $X=[X_T,X_{T-1},\cdots, X_{D+1}]$ and the lagged observation matrix $X^{(d)} =[X_{T-d},X_{T-d-1},\cdots, X_{D+1-d}]$ for $d=1,\cdots, D.$ 

Using the well-known equivalence between constrained and penalized optimization (see \cite{Hastie2009}), we can reformulate \eqref{eq:approx A 2} into a penalized nuclear norm optimization such that
\begin{align*}
\{ \widehat A_d\}_{d=1}^D   = \argmin \sum_{t=D+1}^T    \sum_{i=1}^ n   \left  (X_{t} (s_i) -  \sum_{d=1}^D  \frac{1}{n }\sum_{j=1} ^n A_d  (s_i, s_j)  X_{t-d } (s_j) \right  )^2   + \sum_{d=1}^D \lambda_d \|A_d\|_{\h,*},
\end{align*}
where $(\lambda_1,\cdots,\lambda_D)$ is the tuning parameter. With simple linear algebra, we can rewrite the penalized optimization as
\begin{align}\label{eq:approximate A in AR_v1}
&\min_{R_1,\cdots,R_D}\sum_{t=D+1}^T \left(X_t-\frac{1}{n}\sum_{d=1}^DK^\top R_d K X_{t-d} \right)^\top  \left(X_t-\frac{1}{n}\sum_{d=1}^DK^\top R_d K X_{t-d} \right) + \sum_{d=1}^D\lambda_d \|A_d\|_{\h,*}\nonumber\\
=&\min_{R_1,\cdots,R_D} \left\|X-\frac{1}{n}\sum_{d=1}^DK R_d K X^{(d)} \right\|_F^2 + \sum_{d=1}^D\lambda_d \|A_d\|_{\h,*}.
\end{align}


We now write the nuclear norm $\|A_d\|_{\h,*}$ as a function of $R_d$. By the Representer theorem, $A_d(r,s)=\sum_{i,j}a_{d,ij}\mk(r,s_i)\mk(s,s_j)$, thus the adjoint operator $A_d^\top(r,s)=A_d(s,r)$. Define $k(s)=(\mk(s,s_1), \mk(s,s_2), \cdots, \mk(s,s_n))^\top$, we have $A_d(r,s)=k(r)^\top R_d k(s)$ and $\lb k(s), k(s)^\top\rb_\h=K$. Define $u(s)=k(s)^\top b$, where $b=(b_1,b_2,\cdots,b_n)^\top$. To calculate $\|A_d\|_{\h,*}$, note that
\begin{align*}
A_d^\top A_d[u](s)&=\lb A_d^\top (s,r), A_d[u](r) \rb_\h=\lb A_d(r,s), \lb A_d(r,s), u(s)\rb_\h \rb_\h\\
&=\lb k(r)^\top R_d k(s), \lb k(r)^\top R_d k(s), k(s)^\top b \rb_\h  \rb_\h=k(s)^\top R_d^\top KR_dKb.
\end{align*}
In other words, the eigenvalues of the operator $A_d^\top A$ correspond to the eigenvalues of the matrix $R_d^\top KR_dK$. Thus, \eqref{eq:approximate A in AR_v1} can be further written as
\begin{align}\label{eq:trace_norm1}
&\min_{R_1,\cdots,R_D} \left\|X-\frac{1}{n}\sum_{d=1}^DK R_d K X^{(d)} \right\|_F^2 + \sum_{d=1}^D\lambda_d \cdot \text{trace}((R_d^\top KR_dK)^{\frac{1}{2}}) \nonumber\\
=&\min_{R_1,\cdots,R_D} \left\|X-\frac{1}{n}\sum_{d=1}^DK R_d K X^{(d)} \right\|_F^2 + \sum_{d=1}^D\lambda_d \|K^{\frac{1}{2}}R_dK^{\frac{1}{2}}\|_{*} \nonumber\\
=&\min_{W_1,\cdots,W_D} \left\|X-\sum_{d=1}^D\K_d W_d Z_d \right\|_F^2 + \sum_{d=1}^D \|W_d\|_{*}
\end{align}
where $W_d=\lambda_dK^{\frac{1}{2}}R_dK^{\frac{1}{2}}$, $\K_d=\frac{1}{\lambda_d}K^{\frac{1}{2}}$, $Z_d=\frac{1}{n}K^{\frac{1}{2}}X^{(d)}$ and the first equality comes from the fact that $R_d^\top K R_dK$ and $K^{1/2}R_d^\top KR_dK^{1/2}$ share the same eigenvalues for $d=1,\cdots, D$.

Define $\K=[\K_1, \cdots, \K_D]$, $Z=\begin{bmatrix}
Z_1\\ \vdots \\ Z_D
\end{bmatrix}$ and $W=\begin{bmatrix}
W_1 &  & \\
& \ddots & \\
&  & W_D
\end{bmatrix}$, the optimization in \eqref{eq:trace_norm1} can be further written as
\begin{align}\label{eq:trace_norm2}
\argmin_{W} \left\|X-\K W Z \right\|_F^2 +  \|W\|_*,
\end{align}
where $g(W)=\left\|X-\K W Z \right\|_F^2$ is a convex function of the block diagonal matrix $W$ and $\|W\|_*$ is its trace norm. Note that \eqref{eq:trace_norm2} is a convex function of $W_1,\cdots, W_d$ with a unique global minimizer.

Thus, we formulate the constrained nuclear norm optimization in \eqref{eq:approx A 2} into a standard trace norm minimization problem in the machine learning literature~(e.g. see \cite{Bach2008}, \cite{Candes2009}). In particular, given tuning parameters $\{\lambda_d\}_{d=1}^D$, $\eqref{eq:trace_norm2}$ can be readily solved by the Accelerated Gradient Method~(AGM) in \cite{Ji2009}. Due to the block diagonal structure of $W$, AGM can be performed in a component-wise fashion where the gradient update of the optimization is carried out for each $W_1,\cdots, W_D$ separately. The implementation details of the AGM algorithm can be found in Section \ref{sec:AGM algorithm} of the Appendix.

Given $(\widehat{W}_1, \cdots, \widehat{W}_D)$, the estimated transition operators $\{\widehat A_d\}_{d=1}^D$ can be recovered by
\begin{align*}
\widehat{A}_d(r,s)=k(r)^\top \widehat{R}_d k(s)=\frac{1}{\lambda_d}k(r)^\top K^{-\frac{1}{2}}\widehat{W}_d K^{-\frac{1}{2}} k(s), \text{ for } d = 1,\cdots, D.
\end{align*}
Plugging into \eqref{eq:prediction}, the one-step ahead prediction of $X_{T+1}$ is then
\begin{align*}
\widehat X_{T+1}(r)=\frac{1}{n}\sum_{d=1}^D\frac{1}{\lambda_d}k(r)^\top K^{-\frac{1}{2}}\widehat{W}_d K^{\frac{1}{2}}X_{T+1-d}, \text{ for } r \in [0,1].
\end{align*}

\section{Simulation Studies}\label{sec:simulation}

In this section, we conduct simulation studies to investigate the estimation and prediction performance of the proposed penalized nuclear norm estimator and compare it with the standard transition operator estimation approach in \cite{bosq2000linear} and the state-of-the art functional time series prediction method in \cite{aue2015prediction}.


\subsection{Basic simulation setting}
\textbf{Data generating process}: We first define an FAR($D$) process, borrowed from the simulation setting in \cite{aue2015prediction}, that is used in the simulation study. For $d=1,2,\cdots, D$, we assume the $d$th transition operator $A_d(r,s)$ is of rank $q_d$ and is generated by $q_{d}$ basis functions $\{u_i(s)\}_{i=1}^{q_d}$ such that

\centerline{$A_d(r,s)=\sum_{i,j=1}^{q_d} \lambda_{d,ij} u_i(r)u_j(s),$}
\noindent where $\{u_i(s)\}_{i=1}^{q_d}$ consists of orthonormal basis of $\lt[0,1]$ that will be specified later. Define matrix $\Lambda_d$ such that $\Lambda_{d,ij}=\lambda_{d,ij}$ and define $\bu_{q_d}(s)=(u_1(s),u_2(s),\cdots,u_{q_d}(s))^\top$. We have $A_d(r,s)=\bu_{q_d}(r)^\top \Lambda_d \bu_{q_d}(s).$ We further set the noise function $\epsilon_t$ to be of finite rank $q_\epsilon$ such that $\epsilon_t(s)=\sum_{i=1}^{q_\epsilon} z_{ti}u_i(s)$, where $z_{ti}\overset{i.i.d.}{\sim} U(-a_i, a_i)$ or $z_{ti}\overset{i.i.d.}{\sim} N(0, \sigma_i^2)$. 

Without loss of generality, we set $q_{1}=q_{2}=\cdots=q_D=q_\epsilon=q$ for simplicity. Thus, the FAR($D$) process $\{X_t(s)\}_{t=1}^T$ resides in a finite dimensional subspace spanned by the orthonormal basis $\{u_i(s)\}_{i=1}^{q}$. Denote $X_t(r)=\sum_{i=1}^qx_{ti}u_i(r)$ where $x_{ti}=\int X_t(r)u_i(r)dr$, and denote $x_t=(x_{t1},\cdots,x_{tq})^\top$ and $z_t=(z_{t1},\cdots,z_{tq})^\top$. We have
\begin{align*}
X_t(r)=&\sum_{d=1}^D\int A_d(r,s)X_{t-d}(s)ds + \epsilon_{t}(r)=\sum_{d=1}^D\int \bu_{q}(r)^\top \Lambda_d \bu_{q}(s) X_{t-d}(s)ds + z_{t}^\top \bu_q(r)\\
=&\sum_{d=1}^D\int \bu_q(r)^\top \Lambda_d \bu_q(s) \bu_q(s)^\top x_{t-d} ds + z_{t}^\top \bu_q(r)= \bu_q(r)^\top\left( \sum_{d=1}^D\Lambda_d  x_{t-d} + z_{t}\right).
\end{align*}
This leads to $x_t= \sum_{d=1}^D \Lambda_d x_{t-d} + z_t.$ Thus, the FAR($D$) process can be exactly simulated via a VAR($D$) process. Following the simulation setting in \cite{Yuan2010} and \cite{sun2018optimal}, we set $u_i(s)=1$ if $i=1$ and $u_i(s)=\sqrt{2}\cos((i-1)\pi s)$ for $i=2,\cdots, q$.

Given the transition operators $A_1,\cdots, A_D$~(i.e. $\Lambda_1,\cdots, \Lambda_D$) and the distribution of noise $z_t$, the true FAR($D$) process $\{X_t(s), s\in[0,1]\}_{t=1}^T$ can be simulated and discrete measurements of the functional time series are taken at the sampling points $\{s_i\}_{i=1}^n$. For simplicity, we set $\{s_i\}_{i=1}^n$ to be the $n$ equal-spaced points in $[0,1]$, which resembles the typical sampling scheme of functional time series in real data applications. Simulation based on uniformly distributed $\{s_i\}_{i=1}^n$ gives consistent conclusions.

\textbf{Evaluation criteria}: We evaluate the performance of a method via (a).\ estimation error of $\widehat A_1,\widehat A_2,\cdots, \widehat A_D$ and (b).\ prediction error of the estimated FAR($D$) model.

Specifically, given sample size $(n,T)$, we simulate the observed functional time series $\{X_t(s_i), i=1,\cdots, n\}_{t=1}^{T+0.2T}$, which we then partition into training data $\{X_t(s_i), i=1,\cdots, n\}_{t=1}^T$ for estimation of $A_1,\cdots, A_D$ and test data $\{X_t(s_i), i=1,\cdots, n\}_{t=T+1}^{T+0.2T}$ for evaluation of prediction performance. Denote $\{\widehat{X}_t(s_i),i=1,\cdots,n\}_{t=T+1}^{T+0.2T}$ as the one-step ahead prediction given by the estimated FAR($D$) model. We define
\begin{align}
&\text{MISE}(\widehat{A}_d, A_d)=\int_{[0,1]}\int_{[0,1]}(A_d(r,s)-\widehat{A}_d(r,s))^2drds \bigg/ \int_{[0,1]}\int_{[0,1]}A_d(r,s)^2drds,\label{MISE}\\
&\text{PE}=\frac{1}{0.2nT}\sum_{t=T+1}^{T+0.2T}\sum_{i=1}^{n}(X_t(s_i)-\widehat{X}_t(s_i))^2,\label{PE}
\end{align}
where MISE~(mean integrated squared error) measures the estimation error and PE measures the prediction error. For reference purposes, we also calculate the oracle prediction error and the constant mean prediction error such that
\begin{align*}
&\text{Oracle PE}=\frac{1}{0.2nT}\sum_{t=T+1}^{T+0.2T}\sum_{i=1}^{n}(X_t(s_i)-\widetilde{X}_t(s_i))^2,\quad \text{Mean Zero PE}=\frac{1}{0.2nT}\sum_{t=T+1}^{T+0.2T}\sum_{i=1}^{n}(X_t(s_i)-0)^2,
\end{align*}
where $\widetilde{X}_t(s_i)=\sum_{d=1}^D\int A_d(s_i,r)X_{t-d}(r)dr$ is the (infeasible) oracle predictor for $X_t(s_i)$ and 0 is the constant mean prediction since $E(X_t(s))=0$ for $s\in[0,1]$. Evaluation based on other types of error measures for estimation and prediction error~(besides MISE and PE) gives consistent conclusions and thus is omitted.

\subsection{Estimation methods and implementation details}\label{subsec:estimation prediction}
For comparison, we implement two functional PCA~(FPCA) based estimation approach for FAR: (a).\ the standard estimator in \cite{bosq2000linear} and (b).\ the vector autoregressive based approach in \cite{aue2015prediction}. Both estimators make use of the FPCA conducted on the sample covariance operator $\widetilde{C}(s,r)=\frac{1}{T}\sum_{t=1}^TX_t(s)X_t(r)$ such that $\widetilde{C}(s,r)=\sum_{i=1}^\infty \hat{\lambda}_i \hat{f}_i(s)\hat{f}_i(r)$, where $(\hat{\lambda}_i,\hat{f}_i)$ is the eigenvalue-eigenfunction pair.

\textbf{Standard estimator in \cite{bosq2000linear}}~[Bosq]: The estimator in \cite{bosq2000linear} is designed for estimating the transition operator $A_1(s,r)$ of FAR(1) based on the Yule-Walker equation for FAR(1) such that $D(s,r)=E(X_t(s)X_{t-1}(r))=E(\int A_1(s,s')X_{t-1}(s')ds'X_{t-1}(r))=\int A_1(s,s')C(s',r)ds'$, where $C(s,r)=E(X_t(s)X_t(r))$ is the covariance operator and $D(s,r)=E(X_t(s)X_{t-1}(r))$ is the auto-covariance operator. 

The Yule-Walker equation is inverted via FPCA-based dimension reduction, where all quantities in the Yule-Walker equation are projected on the subspace spanned by the $p$ orthonormal eigenfunctions $\hat{f}(s)=(\hat{f}_1(s),\cdots, \hat{f}_p(s))^\top$ corresponding to the $p$ largest eigenvalues $(\hat{\lambda}_1,\cdots,\hat{\lambda}_p)$ of the sample covariance operator. Specifically, $C(s,r)$ is approximated by $\hat{C}(s,r)=\sum_{i=1}^p\hat{\lambda}_i\hat{f}_i(s)\hat{f}_i(r)$ and $D(s,r)$ is approximated by 
$$\hat{D}(s,r)=\frac{1}{T-1}\sum_{t=2}^T\sum_{i=1}^p\lb X_t, \hat{f}_i  \rb_{\lt} \hat{f}_i(s) \sum_{j=1}^p \lb X_{t-1},\hat{f}_j \rb_{\lt} \hat{f}_j(r).$$ The estimator of $A_1$ takes the form $\hat{A}_1(s,r)=\sum_{i=1}^{p}\sum_{ j=1}^p a_{ij}\hat{f}_i(s)\hat{f}_j(r)$. Denote $\hat{\Lambda}=\text{diag}(\hat{\lambda}_1,\cdots,\hat{\lambda}_p)$, $\hat{d}_t=(\lb X_{t},\hat{f}_1 \rb_{\lt},\cdots, \lb X_{t},\hat{f}_p \rb_{\lt})$ and let $R$ denote the coefficient matrix such that $R_{ij}=a_{ij}$. The Yule-Walker equation implies that $\frac{1}{T-1}\sum_{t=2}^T \hat{d}_t\hat{d}_{t-1}^\top=R\hat{\Lambda}$, and thus $R=\frac{1}{T-1}\sum_{t=2}^T \hat{d}_t\hat{d}_{t-1}^\top \hat{\Lambda}^{-1}$, which provides an estimator of the transition operator $A_1$. The number of functional principal components $p$ used in the projection is typically set as the smallest number of eigenvalues such that the explained variability of the sample covariance operator is over a high threshold $\tau$, say $\tau=$80\%. In the following, we refer to this estimator by Bosq.

Using the fact that an FAR($D$) process can be formulated into an FAR(1) process, the above argument naturally provides an estimator for the transition operators $A_1,\cdots, A_D$ of FAR($D$) with $D>1.$ We refer to \cite{bosq2000linear} for more details.

\textbf{Functional PCA-VAR estimator in \cite{aue2015prediction}}~[ANH]: The basic idea of \cite{aue2015prediction} is a canny combination of FPCA-based dimension reduction and the classical vector autoregressive~(VAR) model, designed for prediction of FAR processes. Specifically, the infinite dimensional functional time series $\{X_t\}_{t=1}^T$ is first projected to the $p$ eigenfunctions $\hat{f}(s)=(\hat{f}_1(s),\cdots, \hat{f}_p(s))^\top$ of the sample covariance operator. After projection, $X_t$ is represented by a $p$-dimensional functional principal score $x_t=(x_{t1},\cdots,x_{tp})^\top$ with $x_{ti}=\int X_t(s)\hat{f}_i(s)ds.$  A VAR($D$) model is then fitted on the $p$-dimensional time series $\{x_t\}_{t=1}^T$ such that $x_t={B}_1x_{t-1}+\cdots+{B}_Dx_{t-D} + \epsilon_t$. Denote the estimated coefficient matrices as $\hat{B}_1,\cdots, \hat{B}_D\in \mathbb{R}^{p\times p}$, the one-step ahead prediction of $X_{t}(s)$ is then $\widehat X_{t}(s)=\hat{f}(s)^\top \hat{x}_{t}=\hat{f}(s)^\top \sum_{d=1}^D \hat{B}_dx_{t-d}$.

Note that this implies $\hat{X}_t(s)=\hat{f}(s)^\top \hat{x}_t=\hat{f}(s)^\top \sum_{d=1}^D \hat{B}_dx_{t-d}=\hat{f}(s)^\top \sum_{d=1}^D \hat{B}_d\int \hat{f}(r)X_{t-d}(r)dr=\sum_{d=1}^D \int \hat{f}(s)^\top\hat{B}_d \hat{f}(r) X_{t-d}(r) dr$. Thus, the FPCA-based prediction algorithm in \cite{aue2015prediction} induces an estimator for the transition operators $\{A_d\}_{d=1}^D$ such that
\begin{align*}
\widehat{A}_d(s,r)=\hat{f}(s)^\top \hat{B}_d\hat{f}(r), \text{ for } d=1,\cdots, D.
\end{align*}
The fFPE criterion in \cite{aue2015prediction} is used to select the number of functional principal components $p$ for a given autoregressive order $D$. In the following, we refer to this estimator by ANH.

\textbf{Implementation of FPCA-based estimators (Bosq and ANH)}: For the implementation of Bosq and ANH, the functional time series is required to be fully observed over the entire interval $[0,1]$. However, under the current simulation setting, only discrete measurements $\{X_t(s_i),i=1,\cdots,n\}_{t=1}^T$ are available. Following \cite{aue2015prediction}, for each $t$, the function $X_t(s), s\in [0,1]$ is estimated using 10 cubic B-spline basis functions based on the discrete measurements $(X_t(s_1),\cdots, X_t(s_n))$. We also use 20 cubic B-spline basis functions for more flexibility (see more details later).

\textbf{Implementation of penalized nuclear norm estimator (RKHS)}: For the implementation of the proposed RKHS-based penalized nuclear norm estimator, we use the rescaled Bernoulli polynomial as the reproducing kernel $\mk$, such that

\centerline{$\mk(x,y)=1+k_1(x)k_1(y)+k_2(x)k_2(y)-k_4(x-y),$}
\noindent where $k_1(x)=x-0.5$, $k_2(x)=\frac{1}{2}(k_1^2(x)-\frac{1}{12})$ and $k_4(x)=\frac{1}{24}(k_1^4(x)-\frac{k_1^2(x)}{2}+\frac{7}{240})$ for $x\in[0,1]$, and $k_4(x-y)=k_4(|x-y|)$ for $x,y\in [0,1]$. Such $\mk$ is the reproducing kernel for $W^{2,2}$. See Chapter 2.3.3 of \cite{Gu2013} for more detail.

The accelerated gradient algorithm in \cite{Ji2009} is used to solve the trace norm minimization as discussed in Section \ref{subsec:optimization}, where the algorithm stops when the relative decrease of function value in \eqref{eq:trace_norm2} is less than $10^{-8}$. A standard 5-fold cross validation is used to select the tuning parameter $(\lambda_1,\cdots,\lambda_D)$. Based on $\{\widehat A_d\}_{d=1}^D$, the one-step ahead prediction of $X_{t}(s_i)$ for $t=T+1,\cdots,T+0.2T$ in the test data can be calculated via
$\widehat X_{t}(s_i)=\sum_{d=1}^D\frac{1}{n}\sum_{j=1}^n\widehat{A}_d(s_i,s_j)X_{t-d}(s_j)$ for $i=1,\cdots, n$ as in \eqref{eq:prediction}.

\subsection{Simulation result for FAR(1)}\label{subsec:simu_far1}
We first start with the simple case of FAR(1), where there is only one transition operator $A(r,s)=A_1(r,s)$. The simulation setting involves the transition matrix $\Lambda=\Lambda_1\in\mathbb{R}^{q\times q}$~(signal) and the noise range $a_{1:q}=(a_1,a_2,\cdots, a_q)$ or the noise variance $\sigma_{1:q}^2=(\sigma_1^2,\sigma_2^2,\cdots,\sigma_q^2)$ for $\{z_{ti}\}_{i=1}^q$~(driving noise). Denote $\sigma(\Lambda)$ as the leading singular value for a matrix $\Lambda$. We consider three different signal-noise settings:
\vspace{-0.2cm}
\begin{itemize}  \setlength\itemsep{-0.2em}
	\item Scenario A (Diag $\Lambda$): $\Lambda=\text{diag}(\kappa,\cdots, \kappa)$ and $z_{ti}\overset{i.i.d.}{\sim} U(-a, a)$ with $a=0.1$ for $i=1,\cdots, q.$
	\item Scenario B (Random $\Lambda$): A random matrix $\Lambda^*$ is first generated via $\Lambda^*_{ij}\overset{i.i.d.}{\sim} N(0,1)$ and we set $\Lambda=\kappa\cdot\Lambda^*/\sigma(\Lambda^*)$, and $z_{ti}\overset{i.i.d.}{\sim} U(-a, a)$ with $a=0.1$ for $i=1,\cdots,q.$
	\item Scenario C (ANH setting): (a) A random matrix $\Lambda^*$ is first generated via $\Lambda^*_{ij}\overset{ind.}{\sim} N(0,\sigma_i\sigma_j)$ and we set $\Lambda=\kappa\cdot\Lambda^*/\sigma(\Lambda^*)$, and $z_{ti}\overset{ind.}{\sim} N(0, \sigma_i^2)$ with $\sigma_{1:q}=(1:q)^{-1}$. (b) Same setting except $\sigma_{1:q}=1.2^{-(1:q)}$.
\end{itemize}
Scenario C is borrowed from \cite{aue2015prediction}. \textcolor{black}{Within each scenario, the intrinsic dimension of FAR(1) is controlled by the dimenion $q$ of the transition matrix $\Lambda$ and the signal strength is controlled by the spectral norm $\kappa$ of $\Lambda$, where a higher $q$ implies a more complex FAR process and a larger $\kappa$ gives a stronger signal.}

\textbf{Signal strength for Scenarios A-C}: \textcolor{black}{Given the same $(q,\kappa)$}, we further discuss the signal strength of the three scenarios from the viewpoint of VAR processes. The main difference between Scenarios A, B and Scenario C is that for Scenarios A and B, the variance of the noise stays at a constant level $a$ across $\{z_{ti}\}_{i=1}^q$, while for Scenario C, the variance of the noise $\sigma_{1:q}$ decays with $q$ and the decay rate is faster in Scenario C(a) than in Scenario C(b). Note that unlike the classical regression setting, the noise $\{z_{ti}\}_{i=1}^q$ of an autoregressive process is not noise in the traditional sense but rather the driving force of the process. Indeed, variation in $\{z_{ti}\}_{i=1}^q$ helps reveal more information about the transition matrix $\Lambda$ and leads to stronger signals. Thus, compared to Scenarios A and B, Scenario C has weaker signals with Scenario C(a) having the lowest signal strength, and intuitively Scenario C can be more well approximated by a lower-dimensional process. Between Scenarios A and B, note that the transition matrix in Scenario A has overall larger and non-decaying singular values, making Scenario A the strongest signal scenario and most difficult to be approximated by a low-dimensional process. 

To summarize, in terms of signal strength, we have Scenario A $>$ B $>$ C(b) $>$ C(a) given the same intrinsic dimension $q$ and spectral norm $\kappa$. This indeed has implications on the numerical results (see more details later). For more discussion of the signal-to-noise ratio for VAR processes, we refer to \cite{Luetkepohl2005}. Additionally, we remark that the numerical performance is insensitive to the distribution of $z_{ti}$ (uniform or normal distribution).


For Bosq and ANH, the function $X_t(s)$ is first estimated using 10 cubic B-spline basis functions. For ANH, we use 20 cubic B-splines when $q=21$ for more flexibility. The performance of Bosq worsens when using 20 cubic B-splines, thus we always use 10 cubic B-splines for Bosq. With the FAR order fixed at $D=1$, the threshold $\tau$ is set at 80\% to select the number of FPCs $p$ for Bosq and the fFPE criterion is used to select the number of FPCs $p$ for ANH. For RKHS, we use 5-fold cross validation to select the tuning parameter $\lambda_1$.

\textbf{Numerical result for FAR(1)}: For Scenarios A and B, we consider three sample sizes: $(1) q=6, n=20, T=100, (2) q=12, n=20, T=400, (3) q=21, n=40, T=400$. For Scenario C, we consider $q=21, n=40, T=400$. As for the signal level, we vary the spectral norm of $\Lambda$ by $\kappa=0.2,0.5,0.8.$ For each simulation setting, i.e. different combination of Scenario A-C and $(q, n, T, \kappa)$, we conduct 100 experiments. Note that the transition matrix $\Lambda$ is randomly generated for each experiment under Scenario B and C.

We summarize the numerical performance of Bosq, ANH and RKHS in Table \ref{tab:FAR1}, where we report the mean MISE (MISE$_{avg}$) and mean PE (PE$_{avg}$) across the 100 experiments~(the conclusion based on median MISE and median PE is consistent and thus omitted). For each experiment, we also calculate the percentage improvement of prediction by RKHS over ANH via Ratio= (PE(ANH) / PE(RKHS)$-1)\times 100\%$. A positive ratio indicates improvement by RKHS. We report the mean ratio (denoted by R$_{avg}$) across the 100 experiments. In addition, we report the percentage of experiments~(denoted by R$_w$) where RKHS achieves a lower PE than ANH. Note that we compare ANH and RKHS as Bosq in general gives the least favorable performance. We further give the boxplot of PE in Figure \ref{fig:FAR1_k5_PE} under signal strength $\kappa=0.5$. The boxplots of PE under $\kappa=0.2, 0.8$ can be found in the Appendix.


Overall, RKHS gives the smallest estimation error~(measured by MISE defined in \eqref{MISE}) and prediction error~(measured by PE defined in \eqref{PE}) while ANH offers the second best performance. 
\textcolor{black}{In general, within each scenario, the improvement of RKHS over comparison methods increases with a higher dimension $q$ and a stronger signal $\kappa$, while for the same $(q,\kappa)$, RKHS yields the most improvement in Scenario A, followed by Scenarios B, C(b), and C(a).} We provide some intuition  as follows. When the intrinsic dimension of the FAR process is low and the signal is weak, the FPCA-based dimension reduction~(which is a hard thresholding method) does not induce much bias and achieves a good bias-variance trade-off. However, when the signal is strong and the intrinsic dimension of the process is high, information lost in the dimension reduction is non-negligible, and thus the proposed RKHS-based regularization outperforms FPCA-based methods as it   corresponds to a soft thresholding method. Note that the improvement in MISE may not lead to the same scale of improvement in PE, an observation also seen in \cite{Didericksen2012}.

Compared to Scenario A, the estimation performance~(MISE) of all methods deteriorate under Scenario B and C, due to the more complex nature of the transition operator and the decaying signal strength. Note that MISE of ANH is noticeably large under $q=21$ for Scenario B. One possible reason is the numerical instability caused by estimation of a large VAR model~(A VAR(20) model, i.e.\ $p=20$ FPC, is selected by ANH 33 out of 100 times.\footnote{Under $q=6$, ANH occasionally results in extremely large estimation and prediction error due to numerical instability of the VAR estimation if more than 6 FPC are selected. We exclude those cases from the numerical result.}). Under Scenario C, Bosq gives the smallest MISE while RKHS still gives the smallest PE, however, MISE is not very meaningful as all methods are unable to recover the transition operators accurately under Scenario C.

As mentioned above, the reported result by ANH under $q=21$ is based on 20 cubic B-splines for more flexibility. For illustration, Figure \ref{fig:FAR1_k5_PE}(c)(f) additionally plots the PE of ANH based on 10 cubic B-splines, which is noticeably worse than the one based on 20 cubic B-splines. This indicates that the smoothing step can substantially affect the performance of methods that rely on fully observed functional time series, though the smoothing error is typically ignored in theoretical results.

\begin{table}[H]
	\centering\small
	\begin{tabular}{llrrrr|rrrr}
		\hline
		\hline & & \multicolumn{4}{c|}{ Scenario A:  $q= 6 , n= 20 , T= 100 $} & \multicolumn{4}{c}{ Scenario B:  $q= 6 , n= 20 , T= 100 $} \\
		& Method & MISE$_{avg}$ & PE$_{avg}$ & R$_{avg}$(\%) & R$_{w}$(\%) & MISE$_{avg}$ &PE$_{avg}$ & R$_{avg}$(\%) & R$_{w}$(\%)\\
		\hline
		$\kappa=0.2$ & RKHS &  \bf 0.894 & \bf 2.182 & 0.20 & 52 & \bf 1.119 & \bf 2.130 & 0.79 & 64 \\ 
		& ANH & 0.960 & 2.185 &  &  & 1.617 & 2.146 &  &  \\
		& Bosq & 1.210 & 2.204 &  &  & 3.543 & 2.223 &  &  \\ 
		$\kappa=0.5$ & RKHS & \bf 0.222 & \bf 2.260 & 3.60 & 76 & \bf 0.639 & \bf 2.243 & 0.74 & 57 \\ 
		& ANH & 0.323 & 2.343 &  &  & 0.742 & 2.256 &  &  \\ 
		& Bosq & 0.308 & 2.322 &  &  & 0.739 & 2.263 &  &  \\ 
		$\kappa=0.8$ & RKHS & 0.078 & \bf 2.340 & 0.42 & 59 & \bf 0.241 & \bf 2.233 & 1.27 & 73 \\ 
		& ANH & \bf 0.065 & 2.351 &  &  & 0.270 & 2.261 &  &  \\
		& Bosq & 0.341 & 3.461 &  &  & 0.394 & 2.320 &  &  \\
		\hline
		\hline & & \multicolumn{4}{c|}{ Scenario A:  $q= 12 , n= 20 , T= 400 $} & \multicolumn{4}{c}{ Scenario B:  $q= 12 , n= 20 , T= 400 $} \\
		& Method & MISE$_{avg}$ & PE$_{avg}$ & R$_{avg}$(\%) & R$_{w}$(\%) & MISE$_{avg}$ &PE$_{avg}$ & R$_{avg}$(\%) & R$_{w}$(\%)\\
		\hline
		$\kappa=0.2$ & RKHS & \bf 0.563 & \bf 4.288 & 0.40 & 59 & \bf 1.017 & \bf 4.241 & 0.10 & 54 \\ 
		& ANH & 0.716 & 4.305 &  &  & 1.081 & 4.245 &  &  \\ 
		& Bosq & 0.668 & 4.298 &  &  & 1.525 & 4.270 &  &  \\
		$\kappa=0.5$ & RKHS & \bf 0.097 & \bf 4.290 & 6.90 & 100 &  \bf 0.299 & \bf 4.296 & 2.42 & 98 \\ 
		& ANH & 0.419 & 4.586 &  &  & 1.078 & 4.400 &  &  \\ 
		& Bosq & 0.450 & 4.829 &  &  & 0.775 & 4.459 &  &  \\ 
		$\kappa=0.8$ & RKHS & \bf 0.044 & \bf4.323 & 38.53 & 100 & \bf 0.122 & \bf 4.297 & 7.94 & 100 \\ 
		& ANH & 0.341 & 5.987 &  &  & 0.840 & 4.639 &  &  \\ 
		& Bosq & 0.464 & 7.762 &  &  & 0.679 & 4.865 &  &  \\ 
		\hline
		\hline & & \multicolumn{4}{c|}{ Scenario A:  $q= 21 , n= 40 , T= 400 $} & \multicolumn{4}{c}{ Scenario B:  $q= 21 , n= 40 , T= 400 $} \\
		& Method & MISE$_{avg}$ & PE$_{avg}$ & R$_{avg}$(\%) & R$_{w}$(\%) & MISE$_{avg}$ &PE$_{avg}$ & R$_{avg}$(\%) & R$_{w}$(\%)\\
		\hline
		$\kappa=0.2$ & RKHS & \bf 0.614 & \bf 7.339 & 0.75 & 78 & \bf 1.016 & \bf 7.256 & 0.04 & 50 \\ 
		& ANH & 0.829 & 7.394 &  &  & 1.062 & 7.259 &  &  \\ 
		& Bosq & 0.806 & 7.380 &  &  & 1.528 & 7.293 &  &  \\ 
		$\kappa=0.5$ & RKHS & \bf 0.092 & \bf 7.350 & 6.57 & 100 & \bf 0.344 & \bf 7.359 & 3.34 & 100 \\ 
		& ANH & 0.304 & 7.833 &  &  & 2.914 & 7.605 &  &  \\ 
		& Bosq & 0.665 & 8.691 &  &  & 0.960 & 7.694 &  &  \\ 
		$\kappa=0.8$ & RKHS & \bf 0.032 & \bf 7.424 & 23.94 & 100 & \bf 0.129 & \bf 7.365 & 5.39 & 100 \\ 
		& ANH & 0.271 & 9.200 &  &  & 3.985 & 7.762 &  &  \\ 
		& Bosq & 0.674 & 15.406 &  &  & 0.899 & 8.594 &  &  \\ 
		\hline
		\hline & & \multicolumn{4}{c|}{ Scenario C(a):  $q= 21 , n= 40 , T= 400 $} & \multicolumn{4}{c}{ Scenario C(b):  $q= 21 , n= 40 , T= 400 $} \\
		& Method & MISE$_{avg}$ & PE$_{avg}$ & R$_{avg}$(\%) & R$_{w}$(\%) & MISE$_{avg}$ &PE$_{avg}$ & R$_{avg}$(\%) & R$_{w}$(\%)\\
		\hline
		$\kappa=0.2$ & RKHS & 1.193 & 1.657 & 0.08 & 55 & \bf 1.042 & \bf 2.349 & 0.04 & 60 \\ 
		& ANH & 1.696 & 1.658 &  &  & 1.068 & 2.350 &  &  \\ 
		& Bosq & \bf 1.061 & \bf 1.655 &  &  & 1.182 & 2.356 &  &  \\ 
		$\kappa=0.5$ & RKHS & 1.669 & \bf 1.696 & 0.42 & 56 & 1.736 & \bf 2.433 & 0.07 & 64 \\ 
		& ANH & 1.788 & 1.703 &  &  & 1.827 & 2.435 &  &  \\ 
		& Bosq & \bf 0.996 & 1.746 &  &  & \bf 0.985 & 2.480 &  &  \\ 
		$\kappa=0.8$ & RKHS & 1.880 & 1.759 & -0.54 & 38 & 1.618 & \bf 2.481 & 0.58 & 63 \\ 
		& ANH & 1.740 & \bf 1.749 &  &  & 1.308 & 2.495 &  &  \\ 
		& Bosq & \bf 0.975 & 1.929 &  &  & \bf 0.945 & 2.745 &  &  \\ 
		\hline\hline
	\end{tabular}
	\caption{Numerical performance of various methods for FAR(1) processes. Methods considered are RKHS~(this paper), ANH~\citep{aue2015prediction}, and Bosq~\citep{bosq2000linear}. Bold font indicates the best performance, where the proposed RKHS method is generally the best performer in Scenarios A and B. (PE$_{avg}$ is multiplied by 100 in scale for Scenarios A and B, but not for Scenario C.)}
	\label{tab:FAR1}
\end{table}

\begin{figure}[H]
	\begin{subfigure}{0.32\textwidth}
		\includegraphics[angle=270, width=1.25\textwidth]{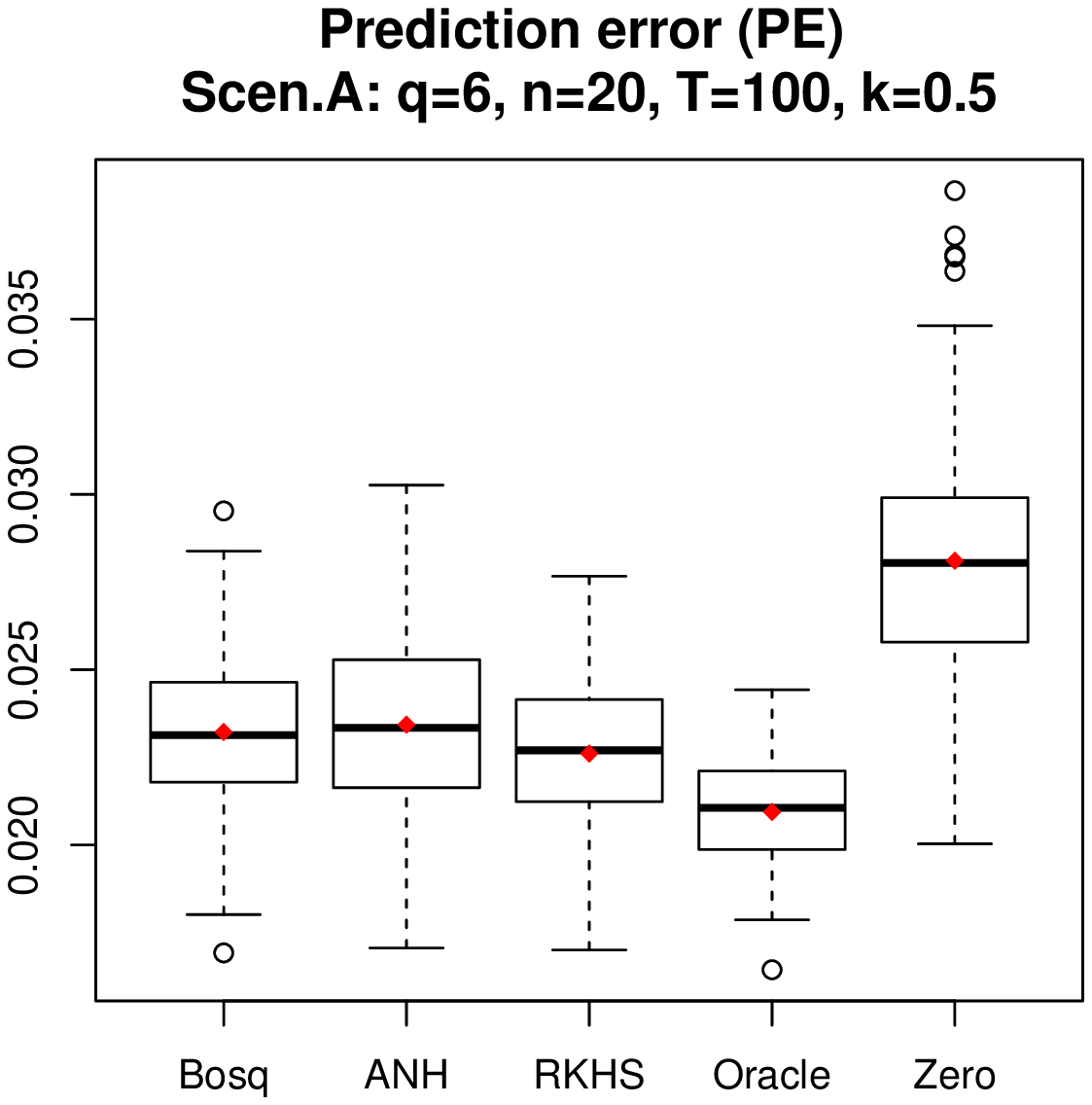}
		\vspace{-1cm}
		\caption{}
	\end{subfigure}
	~
	\begin{subfigure}{0.32\textwidth}
		\includegraphics[angle=270, width=1.25\textwidth]{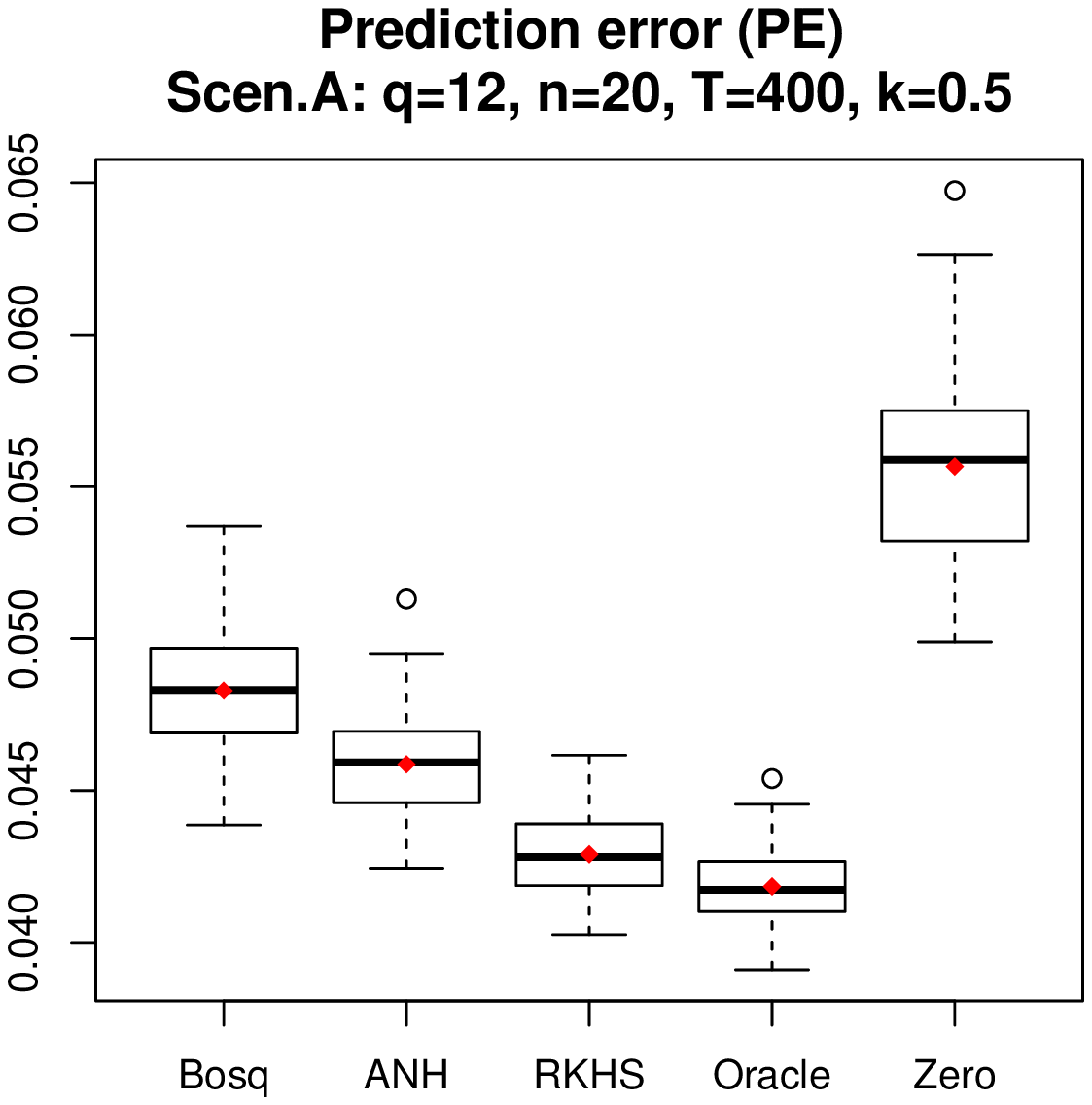}
		\vspace{-1cm}
		\caption{}
	\end{subfigure}
	~
	\begin{subfigure}{0.32\textwidth}
		\includegraphics[angle=270, width=1.25\textwidth]{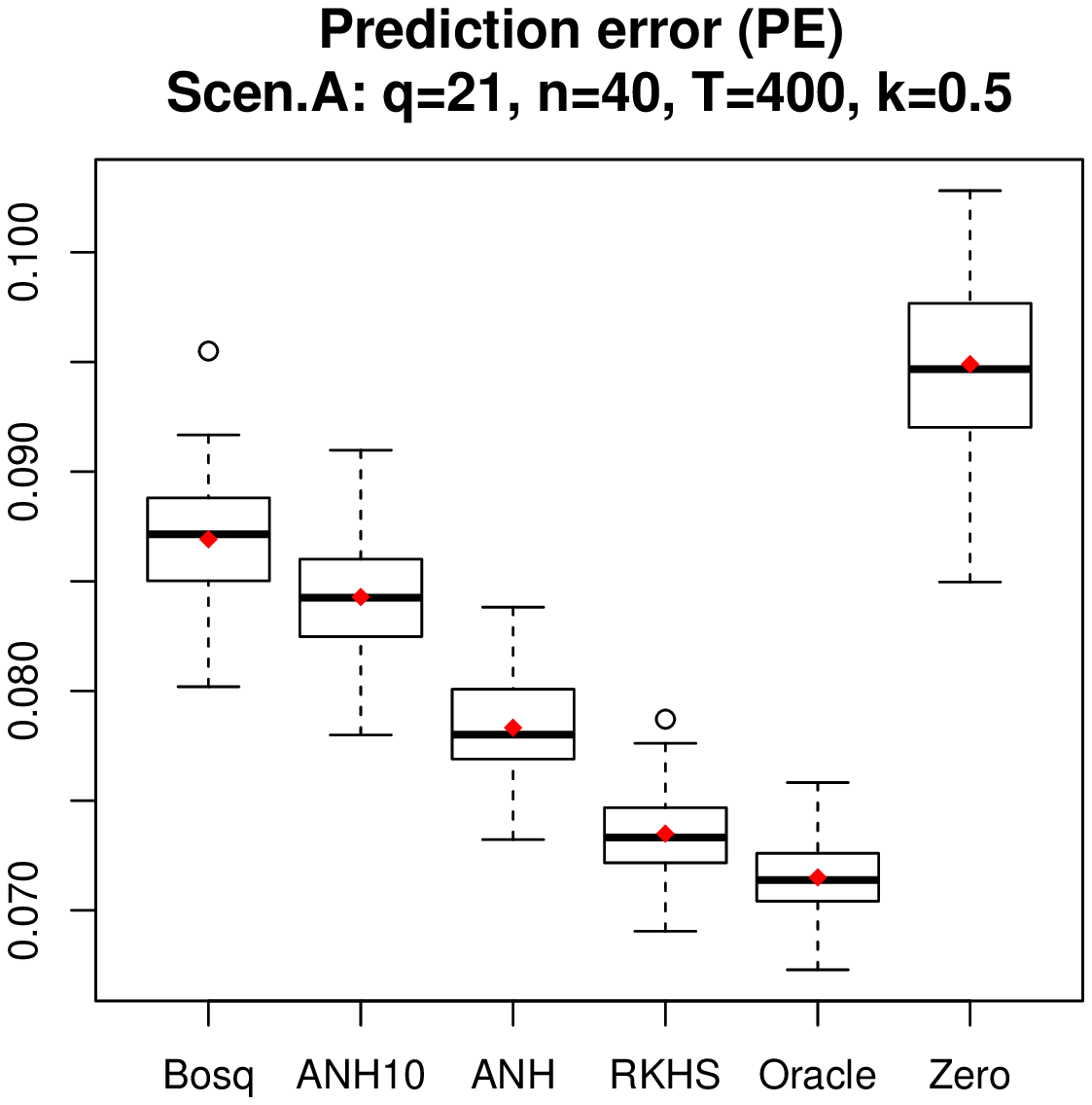}
		\vspace{-1cm}
		\caption{}
	\end{subfigure}
    ~
	\begin{subfigure}{0.32\textwidth}
		\includegraphics[angle=270, width=1.25\textwidth]{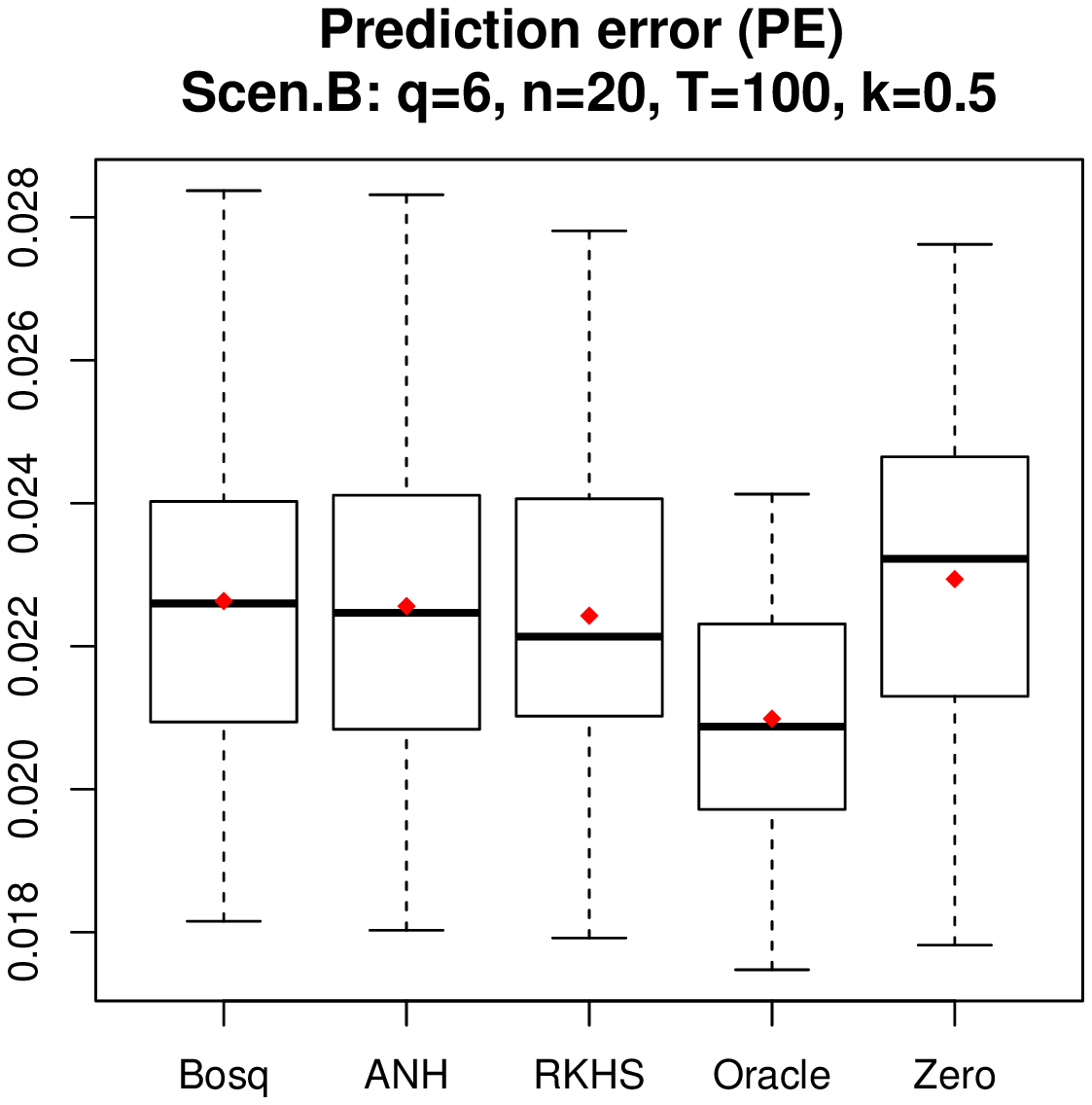}
		\vspace{-1cm}
		\caption{}
	\end{subfigure}
	~
	\begin{subfigure}{0.32\textwidth}
		\includegraphics[angle=270, width=1.25\textwidth]{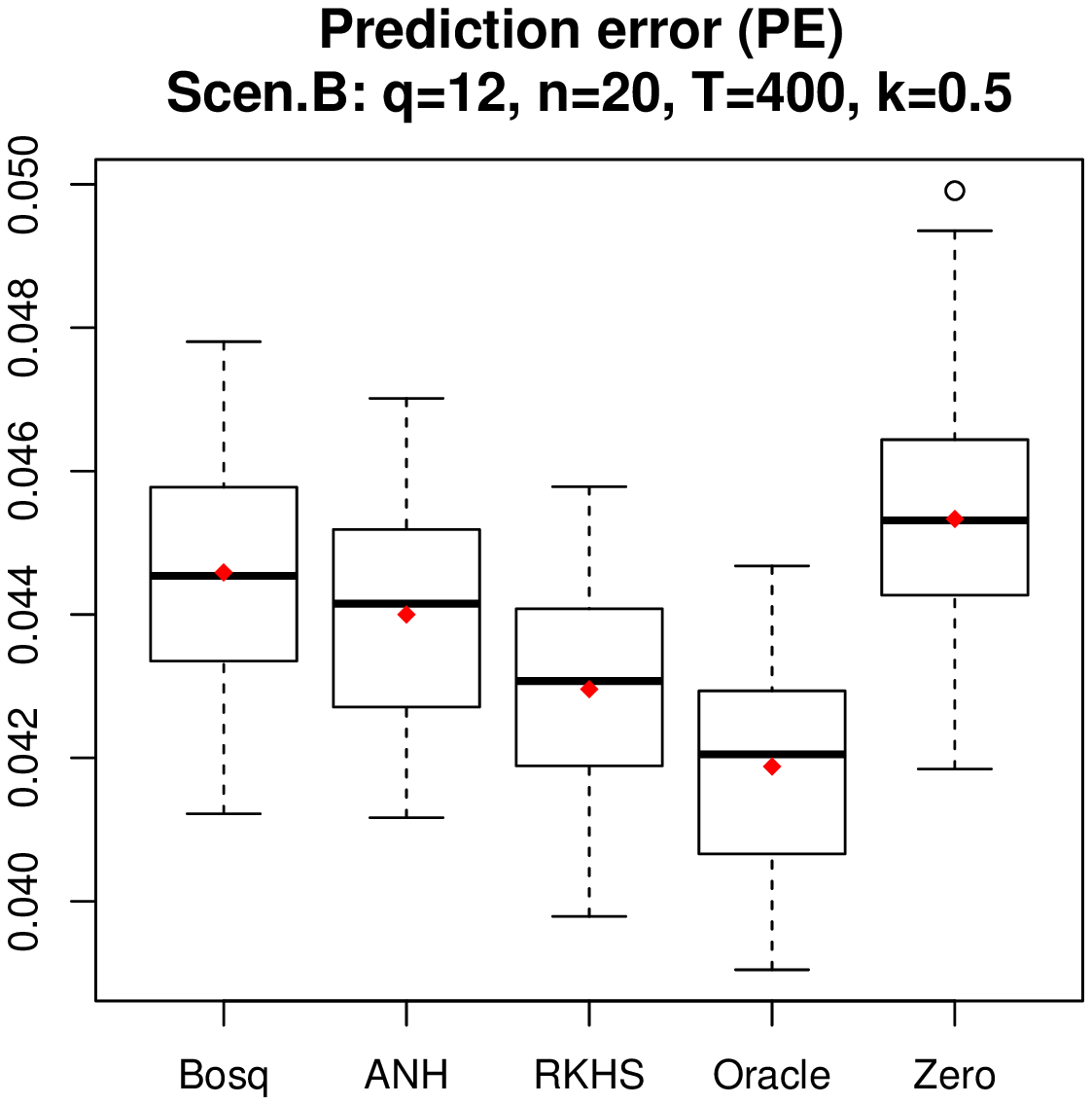}
		\vspace{-1cm}
		\caption{}
	\end{subfigure}
	~
	\begin{subfigure}{0.32\textwidth}
		\includegraphics[angle=270, width=1.25\textwidth]{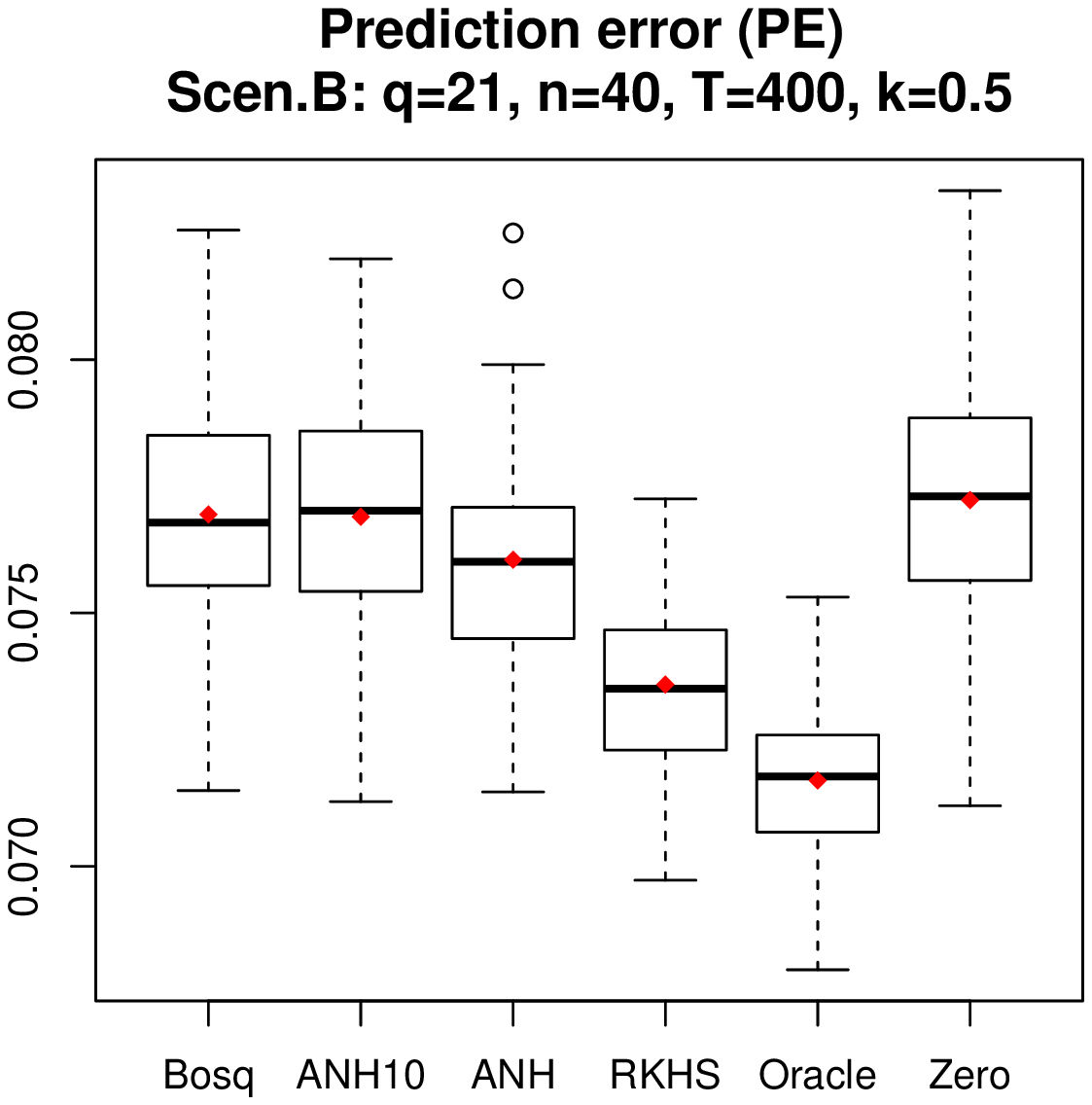}
		\vspace{-1cm}
		\caption{}
	\end{subfigure}
	~
	\begin{subfigure}{0.32\textwidth}
		\includegraphics[angle=270, width=1.25\textwidth]{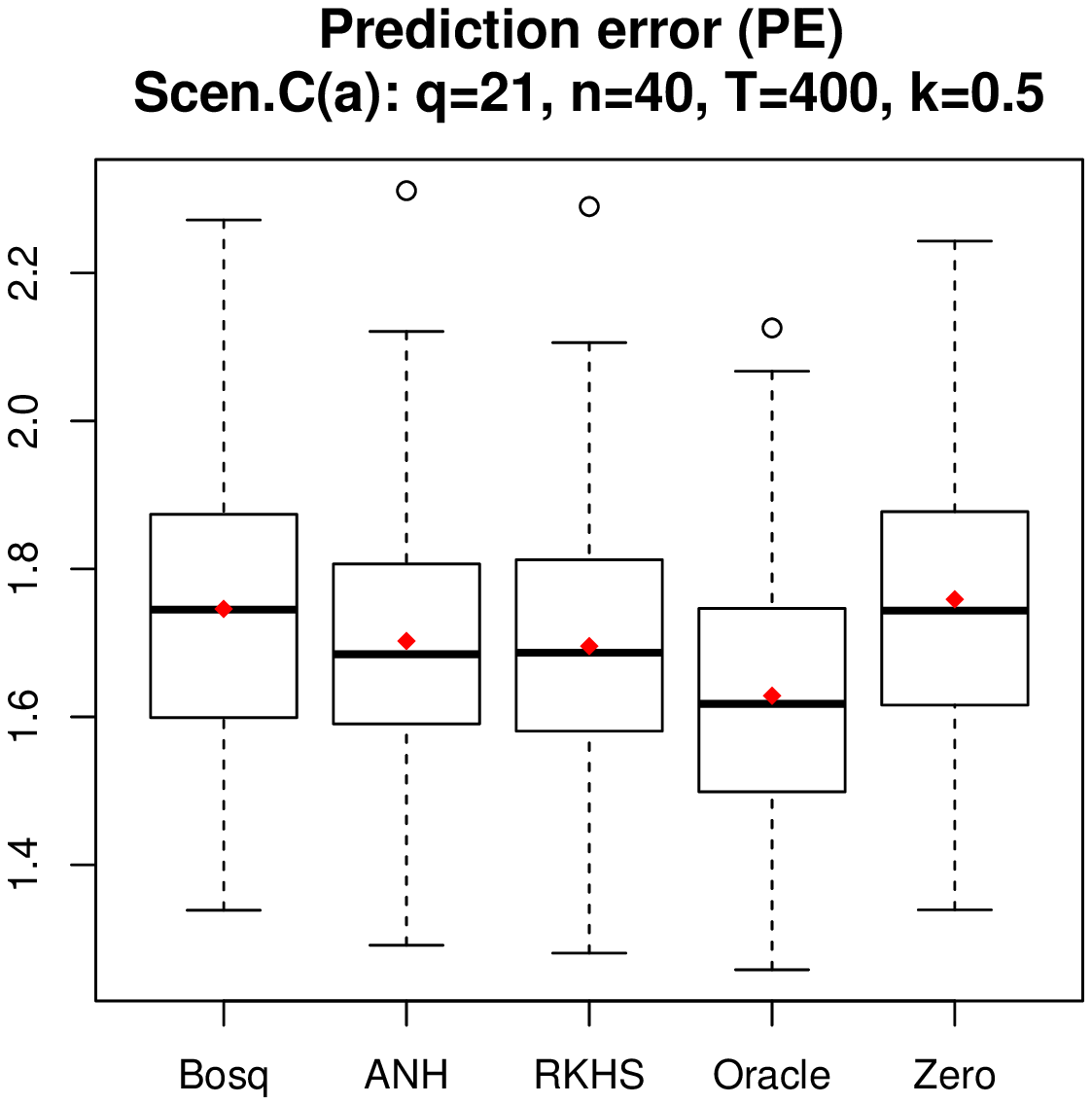}
		\vspace{-1cm}
		\caption{}
	\end{subfigure}
		~
	\begin{subfigure}{0.32\textwidth}
		\includegraphics[angle=270, width=1.25\textwidth]{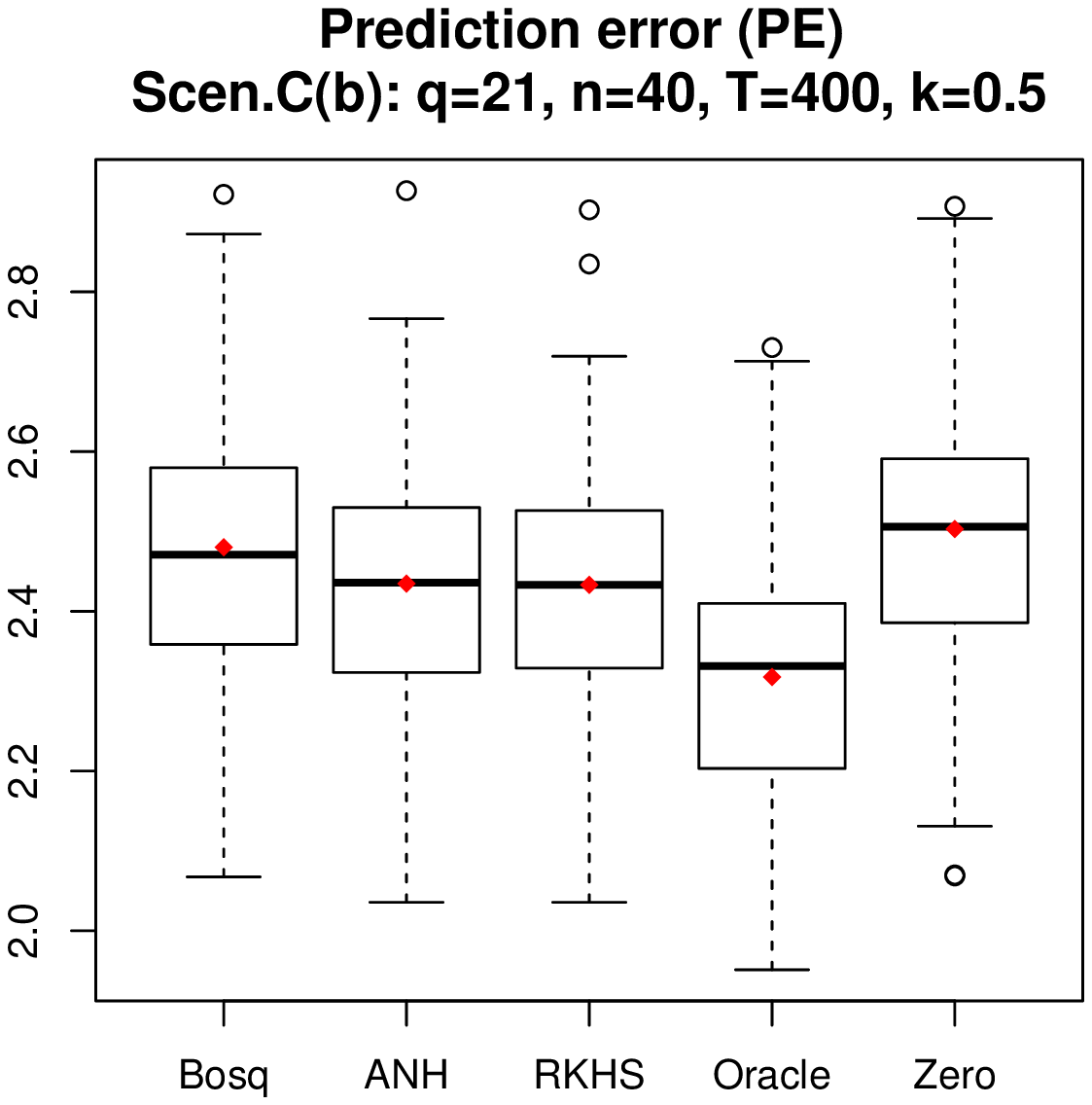}
		\vspace{-1cm}
		\caption{}
	\end{subfigure}
	\caption{Boxplot of prediction error~(PE) for FAR(1) across 100 experiments with signal strength $\kappa=0.5.$ In (c) and (f), ANH10 stands for ANH based on 10 cubic B-splines under $q=21$. Red points denote average PE across 100 experiments for each method.}
	\label{fig:FAR1_k5_PE}
\end{figure}

\subsection{Simulation result for FAR with autoregressive order selection}
This section investigates the performance of RKHS under an unknown FAR order $D$. Specifically, we consider an FAR(2) process under three signal-noise settings adapted from Section \ref{subsec:simu_far1}:
\vspace{-0.2cm}
\begin{itemize}\setlength\itemsep{-0.2em}
	\item Scenario A2 (Diag $\Lambda$): For $d=1,2$, $\Lambda_d=\text{diag}(\kappa_d,\cdots, \kappa_d)$ and $z_{ti}\overset{i.i.d.}{\sim} U(-a, a)$ with $a=0.1$ for $i=1,\cdots, q$.
	\item Scenario B2 (Random $\Lambda$): For $d=1,2$, a random matrix $\Lambda_d^*$ is first generated via $\Lambda^*_{d,ij}\overset{i.i.d.}{\sim} N(0,1)$ and we set $\Lambda_d=\kappa_d\cdot\Lambda_d^*/\sigma(\Lambda_d^*)$, and $z_{ti}\overset{i.i.d.}{\sim} U(-a, a)$ with $a=0.1$, for $i=1,\cdots, q$.
	\item Scenario C2 (ANH setting): (a) For $d=1,2,$ a random matrix $\Lambda_d^*$ is first generated via $\Lambda^*_{d,ij}\overset{ind.}{\sim} N(0,\sigma_i\sigma_j)$ and we set $\Lambda_d=\kappa_d\cdot\Lambda_d^*/\sigma(\Lambda_d^*)$, and $z_{ti}\overset{ind.}{\sim} N(0, \sigma_i^2)$, where $\sigma_{1:q}=(1:q)^{-1}$. (b) Same setting except $\sigma_{1:q}=1.2^{-(1:q)}$.
\end{itemize}

The implementation of ANH and RKHS is the same as that for FAR(1). The only difference is that we do not impose the autoregressive order $D=2$. Instead, we assume $D_{max}=2$ and let the algorithms select the FAR order. For ANH, the fFPE criterion is used to select the number of FPC $p$ and the FAR order $D$. For RKHS, we use 5-fold cross validation to select both $D$ and the tuning parameters $(\lambda_1,\lambda_2)$. Note that both fFPE and CV are prediction-based order selection criteria and thus may not always favor the correct autoregressive order. We do not include Bosq in the comparison as its performance is typically inferior to ANH and RKHS (as shown in Section \ref{subsec:simu_far1}) and additional procedure is needed to determine the FAR order for Bosq, see for example \cite{Kokoszka2013}.

\textbf{Numerical result for autoregressive order selection}: For Scenarios A2 and B2, we consider $(1) q=6, n=20, T=100, (2) q=12, n=20, T=400, (3) q=21, n=40, T=400$. For Scenario C2, we consider $q=21, n=40, T=400$. As for the signal level, we consider $(\kappa_1,\kappa_2)=(0.5,0.3)$ and $(\kappa_1,\kappa_2)=(0,0.5)$ similar as in \cite{aue2015prediction}. For each simulation setting, we conduct 100 experiments.

We summarize the numerical performance of RKHS and ANH in Table \ref{tab:FAR_orderselection}. Note that MISE is not well-defined when an incorrect autoregressive order is selected, thus we only report PE. For each experiment, PE is calculated based on the selected FAR model, which may or may not be FAR(2). Same as the analysis for FAR(1), we report the mean ratio (R$_{avg}$) of prediction improvement by RKHS and the percentage of experiments~(R$_w$) where RKHS achieves a lower PE than ANH. In addition, we report the percentage of experiments ($D_T$) where the algorithm selects the correct autoregressive order $D$. We further give the boxplot of PE in Figure \ref{fig:FAR_orderselection_PE_k1_5_k2_3} under signal strength $(\kappa_1,\kappa_2)=(0.5,0.3)$. The boxplot of PE under $(\kappa_1,\kappa_2)=(0,0.5)$ can be found in the Appendix.

RKHS continues to offer better performance for estimation and prediction under autoregressive order selection. Note that for $(\kappa_1,\kappa_2)=(0,0.5)$, both RKHS and ANH can identify the correct FAR order accurately, as ignoring $\kappa_2$ can result in large prediction error. However, this is not the case for $(\kappa_1,\kappa_2)=(0.5,0.3)$, where due to possible bias-variance trade-off, FAR(1) may be the favored model for prediction. 

\begin{table}[H]
	\centering
	\centerline{\small
	\begin{tabular}{llrrrr|rrrr}
		\hline
		\hline & & \multicolumn{4}{c|}{ Scenario A2:  $q= 6 , n= 20 , T= 100 $} & \multicolumn{4}{c}{ Scenario B2:  $q= 6 , n= 20 , T= 100 $} \\
		& Method  & PE$_{avg}$ & D$_{T}$(\%) & R$_{avg}$(\%) & R$_{w}$(\%) &PE$_{avg}$ & D$_{T}$(\%) & R$_{avg}$(\%) & R$_{w}$(\%)\\
		\hline
		$\kappa_1,\kappa_2=0.5,0.3$ & RKHS & \bf 2.557 & 49 & 5.28 & 70 & \bf 2.315 & 6 & 2.24 & 71 \\ 
		& ANH & 2.696 & 93 &  &  & 2.364 & 36 &  &  \\ 
		$\kappa_1,\kappa_2=0,0.5$ & RKHS & \bf 2.244 & 100 & 9.45 & 91 & \bf 2.248 & 99 & 3.96 & 72 \\ 
		& ANH & 2.456 & 100 &  &  & 2.334 & 88 &  &  \\ 
		\hline
		\hline & & \multicolumn{4}{c|}{ Scenario A2:  $q= 12 , n= 20 , T= 400 $} & \multicolumn{4}{c}{ Scenario B2:  $q= 12 , n= 20 , T= 400 $} \\
		& Method  & PE$_{avg}$ & D$_{T}$(\%) & R$_{avg}$(\%) & R$_{w}$(\%) &PE$_{avg}$ & D$_{T}$(\%) & R$_{avg}$(\%) & R$_{w}$(\%)\\
		\hline
		$\kappa_1,\kappa_2=0.5,0.3$ & RKHS & \bf 4.418 & 100 & 26.26 & 100 & \bf 4.389 & 61 & 3.01 & 99 \\ 
		& ANH & 5.577 & 100 &  &  & 4.521 & 52 &  &  \\ 
		$\kappa_1,\kappa_2=0,0.5$ & RKHS & \bf 4.296 & 100 & 9.56 & 100 & \bf 4.290 & 100 & 4.53 & 100 \\ 
		& ANH & 4.707 & 100 &  &  & 4.484 & 100 &  &  \\ 
		\hline
		\hline & & \multicolumn{4}{c|}{ Scenario A2:  $q= 21 , n= 40 , T= 400 $} & \multicolumn{4}{c}{ Scenario B2:  $q= 21 , n= 40 , T= 400 $} \\
		& Method  & PE$_{avg}$ & D$_{T}$(\%) & R$_{avg}$(\%) & R$_{w}$(\%) &PE$_{avg}$ & D$_{T}$(\%) & R$_{avg}$(\%) & R$_{w}$(\%)\\
		\hline
		$\kappa_1,\kappa_2=0.5,0.3$ & RKHS & \bf 7.569 & 100 & 22.78 & 100 & \bf 7.558 & 42 & 3.69 & 100 \\ 
		& ANH & 9.292 & 100 &  &  & 7.836 & 3 &  &  \\ 
		$\kappa_1,\kappa_2=0,0.5$ & RKHS & \bf 7.347 & 100 & 12.22 & 100 & \bf 7.374 & 100 & 5.11 & 100 \\ 
		& ANH & 8.245 & 100 &  &  & 7.750 & 99 &  &  \\ 
		\hline
		\hline & & \multicolumn{4}{c|}{ Scenario C2(a):  $q= 21 , n= 40 , T= 400 $} & \multicolumn{4}{c}{ Scenario C2(b):  $q= 21 , n= 40 , T= 400 $} \\
		& Method  & PE$_{avg}$ & D$_{T}$(\%) & R$_{avg}$(\%) & R$_{w}$(\%) &PE$_{avg}$ & D$_{T}$(\%) & R$_{avg}$(\%) & R$_{w}$(\%)\\
		\hline
		$\kappa_1,\kappa_2=0.5,0.3$ & RKHS & \bf 1.727 & 65 & 0.24 & 51 & 2.494 & 51 & -0.18 & 43 \\ 
		& ANH & 1.731 & 9 &  &  & \bf 2.489 & 1 &  &  \\ 
		$\kappa_1,\kappa_2=0,0.5$ & RKHS & \bf 1.692 & 99 & 1.19 & 71 & \bf 2.438 & 95 & 1.87 & 81 \\ 
		& ANH & 1.712 & 96 &  &  & 2.483 & 88 &  &  \\ 
		\hline\hline
\end{tabular}}
	\caption{Numerical performance of various methods for FAR(2) processes with autoregressive order selection. Methods considered are RKHS~(this paper) and ANH~\citep{aue2015prediction}. Bold font indicates the best performance, where the proposed RKHS method is generally the best performer in Scenarios A2 and B2. (PE$_{avg}$ is multiplied by 100 in scale for Scenarios A2 and B2, but not for Scenario C2.)} 
\label{tab:FAR_orderselection}
\end{table}

\begin{figure}[H]
	\begin{subfigure}{0.32\textwidth}
		\includegraphics[angle=270, width=1.25\textwidth]{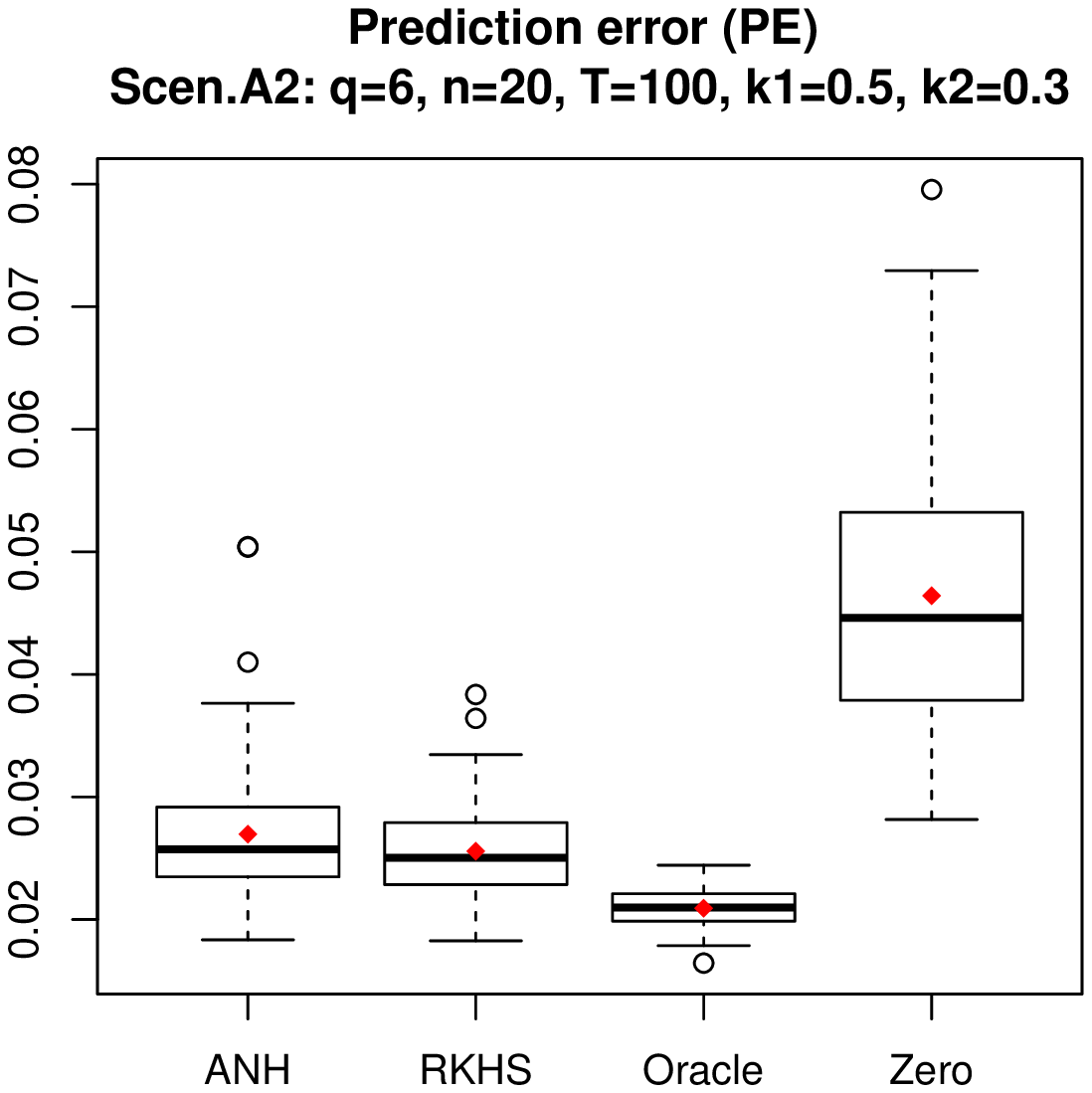}
		\vspace{-1cm}\caption{}
	\end{subfigure}
	~
	\begin{subfigure}{0.32\textwidth}
		\includegraphics[angle=270, width=1.25\textwidth]{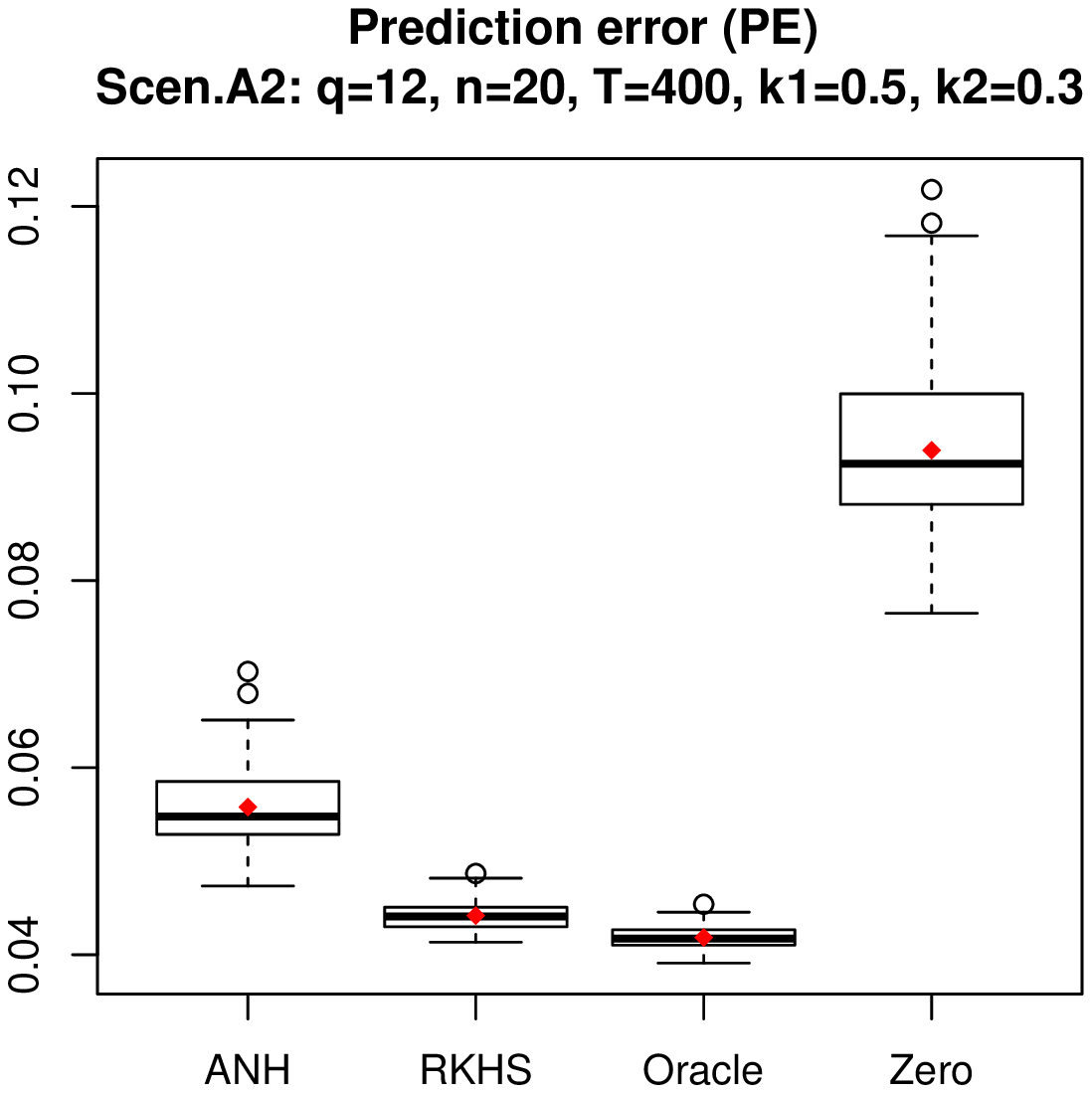}
		\vspace{-1cm}\caption{}
	\end{subfigure}
	~
	\begin{subfigure}{0.32\textwidth}
		\includegraphics[angle=270, width=1.25\textwidth]{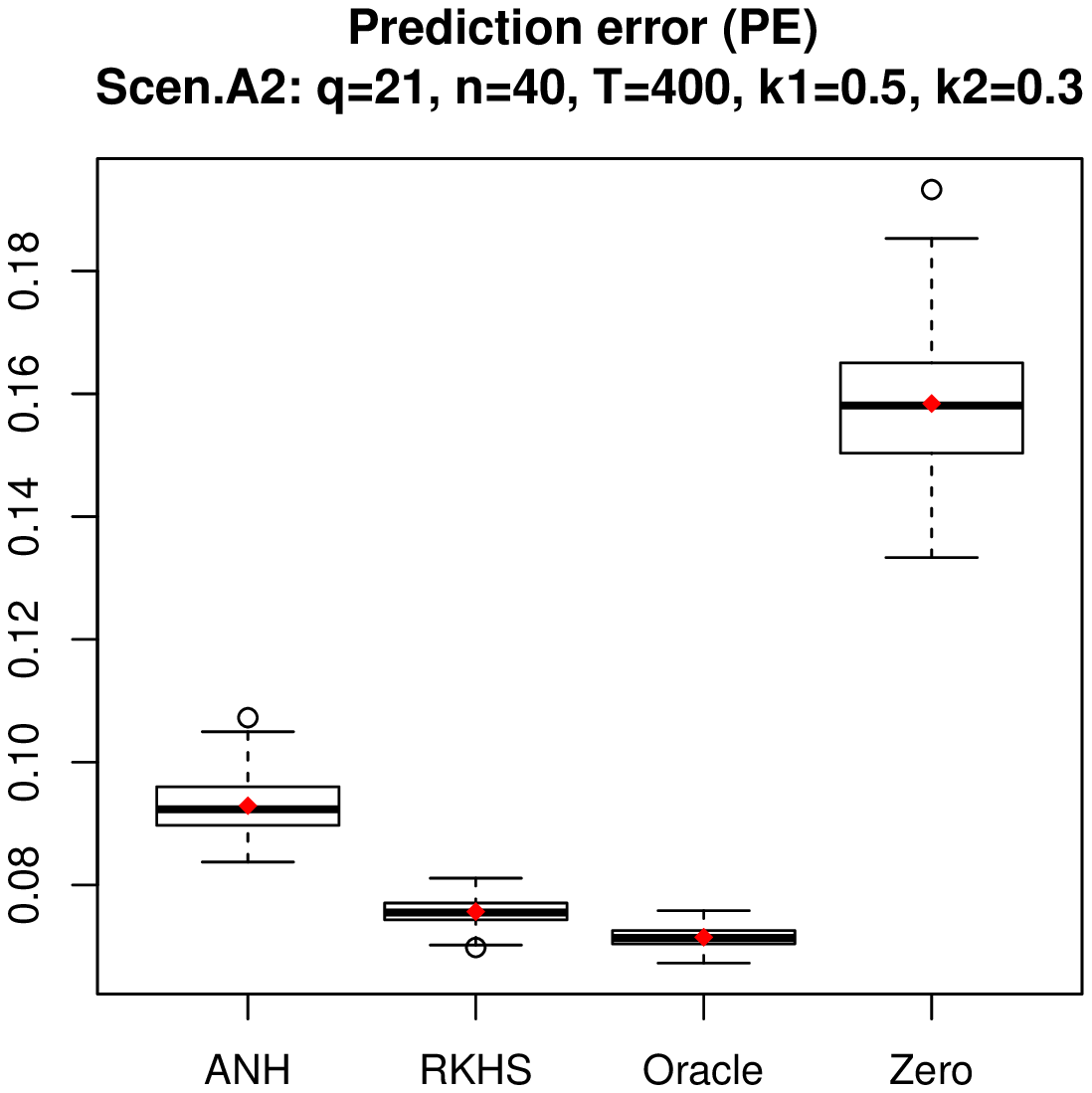}
		\vspace{-1cm}\caption{}
	\end{subfigure}
	~
	\begin{subfigure}{0.32\textwidth}
		\includegraphics[angle=270, width=1.25\textwidth]{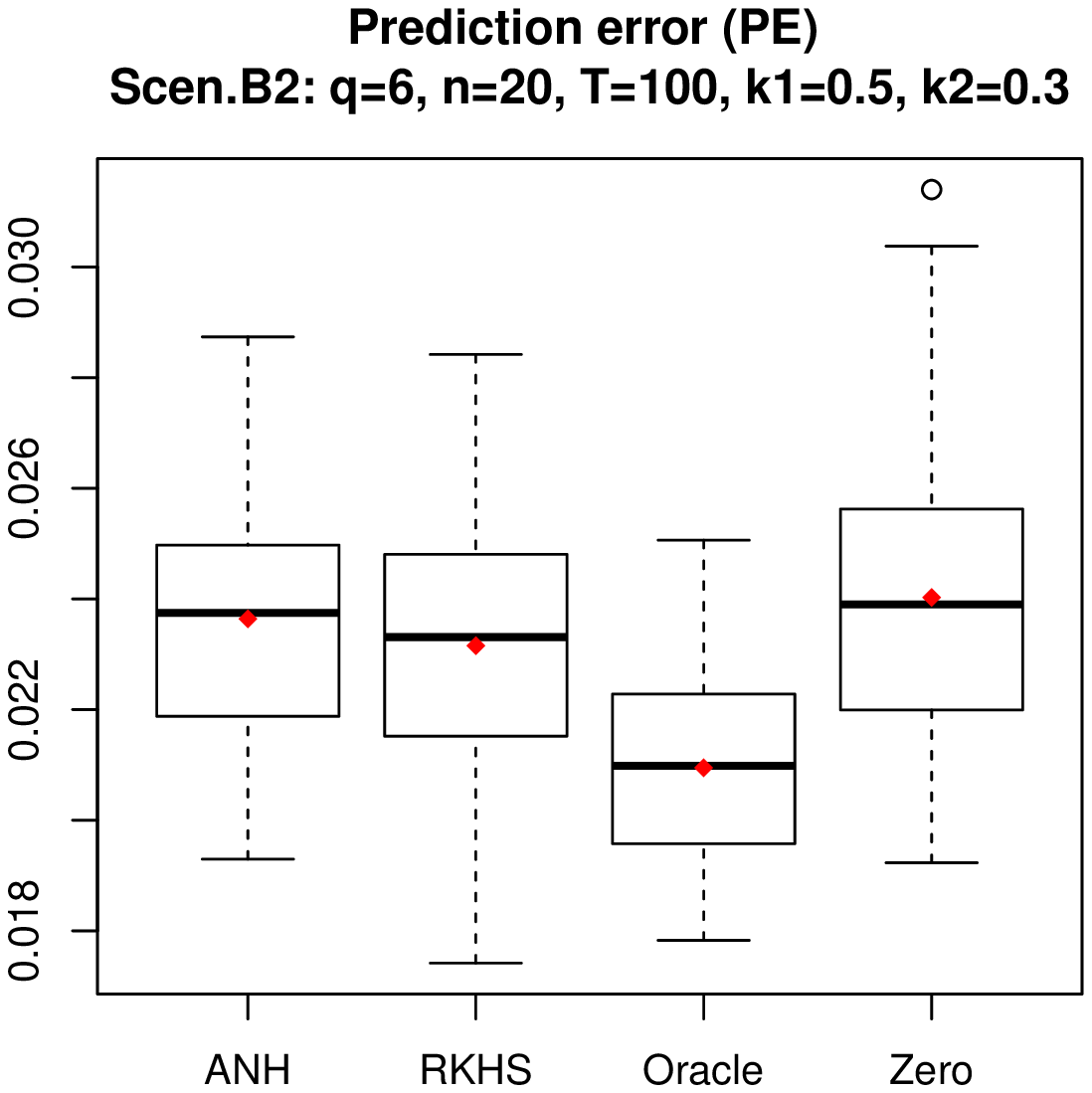}
		\vspace{-1cm}\caption{}
	\end{subfigure}
	~
	\begin{subfigure}{0.32\textwidth}
		\includegraphics[angle=270, width=1.25\textwidth]{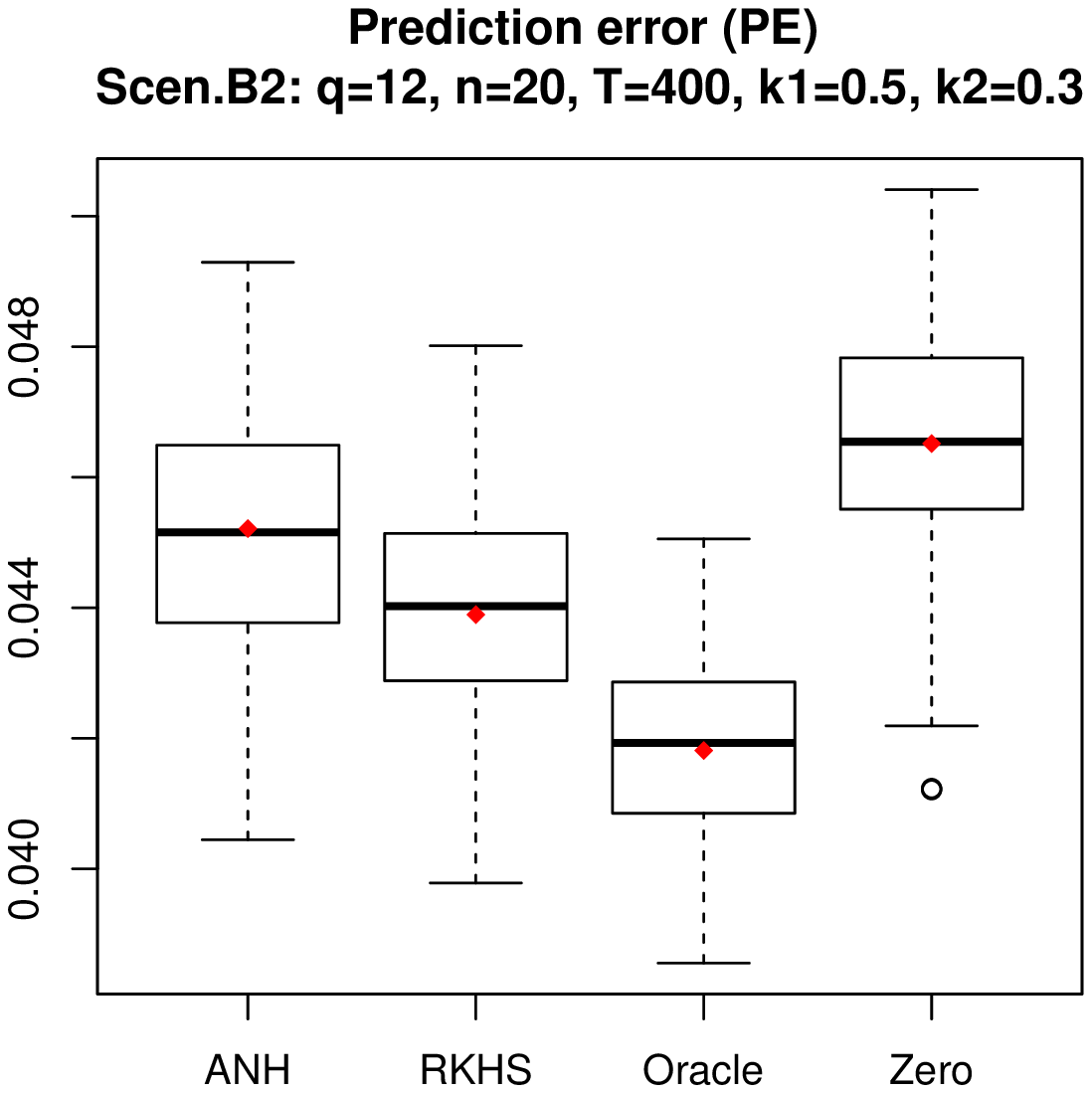}
		\vspace{-1cm}\caption{}
	\end{subfigure}
	~
	\begin{subfigure}{0.32\textwidth}
		\includegraphics[angle=270, width=1.25\textwidth]{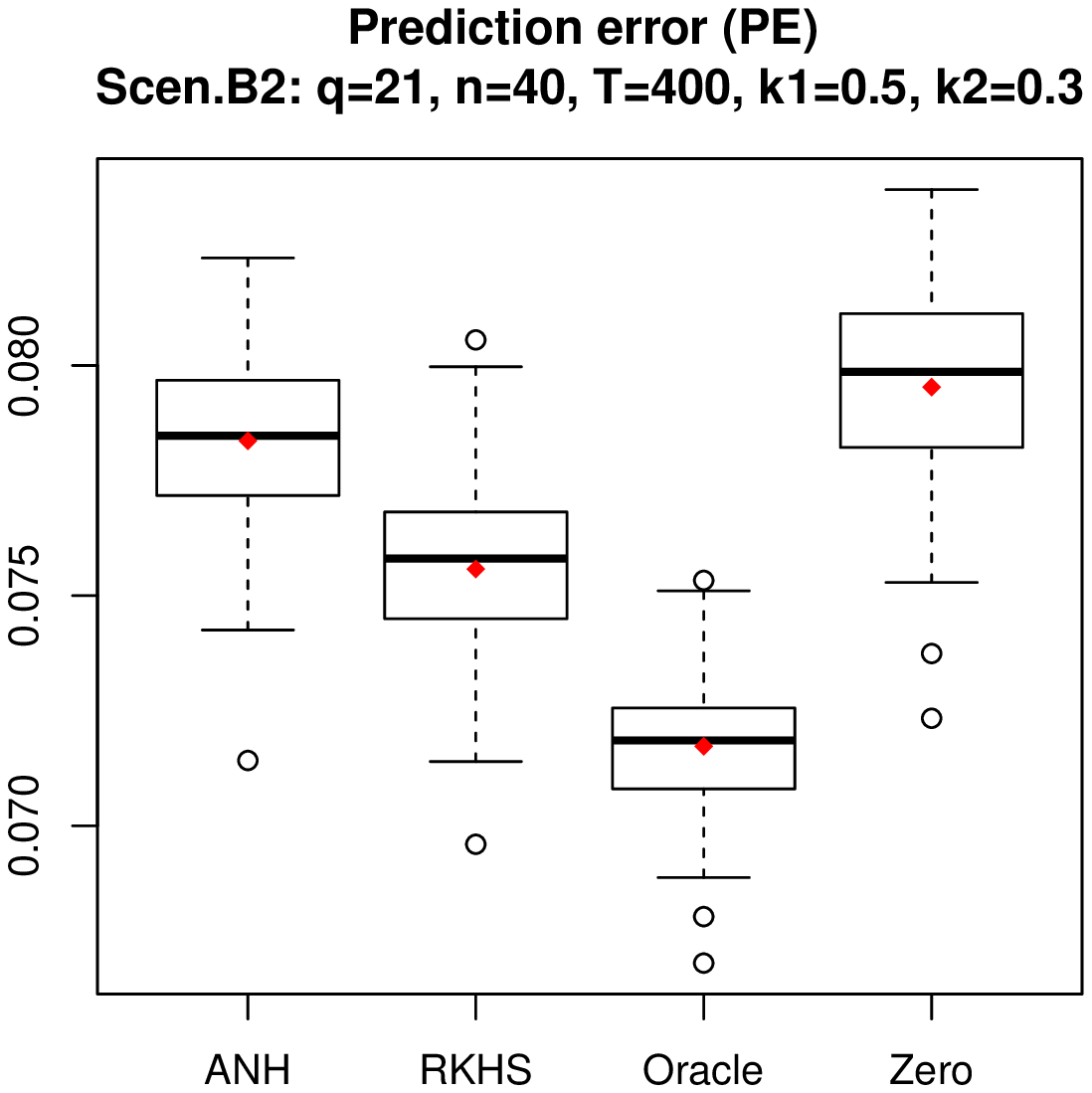}
		\vspace{-1cm}\caption{}
	\end{subfigure}
	~
	\begin{subfigure}{0.32\textwidth}
		\includegraphics[angle=270, width=1.25\textwidth]{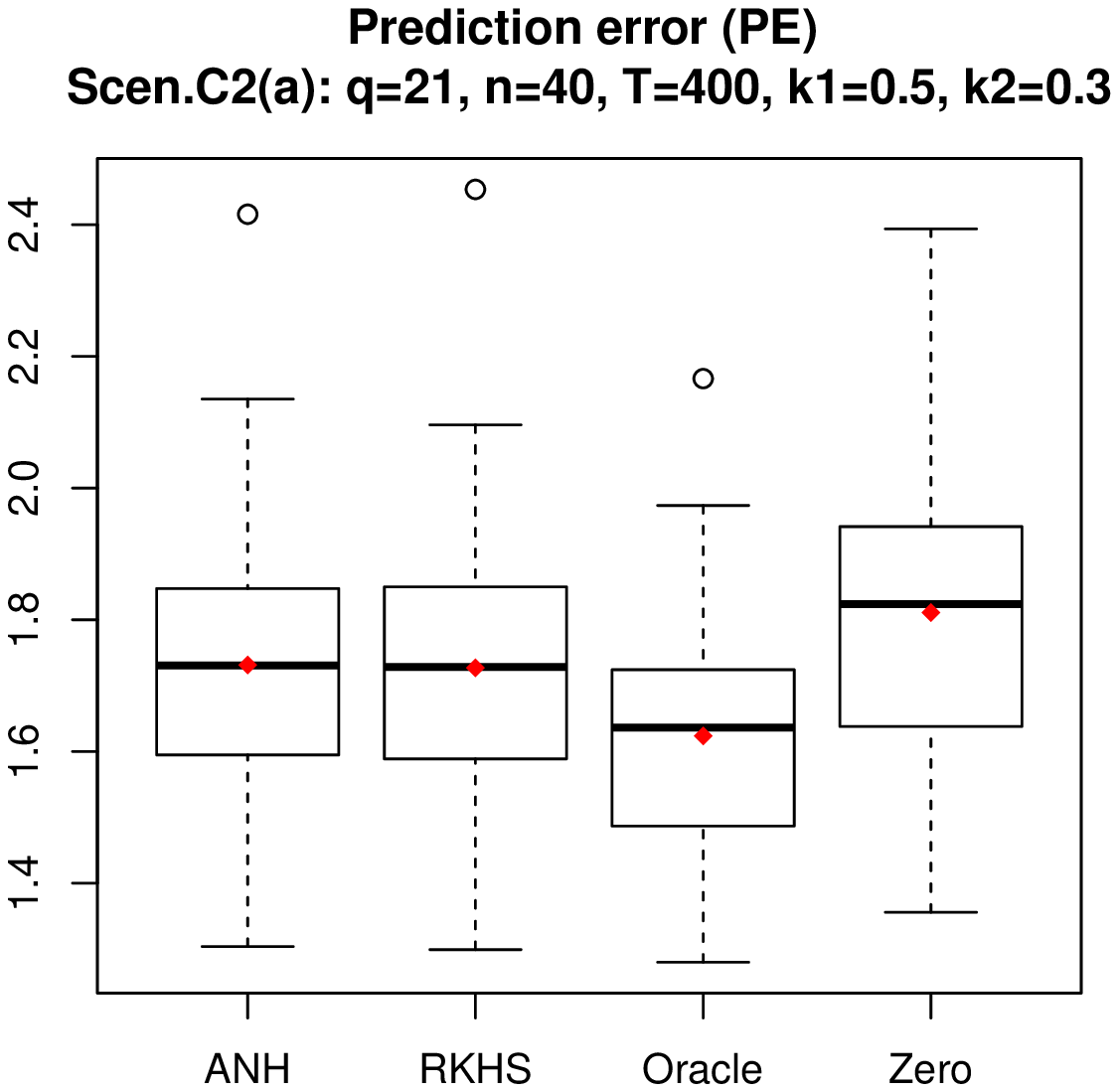}
		\vspace{-1cm}\caption{}
	\end{subfigure}
	~
	\begin{subfigure}{0.32\textwidth}
		\includegraphics[angle=270, width=1.25\textwidth]{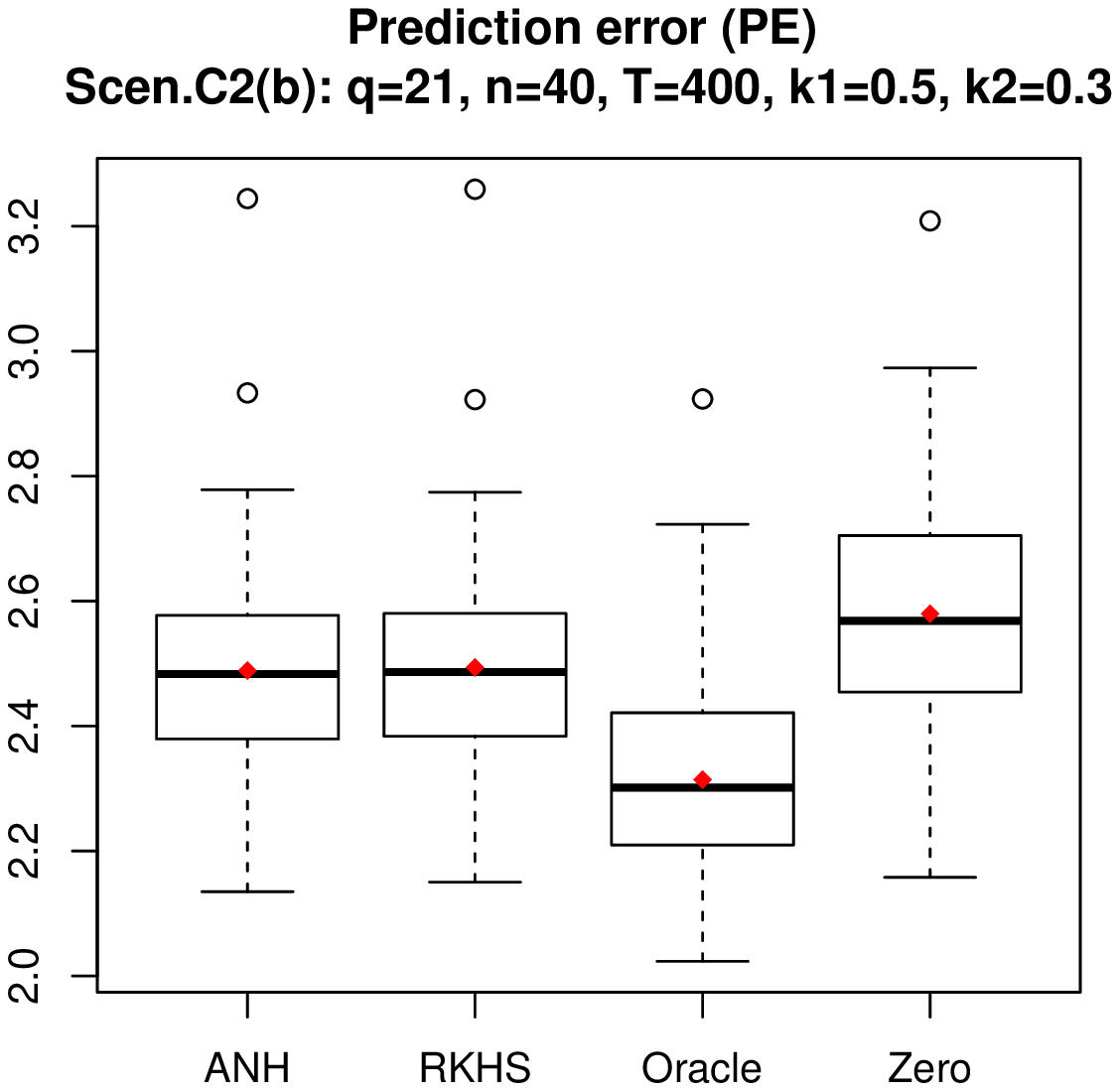}
		\vspace{-1cm}\caption{}
	\end{subfigure}
	\caption{Boxplot of prediction error~(PE) for FAR with autoregressive order selection across 100 experiments with signal strength $\kappa_1=0.5, \kappa_2=0.3.$ Red points denote average PE across 100 experiments for each method.}
	\label{fig:FAR_orderselection_PE_k1_5_k2_3}
\end{figure}

\section{Real Data Application}\label{sec:realdata}
In this section, we give an illustrative example of RKHS in predicting functional time series. Specifically, we consider the Utility demand data in the book \cite{Hyndman2008}, which is publicly available in R package \texttt{expsmooth}. The data contains 125 daily curves of hourly utility demand from a company in Midwest U.S., starting from January 2003. Thus we have $T=125$ and $n=24.$

The original functional time series is given in Figure \ref{fig:utility}. As can be seen, the utility demand seems to be non-stationary with a downward trend. Thus, we take the first-order difference and study the differenced time series, in other words, we study the derivative of the utility demand curve. We partition the time series into training data, which contains the first 100 daily curves, and test data, which contains the last 25 daily curves.

Based on the training data, we estimate three FAR(1) models using Bosq, ANH and RKHS respectively, where the implementation of the three methods is the same as in Section \ref{subsec:simu_far1}. The estimated model is then used to generate one-day ahead prediction for the 25 daily curves in the test data. For reference, we also implement a naive prediction, where the lagged $X_{t-1}$ is used to predict $X_t.$ For each day $t=101,\cdots, 125$ in the test data, we calculate the prediction error for each hour $e_{ti}=X_t(s_i)-\widehat{X}_t(s_i), i=1,\cdots,24$ where $\widehat{X}_t(s_i)$ denotes the predicted value, and we define RMSE$_t=\sqrt{\sum_{i=1}^{24}e_{ti}^2/24}$ and MAE$_t=\sum_{i=1}^{24}|e_{ti}|/24$. Table \ref{tab:utility} summarizes the prediction performance of the four methods, where for each method, we report the average RMSE and average MAE across the 25 days on the test data. RKHS gives the best performance, followed by ANH. In addition, we give the percentage of days when RKHS gives the best performance among the four methods in terms of RMSE$_t$ and MAE$_t$, and it is seen that RKHS wins around 60\% of the time.

Figure \ref{fig:utility_Aest} gives the 3D plot of the transition operator $\widehat{A}(s,r)$ estimated by Bosq, ANH and RKHS. As can be seen, the transition operators by ANH and RKHS exhibit wider range and more complex nature than Bosq. This is also confirmed by the singular values of the estimated transition operator, as is shown in Figure \ref{fig:utility}(right). Bosq selects $p=3$ FPC and the fFPE criterion of ANH selects $p=5$ FPC, which reflects on the rank of the estimated $\widehat A(s,r)$.

For illustration, Figure \ref{fig:prediction}(left) further plots the mean observation $\overline{X}(s_i)=\sum_{t=101}^{125}X_t(s_i)/25$ for each hour $i=1,\cdots, 24$ across the 25 days on the test data, together with the mean prediction $\overline{\widehat{X}}(s)=\sum_{t=101}^{125}\widehat{X}_t(s)/25$ by Bosq, ANH and RKHS. As can be seen, the mean prediction given by RKHS performs the best, while Bosq and ANH seem to miss some variation in the data, possibly due to information loss in dimension reduction. Figure \ref{fig:prediction}(center, right) plot sectional views of the estimated transition operator $\widehat A(s,\cdot)$ at $s=5$ and $s=20$ hour respectively. Same as in Figure \ref{fig:utility_Aest}, the transition operators estimated by RKHS and ANH exhibit more structures than Bosq. Additionally, it seems that the end of day observations have high predictive power as $\widehat A(s,r)$ takes larger absolute values around $r=20$, though the direction may be different for different hours $s$.

\begin{table}[H]
	\centering
	\begin{tabular}{rrrrrr}
		\hline\hline
		& Bosq & ANH & RKHS & Naive & RKHS wins \\ 
		\hline
		average RMSE & 268.30 & 239.14 & \bf 201.64 & 301.80 & 60\% \\ 
		average MAE & 191.05 & 173.23 & \bf 147.84 & 203.79 & 60\% \\ 
		\hline\hline
	\end{tabular}
	\caption{Prediction performance of Bosq, ANH, RKHS and Naive on the test data.}
\label{tab:utility}
\end{table}

\begin{figure}[H]
	\includegraphics[angle=270, width=1\textwidth]{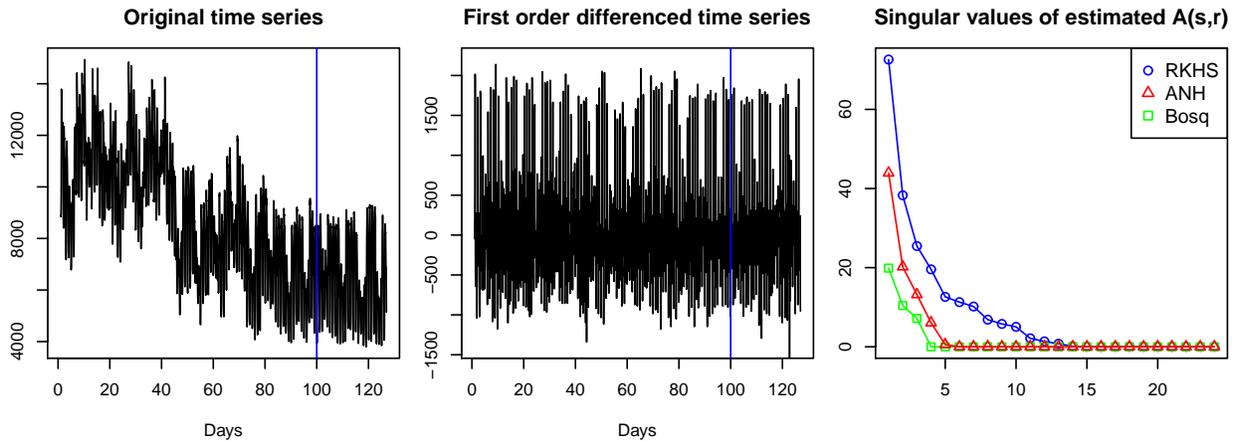}
	\caption{Left plot: Original utility demand time series. Center plot: First order differenced utility demand time series. (The blue vertical line marks the beginning of test data.) Right plot: Singular values of estimated $A(s,r)$ by Bosq, ANH and RKHS.}
	\label{fig:utility}
\end{figure}

\begin{figure}[H]
	\hspace{-1cm}
	\begin{subfigure}{0.3\textwidth}
		\includegraphics[angle=270, width=1.5\textwidth]{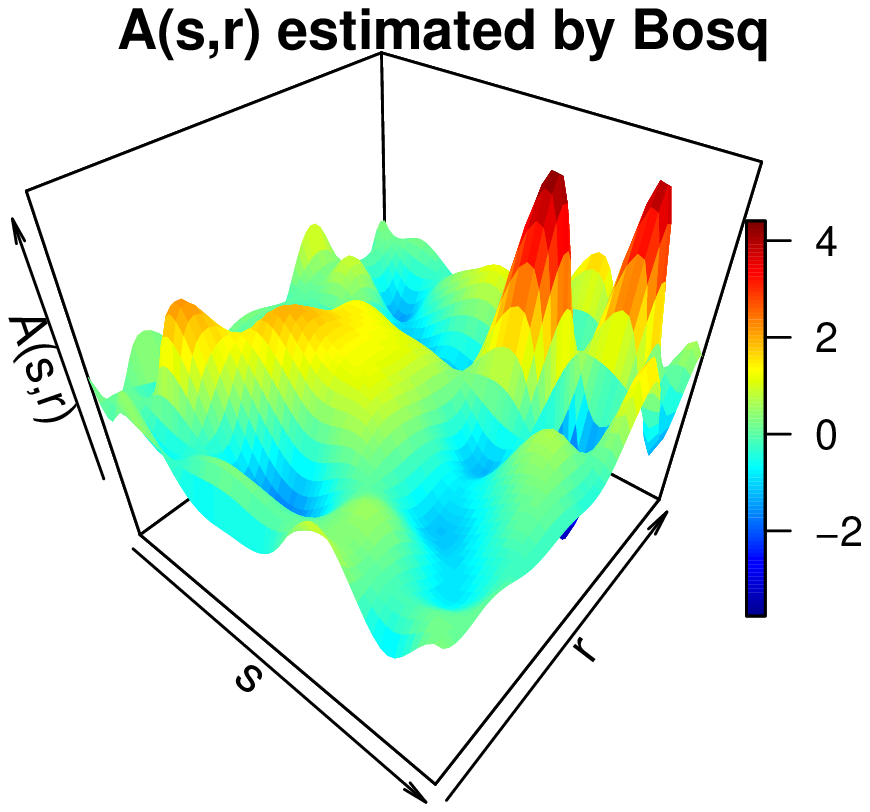}
		\vspace{-1cm}
	\end{subfigure}
	~
	\begin{subfigure}{0.3\textwidth}
		\includegraphics[angle=270, width=1.5\textwidth]{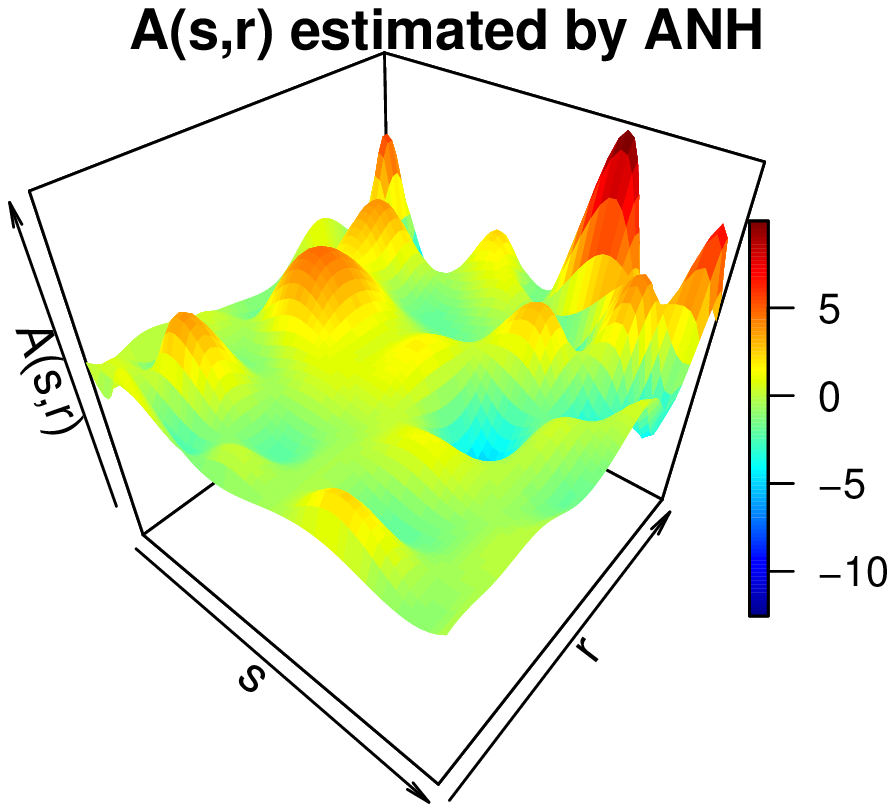}
		\vspace{-1cm}
	\end{subfigure}
	~
	\begin{subfigure}{0.3\textwidth}
		\includegraphics[angle=270, width=1.5\textwidth]{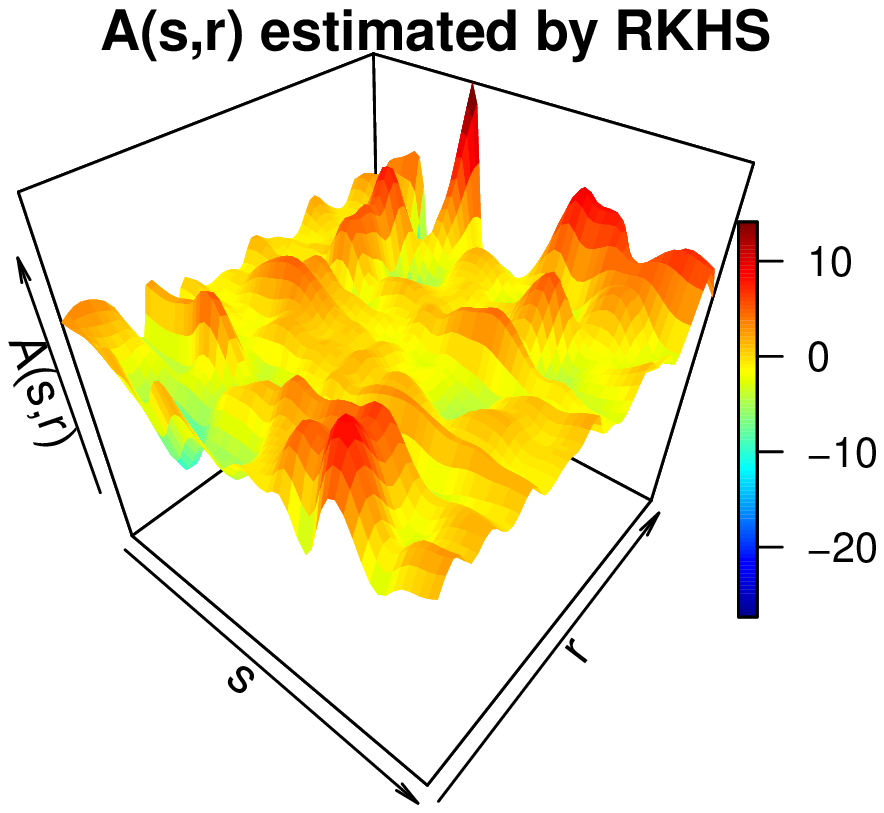}
		\vspace{-1cm}
	\end{subfigure}
	\caption{3D plot of the transition operator $A(s,r)$ estimated by Bosq, ANH and RKHS.}
    \label{fig:utility_Aest}
\end{figure}

\begin{figure}[H]
	\includegraphics[angle=270, width=1\textwidth]{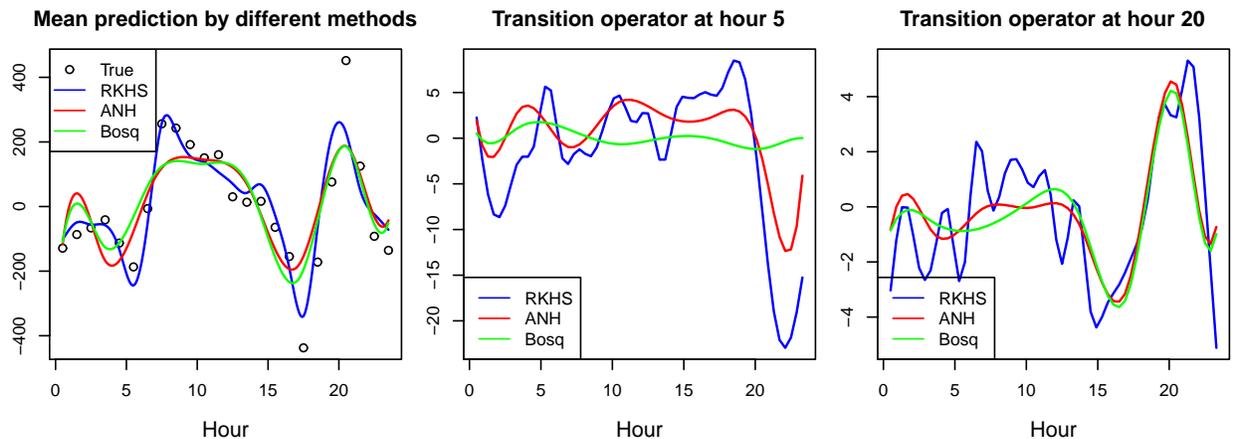}
	\caption{Left plot: Mean prediction for each hour on the test data. Center plot: Estimated transition operator $\widehat A(s,\cdot)$ at hour $s=5$ (sectional view). Right plot: Estimated transition operator $\widehat A(s,\cdot)$ at hour $s=20$ (sectional view).}
	\label{fig:prediction}
\end{figure}

\section{Conclusion}\label{sec:conclusion}
In this paper, we study the inference~(estimation and prediction) of the FAR process through the lens of RKHS. Unlike existing literature, the proposed inference framework does not require dimension reduction thanks to the derived Representer theorem. The proposed method works directly with discrete measurements of the functional time series and we show that the nuclear norm penalization estimator is consistent with sharp convergence rate. Simulation studies and a real data application further demonstrate the promising performance of the proposed method and the advantage of dimension reduction free inference. A natural extension of the current framework is the scenario of noisy measurements, where the functional time series is observed with additional (i.i.d.) measurement errors at each sampling point. We expect the method and theory developed in the current paper continue to work under the noisy measurement scenario and leave the thorough investigation for future research.


\setstretch{1}
\bibliographystyle{apalike}
\bibliography{citations} 

\begin{thebibliography}{}

\bibitem[Antoniadis et~al., 2006]{Antoniadis2006}
Antoniadis, A., Paparoditis, E., and Sapatinas, T. (2006).
\newblock A functional wavelet–kernel approach for time series prediction.
\newblock {\em Journal of the Royal Statistical Society - Series B}, 68.

\bibitem[Antoniadis and Sapatinas, 2003]{Antoniadis2003}
Antoniadis, A. and Sapatinas, T. (2003).
\newblock Wavelet methods for continuous-time predictionusing hilbert-valued
  autoregressive processes.
\newblock {\em Journal of Multivariate Analysis}, 87(1):133–158.

\bibitem[Aue et~al., 2015]{aue2015prediction}
Aue, A., Norinho, D.~D., and H{\"o}rmann, S. (2015).
\newblock On the prediction of stationary functional time series.
\newblock {\em Journal of the American Statistical Association},
  110(509):378--392.

\bibitem[Bach, 2008]{Bach2008}
Bach, F.~R. (2008).
\newblock Consistency of trace norm minimization.
\newblock {\em Journal of Machine Learning Research}, (8):1019--1048.

\bibitem[Bartlett et~al., 2005]{bartlett2005local}
Bartlett, P.~L., Bousquet, O., and Mendelson, S. (2005).
\newblock Local rademacher complexities.
\newblock {\em Annals of Statistics}, 33(4):1497--1537.

\bibitem[Basu and Michailidis, 2015]{basu2015regularized}
Basu, S. and Michailidis, G. (2015).
\newblock Regularized estimation in sparse high-dimensional time series models.
\newblock {\em Annals of Statistics}, 43(4):1535--1567.

\bibitem[Besse and Cardot, 1996]{Besse1996}
Besse, P.~C. and Cardot, H. (1996).
\newblock Approximation spline de la prevision d'un processus fonctionnel
  autorégressif d'ordre 1.
\newblock {\em Canadian Journal of Statistics}, 24(4):467--487.

\bibitem[Besse et~al., 2000]{Besse2000}
Besse, P.~C., Cardot, H., and Stephenson, D.~B. (2000).
\newblock Autoregressive forecasting of some functional climatic variations.
\newblock {\em Scandinavian Journal of Statistics}, 27(4):673--687.

\bibitem[Bosq, 1998]{Bosq1998}
Bosq, D. (1998).
\newblock {\em Nonparametric Statistics for Stochastic Processes: Estimation
  and Prediction}.
\newblock Springer-Verlag New York, 2 edition.

\bibitem[Bosq, 2000]{bosq2000linear}
Bosq, D. (2000).
\newblock {\em Linear processes in function spaces: theory and applications}.
\newblock Springer-Verlag New York.

\bibitem[Brezis, 2011]{brezis2011functional}
Brezis, H. (2011).
\newblock {\em Functional analysis, Sobolev spaces and partial differential
  equations}.
\newblock Springer-Verlag New York.

\bibitem[Brockwell and Davis, 1991]{Brockwell1991}
Brockwell, P.~J. and Davis, R.~A. (1991).
\newblock {\em Time Series: Theory and Methods}.
\newblock Springer-Verlag New York.

\bibitem[Bueno-Larraz and Klepsch, 2019]{BuenoLarraz2019}
Bueno-Larraz, B. and Klepsch, J. (2019).
\newblock Variable selection for the prediction of c[0,1]-valued autoregressive
  processes using reproducing kernel hilbert spaces.
\newblock {\em Technometrics}, 61(2):139--153.

\bibitem[Candès and Recht, 2009]{Candes2009}
Candès, E.~J. and Recht, B. (2009).
\newblock Exact matrix completion via convex optimization.
\newblock {\em Foundations of Computational Mathematics}, (9):717–772.

\bibitem[Didericksen et~al., 2012]{Didericksen2012}
Didericksen, D., Kokoszka, P., and Zhang, X. (2012).
\newblock Empirical properties of forecasts with the functionalautoregressive
  model.
\newblock {\em Computational Statistics}, (27):285–298.

\bibitem[Ferraty and Vieu, 2006]{Ferraty2006}
Ferraty, F. and Vieu, P. (2006).
\newblock {\em Nonparametric Functional Data Analysis}.
\newblock Springer-Verlag New York.

\bibitem[Gu, 2013]{Gu2013}
Gu, C. (2013).
\newblock {\em Smoothing Spline ANOVA Models}.
\newblock Springer-Verlag New York, 2 edition.

\bibitem[Hastie et~al., 2009]{Hastie2009}
Hastie, T., Tibshirani, R., and Friedman, J. (2009).
\newblock {\em The Elements of Statistical Learning}.
\newblock Springer-Verlag New York, 2 edition.

\bibitem[Heil, 2018]{Heil2018}
Heil, C. (2018).
\newblock {\em Metrics, Norms, Inner Products, and Operator Theory}.
\newblock Birkhäuser Basel.

\bibitem[Horváth and Kokoszka, 2012]{Horvath2012}
Horváth, L. and Kokoszka, P. (2012).
\newblock {\em Inference for Functional Data with Applications}.
\newblock Springer.

\bibitem[Hyndman et~al., 2008]{Hyndman2008}
Hyndman, R., Koehler, A.~B., Ord, J.~K., and Snyder, R.~D. (2008).
\newblock {\em Forecasting with Exponential Smoothing: The State Space
  Approach}.
\newblock Springer-Verlag Berlin Heidelberg, 1 edition.

\bibitem[Hyndman and Shang, 2009]{Hyndman2009}
Hyndman, R.~J. and Shang, H.~L. (2009).
\newblock Forecasting functional time series.
\newblock {\em Journal of the Korean Statistical Society}, 38(3):199--211.

\bibitem[Hyndman and Ullah, 2007]{Hyndman2007}
Hyndman, R.~J. and Ullah, M.~S. (2007).
\newblock Robust forecasting of mortality and fertility rates: A functional
  data approach.
\newblock {\em Computational Statistics \& Data Analysis}, 51(10):4942--4956.

\bibitem[Hörmann and Kokoszka, 2010]{Hoermann2010}
Hörmann, S. and Kokoszka, P. (2010).
\newblock Weakly dependent functional data.
\newblock {\em Annals of Statistics}, 38(3):1845--1884.

\bibitem[Ji and Ye, 2009]{Ji2009}
Ji, S. and Ye, J. (2009).
\newblock An accelerated gradient method for trace norm minimization.
\newblock In {\em ICML '09: Proceedings of the 26th Annual International
  Conference on Machine Learning}, page 457–464.

\bibitem[Kargin and Onatski, 2008]{Kargin2008}
Kargin, V. and Onatski, A. (2008).
\newblock Curve forecasting by functional autoregression.
\newblock {\em Journal of Multivariate Analysis}, 99(10):2508--2526.

\bibitem[Kokoszka and Reimherr, 2013]{Kokoszka2013}
Kokoszka, P. and Reimherr, M. (2013).
\newblock Determining the order of the functional autoregressive model.
\newblock {\em Journal of Time Series Analysis}, 34(1):116--129.

\bibitem[Kokoszka et~al., 2017]{Kokoszka2017}
Kokoszka, P., Rice, G., and Shang, H.~L. (2017).
\newblock Inference for the autocovariance of a functional time series under
  conditional heteroscedasticity.
\newblock {\em Journal of Multivariate Analysis}, 162:32--50.

\bibitem[Koltchinskii and Yuan, 2010]{koltchinskii2010sparsity}
Koltchinskii, V. and Yuan, M. (2010).
\newblock Sparsity in multiple kernel learning.
\newblock {\em Annals of Statistics}, 38(6):3660--3695.

\bibitem[Lütkepohl, 2005]{Luetkepohl2005}
Lütkepohl, H. (2005).
\newblock {\em New Introductionto MultipleTime Series Analysis}.
\newblock Springer-Verlag Berlin Heidelberg.

\bibitem[Mendelson, 2002]{mendelson2002geometric}
Mendelson, S. (2002).
\newblock Geometric parameters of kernel machines.
\newblock In {\em International Conference on Computational Learning Theory},
  pages 29--43. Springer.

\bibitem[Nickl and P{\"o}tscher, 2007]{nickl2007bracketing}
Nickl, R. and P{\"o}tscher, B.~M. (2007).
\newblock Bracketing metric entropy rates and empirical central limit theorems
  for function classes of besov-and sobolev-type.
\newblock {\em Journal of Theoretical Probability}, 20(2):177--199.

\bibitem[Ramsay and Silverman, 2005]{Ramsay2005}
Ramsay, J. and Silverman, B.~W. (2005).
\newblock {\em Functional Data Analysis}.
\newblock Springer-Verlag New York.

\bibitem[Raskutti et~al., 2012]{raskutti2012minimax}
Raskutti, G., Wainwright, M.~J., and Yu, B. (2012).
\newblock Minimax-optimal rates for sparse additive models over kernel classes
  via convex programming.
\newblock {\em Journal of Machine Learning Research}, 13(Feb):389--427.

\bibitem[Shang, 2013]{Shang2013}
Shang, H.~L. (2013).
\newblock Functional time series approach for forecasting very short-term
  electricity demand.
\newblock {\em Journal of Applied Statistics}, 40(1):152–168.

\bibitem[Sun et~al., 2018]{sun2018optimal}
Sun, X., Du, P., Wang, X., and Ma, P. (2018).
\newblock Optimal penalized function-on-function regression under a reproducing
  kernel hilbert space framework.
\newblock {\em Journal of the American Statistical Association},
  113(524):1601--1611.

\bibitem[Wahba, 1990]{Wahba1990}
Wahba, G. (1990).
\newblock {\em Spline Models for Observational Data}.
\newblock SIAM, Philadelphia.

\bibitem[Wong et~al., 2017]{wong2017lasso}
Wong, K.~C., Li, Z., and Tewari, A. (2017).
\newblock Lasso guarantees for beta-mixing heavy tailed time series.
\newblock {\em arXiv preprint arXiv:1708.01505}.

\bibitem[Yuan and Cai, 2010]{Yuan2010}
Yuan, M. and Cai, T.~T. (2010).
\newblock A reproducing kernel {Hilbert} space approach to functional linear
  regression.
\newblock {\em Annals of Statistics}, 38(6):3412--3444.

\end{thebibliography}
\appendix

Appendix A-E contains technical proofs of theorems in the main text. Appendix F gives the implementation details of the accelerated gradient method~(AGM). Appendix G contains additional simulation results.

\section{Results related to linear compact operators $:\h \to \h$}
In this section, we give some properties of the linear compact operator $A :\h \to \h$ that will be used later in the proof. \Cref{lemma:integral operator well defined} shows that a linear compact operator is also bounded operator $A : \lt\to\lt.$
\begin{lemma}\label{lemma:integral operator well defined}
Suppose $A:\mathcal H \to  \mathcal H   $ is a compact operator. Then $$ \| A\|_{\lt,\op} :=  \sup_{\|v\|_{\lt } \le 1, \|u\|_{\lt } \le 1} \iint  A(r,s)u(r) v(s)  drds  \le \mu_1  \|A\|_{\h, \op}.$$
\end{lemma}
\begin{proof} Let $ u =\sum_{i=1}^\infty \alpha_i\phi_i$ and $ v =\sum_{i=1}^\infty \beta_i\phi_i$ and  that
	$ \|u\|_{\lt} \le 1$, $ \|v\|_{\lt}  \le 1$.
	Observe that by \eqref{eq:expression of A}, 
	\begin{align*}
	\iint  A(r,s) u(r) v(s)  dr ds   =
	&    \sum_{i,j=1}^\infty a_{ij}   \sqrt { \mu_i \mu_j }   \sum_{k=1}^\infty \alpha_k\lb \phi_i, \phi_k   \rb_{\lt}   \cdot  
	\sum_{l=1}^\infty \beta_l  \lb   \phi_j,\phi_l \rb_{\lt}  
	\\
	=& \sum_{i,j=1}^\infty a_{ij}  \sqrt {\mu_{i}}  \alpha_i   \sqrt {\mu_{j}} \beta_j 
	\\
	\le & \sup_{ \sum_{ i=1}^ \infty c_i^2 \le \mu_1, \sum_{ j=1}^ \infty d_j^2 \le \mu_1  } \sum_{i,j=1}^\infty a_{ij}  c_i d_j
	\\   
	\le & \mu_1\| A\|_{\h  , \op } ,
	\end{align*}
	where the last inequality follows from the definition in \eqref{eq:bounded bivariate operator}  that 
	$$ \| A\|_{\h ,  \op}  = \sup_{\sum_{ i=1}^ \infty c_i^2 \le 1, \sum_{ j=1}^ \infty d_j^2 \le  1 }  \sum_{i,j=1}^\infty a_{ij}  c_i d_j .$$
\end{proof}    
\begin{remark}\label{remark:mu1}
We note that from equation \eqref{eq:eigen expansion of kernel}, it holds that 
$$ C_{\mathbb K}  \ge \int \mathbb K(r,r) dr  = \sum_ {k=1}^\infty  \mu_k \|\phi_k \|_\lt ^2 =  \sum_ {k=1}^\infty  \mu_k.  $$
Therefore $\mu_1\le C_{\mathbb K} = 1$. We will use these inequalities repeatedly in our analysis.    
\end{remark}


\begin{theorem}\label{lemma:low rank expansion}
	Suppose $A:\mathcal H \to  \mathcal H   $ is a compact operator. Then there exist two collections of sub-basis $ \{ \psi_k\}_{k=1}^\infty $ and  $ \{ \omega_k\}_{k=1}^\infty $  in $\mathcal H$,
	such that 
	$$  A (r,s) = \sum_{k=1}^\infty  a_k \psi_k(s) \omega_k(r).   $$ 
	Suppose in addition,  $\text{rank} (A) \le K$. Then $$  A (r,s) = \sum_{k=1}^K a_k \psi_k(s) \omega_k(r). $$ 
	
\end{theorem}
\begin{proof} This is the well known spectral theory for compact operators on Hilbert space. See Chapter 5 of \cite{brezis2011functional} for a detailed proof.
\end{proof}

\begin{lemma} \label{lemma:bound of A 1}
	Let $A : \mathcal H \to \h $ be any compact linear operator. Then 
	$$ \max \left \{  \sup_{r\in [0,1]} \| A(r, \cdot  )   \|_\h , \sup_{s\in [0,1]} \| A(\cdot ,s )   \|_\h , \sup_{r,s \in[0,1] } |A(r,s) |   \right\} \le \|  A \|_{\h, *} .$$
\end{lemma}
\begin{proof}
	By  \Cref{lemma:low rank expansion}, 
	$$  A (r,s) = \sum_{k=1}^\infty  a_k \psi_k(r) \omega_k(s). $$
	Therefore 
	$  \| A\|_{\h, *} = \sum_{k=1}^\infty   | a_k|    $
	and that 
	\begin{align*}  
	\sup_{r\in [0,1] }\| A(r,\cdot )   \|_\h  \le  \sum_{k=1}^\infty   \sup_{r\in [0,1] }  |  a_k \psi_k(r) |  \| \omega_k \|_\h  \le    \sum_{k=1} ^\infty    |  a_k  | \| \psi_k \|_\h \le  \sum_{k=1} ^ \infty   | a_k|  \le  \| A\|_{\h, *} ,
	\end{align*}
	where $\| \psi_k\|_{\infty} \le \| \psi_k\|_\h \le 1  $ is used in deriving the inequality.  Similar argument shows that 
	\begin{align*} \sup_{s\in [0,1] } \| A(\cdot ,s )   \|_\h  \le  \| A\|_{\h, *} .
	\end{align*}
	For the last part of the inequality, observe that for any fixed $r\in[0,1]$,
	$$ A(r, s) =  \lb A(r,\cdot), \mathbb K_s(\cdot)  \rb_\h .$$
	Therefore
	$$ \sup_{s\in [0,1] }|A(r, s )| \le   \sup_{s\in [0,1] } |  \lb A(r,\cdot), \mathbb K_s(\cdot)  \rb_\h  |
	\le \| A(r,\cdot)\|_\h  \sup_{s\in [0,1] }  \| \mathbb K_s(\cdot)  \|_\h \le \| A(r,\cdot)\|_\h , $$
	where  $ \|  \mathbb K_s(\cdot) \|_\h^2 =\mathbb K(s,s ) \le 1$ is used in the last inequality.
	Therefore
	$$\sup_{r, s\in [0,1] }|A(r, s )|   \le \sup_{r\in [0,1] }\|A(r, \cdot )\|_\h  \le   \|A\|_{\h,*}.$$
\end{proof} 

\section{Proof of \Cref{example:farD}}
\begin{proof}[Proof of  \Cref{example:farD}]
	\Cref{example:farD} directly follows from  \Cref{lemma:fard bounded X} and  \Cref{eq:fard restricted eigenvalue}. Specifically, \Cref{lemma:fard bounded X} proves the stationarity and boundedness of $\{X_t\}$ and \Cref{eq:fard restricted eigenvalue} proves the restricted eigenvalue condition \eqref{eq:FAR restricted eigenvalue condition} of $\{X_t\}$.
\end{proof} 


We start with some general definitions and results for functional time series from \cite{bosq2000linear}.

\begin{definition}
Let $\mathbb B : \h \to\h $ be any linear operator. Define 
$$ \|\mathbb B \|_{\h \to \h } :  = \sup_{\|f\|_\h \le 1 ,  \|g\|_\h  \le 1} \lb  \mathbb B(f) ,g \rb_{\h} .$$ 
$\mathbb B$ is said to be bounded if $\|\mathbb B \|_{\h \to \h } < \infty $.
\end{definition}
\begin{theorem} \label{theorem:bosq}
Let $\{ \mathbb A_d\}_{d=1}^D $ be a collection of bounded linear  operators from $\h \to \h$. Suppose that 
$\{ X_t\}_{t=-\infty}^\infty$ and that $\{\epsilon_t\}_{t=-\infty} ^ \infty $ are two collections of functions in $\h$ such that
 $ X_t = \sum_{d=1}^D \mathbb A_d  (X_{t-d}) +\epsilon_t.  $ 
Suppose in addition that  \begin{align} \label{spetrum bounded by 1}
\sup _{|z| \le 1 ,z \in \mathbb C   } \left \|  \sum_{d=1}^D z^d  \mathbb A_d   \right   \|_{\h \to \h  }  =\gamma  < 1,
\end{align} 
then there exists a {unique} collection of $\{ \mathbb B_{i}\}_{i=1}^\infty$ being operators from   $\h \to \h$ such that 
$ X_t = \sum_{i=0}^\infty \mathbb  B_i  (\epsilon_{t-i })   $ and that 
$$ \sum_{i=0}^\infty \| \mathbb B_i \|_{\h \to \h  }  \le \frac{1}{1-\gamma}. $$
\end{theorem}
\begin{proof}
The proof of the theorem follows immediately from Theorem 5.1 and Theorem 5.2 of  \cite{bosq2000linear} and thus is omitted.
\end{proof}

\Cref{lemma:norms in stationary process} is used in the proof of \Cref{lemma:fard bounded X}.
\begin{lemma} \label{lemma:norms in stationary process}Let $B(r,s)$ be any bivariate functions on $[0,1]\times [0,1]$ such that 
\begin{align}\label{eq:expression of B}
B(r,s) =     \sum_{i,j=1}^\infty b_{ij}  \Phi_i(r) \Phi_j(s),
\end{align}
where $ \{\Phi_i\}_{i=1}^\infty$  are the eigen-basis of $\mathbb K$ as in \eqref{eq:expression of A}. 
For any $f\in \h$, let $\mathbb B$ denote the operator from $\h \to \h$ such that 
 $$\mathbb B (f)(\cdot)   : = \int B(\cdot ,s ) f (s) ds.  $$
Then it holds that 
$$ \| \mathbb B \|_{\h\to \h  }\le \mu_1 \| B\|_{\h, \op}.  $$
\end{lemma}
\begin{proof}
Since   $\{ \Phi_i\}_{i=1}^\infty $  are orthonormal basis of $\h$, 
for any $f\in \h$ such that $\|f\|_\h =1$, it holds that 
$ f= \sum_{i=1}^\infty c_i \Phi_i $ with $\sum_{i=1}^\infty c_i^2 =1$. 
\begin{align*}
\| \mathbb B \|_{\h \to \h } =     & \sup_{\|f\|_\h \le 1 , \|g\|_\h  \le 1 } \lb   \mathbb B(f) ,g \rb_{\h} 
\\
= 
&
\sup_{\|f\|_\h \le 1 , \|g\|_\h  \le 1 } \lb   \int B(\cdot ,s ) f (s) ds   ,g (\cdot ) \rb_{\h}  
\\
=
&
 \sup_{\|f\|_\h \le 1 , \|g\|_\h  \le 1 }    \sum_{i,j=1}^\infty b_{ij}  \lb  \Phi_i, g\rb _\h  \lb \Phi_j, f\rb_\lt 
 \\
 = & \sup_{\sum_{k=1}^\infty  c_k ^2 \le 1 ,  \sum_{l=1}^\infty  d_l^2   \le 1 }   
  \sum_{i,j=1}^\infty b_{ij}    \lb  \Phi_i, \sum_{l=1}^\infty d_l \Phi_ l  \rb _\h  \lb \Phi_j, \sum_{k=1}^\infty c_k \Phi_ k\rb_\lt 
  \\
  = & 
  \sup_{\sum_{k=1}^\infty  c_k ^2 \le 1 ,   \sum_{l=1}^\infty  d_l^2   \le 1 }   
  \sum_{i,j=1}^\infty b_{ij}    d_i \mu_j c_j 
  \\
  \le &  \sup_{\sum_{k=1}^\infty  (c_k') ^2 \le   \mu_1 ^2    , \sum_{l=1}^\infty  d_l^2   \le   1 }   
  \sum_{i,j=1}^\infty b_{ij}    d _i   c'_j.
\end{align*}
Since  \eqref{eq:bounded bivariate operator} gives
 $$ \| B\|_{\h, \op} =  \sup_{\sum_{k=1}^\infty   c_k ^2   \le   1    , \sum_{l=1}^\infty  d_l^2    \le   1 }   
  \sum_{i,j=1}^\infty b_{ij}    d _i   c _j , $$
  the desired result immediately follows. 
\end{proof}

\begin{lemma}\label{lemma:fard bounded X}
	Under the conditions in \Cref{example:farD},  there is a unique stationary solution $\{X_t\}_{t=-\infty}^\infty$ to \eqref{eq:AR model} and 
	$$ \| X_t\|_\h \le    \frac{C_\epsilon}{ 1- \gamma_A} .$$  
\end{lemma}

\begin{proof}  For any $z\in \mathbb C$, let 
$B ( r,s)  : =  \sum_{d=1}^D z^d A_d^*(r,s) $ and let $\mathbb B $ be the operator   such that 
$$ \mathbb B  (f) (\cdot)  :  = \int  B  (\cdot,  s) f(s) ds . $$ 
Then from \Cref{lemma:norms in stationary process}, it holds that 
\begin{align} 
\label{eq:complex norm}\|\mathbb B  \| _{\h \to \h } \le \mu_1 \|B \|_{\h ,\op} \leq \|B \|_{\h ,\op} , 
\end{align}
where the last inequality follows from \Cref{remark:mu1}. 
Denote 
$$ \mathbb A_d  (f) (\cdot) : =  \int A_d^*(\cdot, s) f(s)ds.  $$ Then
$$\mathbb B  = \sum_{d=1}^D z^d  \mathbb A_d . $$
Therefore  the above  equality and  \eqref{eq:complex norm} imply that
$$\left \|  \sum_{d=1}^D z^d  \mathbb A_d   \right   \|_{\h \to \h  }  \le\left \|  \sum_{d=1}^D z^d  A^*_d   \right   \|_{\h,\op  }   . $$
By assumption, $ \sup _{|z| \le 1 ,z\in \mathbb C}\left \|  \sum_{d=1}^D z^d  A^*_d   \right   \|_{\h,\op  }  \le  \gamma_A<1.   $
Therefore 	by \Cref{theorem:bosq},   there exists a {unique} collection of operators $\{ \mathbb B_i\}_{i=1}^\infty$ such that
$$ X_t = \sum_{i=0}^\infty \mathbb B_{i} (\epsilon_{t-i})$$ 
and that 
$$ \sum_{i=0}^\infty \| \mathbb B _i \|_{\h \to \h   }  \le  \frac{1}{ 1- \gamma_A}  .$$
Therefore with probability $1$,
$$ \| X_t\|_\h \le \sum_{i=0}^\infty\|B _i \|_{\h  \to \h  }  \| \epsilon_{t-i}\|_\h \le   \frac{C_\epsilon}{ 1- \gamma_A} .$$  
\end{proof}

The following definition is used throughout the proof in the Appendix.
\begin{definition}
For any bounded bivariate function $B(r,s) : [0,1]\times [0,1] \to \mathbb R, $
define 
\begin{align*} &\col (B) : = \{ u \in \lt :      u(\cdot) = \int B (\cdot, s) w(s) ds \text{ for some } w \in \lt  \} ,
\\
&
 \row(B) : = \{ u \in \lt :      u(\cdot) = \int B (s, \cdot ) w(s) ds \text{ for some } w \in \lt  \}.
\end{align*} 
\end{definition}

\begin{lemma} \label{eq:fard restricted eigenvalue}
	Under the conditions in \Cref{example:farD}, it holds that 
	$$ E \left ( \int  \sum_{d=1}^D    v_d(s) X_{t -d}  (s) ds\right)^2  \ge \frac{\kappa_\epsilon}{ \left( 1 - \gamma_A\right) ^2  }  \sum_{d=1}^D\|v_d\|_{  \lt }^2  \ \text{ for all }   \{v_d \}_{d=1}^D\subset \h . $$
\end{lemma}

\begin{proof}  Let  
	$S$  be a subspace  of $\lt$ such that 
	$$  S \supset \s \{ \col(A^*_d)\cup   \row(A^*_d)  \}_{d=1}^D .  $$
	Observe that  \Cref{lemma:low rank expansion}  together with $\text{rank}(A_d^*) < \infty $ directly implies that the dimensions of $\col(A^*_d)$  and $\row(A^*_d)$ are finite. 
	Let $\{ w_i\}_{i=1}^N$ be the orthonormal sub-basis in $\lt$  of $S$.
	\\
	\\
	{\bf Step 1.} 
	Let $\{ a_d\}_{d=1}^D$ be a collection of matrices in $\mathbb R^{N \times N}$ such that 
	$$ a_{d} (i,j) = \iint A_d^* (r,s) w_i(r) w_j(s) dr ds   .$$
	Since $S$ contains $ \col(A^*_d)\cup   \row(A^*_d)$, it holds that  
	$$ A_d^* (r,s)  = \sum_{i,j =1}^N a_{d} (i,j) w_i(r) w_j(s)  $$ 
	In addition for $i=1,\cdots,N$, let 
	$$y_{t} (i)  = \int X_t(s) w_i (s) ds  \text{ and } \varepsilon_{t}(i)  =   \int \epsilon_t(s) w_i (s) ds. $$
	 Since 
	  \begin{align*}   
	X_{t} (\cdot )  = \sum_{d=1}^D   \int A_d ^*(\cdot ,s) X_{t-d} (s) ds     +\epsilon_{t} (  \cdot ) ,
	\end{align*} 
	it holds that for all $1\le i\le N$, 
 \begin{align}  \label{eq:FARD assumption}
	 \lb X_{t} (\cdot )  , w _i(\cdot)\rb _\lt   = \lb  \sum_{d=1}^D   \int A_d ^*(\cdot ,s) X_{t-d} (s) ds,   w _i(\cdot)\rb _\lt     
	  +\lb	 \epsilon_{t} (  \cdot ) ,w _i(\cdot)\rb _\lt    
	\end{align}  
	Therefore 	$$    y_{t } (i)  =    \left(  \sum_{j=1}^N a_d(i,j) y_{t-d} (j)  \right)    +  \varepsilon_{t}(i)   . $$
	Then it holds that  
	$$ y_{t}  = \sum_{d=1}^D a_d y_{t-d}  + \varepsilon_{t}$$ 
	and thus $\{y_t\}_{t=1}^T $ is a VAR(D) process in $\mathbb R^N$.   
	\\
	\\
	{\bf Step 2.}
	By \Cref{lemma:integral operator well defined}, 
	$$ \left \|  \sum_{d=1}^D z^d  A_d^*   \right   \|_{\lt , \op} \le  \left \|  \sum_{d=1}^D z^d  A_d^*   \right   \|_{\h, \op}     . $$
	Therefore 
	$$ \sup_{|z| \le 1 }\left \|  \sum_{d=1}^D z^d  A_d^*   \right   \|_{\lt , \op} 
	\le \sup_{|z| \le 1 }\left \|  \sum_{d=1}^D z^d  A_d^*   \right   \|_{\h , \op}  \le \gamma_A      .$$
	Let  $ \alpha  \in \mathbb R^ N$  be such that $\|  \alpha \|_2 =1$.  
	Denote $ \alpha  = [ \alpha_{ 1},\ldots, \alpha _{ N}]$ and suppose that 
	$u  =\sum_{i=1}^N \alpha _{ i}w_i$, where $\{ w_i\}_{i=1}^N$ is sub-basis in $\lt$  of $S$. So 
	$\|u\|_\lt =1$. 
	Then by the definition of $\{ w_i\}_{i=1}^N$, it holds that 
	\begin{align*}
	\left\| \left (\sum_{d=1}^D z^d  a_d   \right )  \alpha   \right\|_2 ^2 
	=  \left\| \left (\ \sum_{d=1}^D z^d  A_d^*   \right )  u  \right\|_\lt  ^2  
	\le \left \|  \sum_{d=1}^D z^d  A_d^*   \right   \|_{\lt , \op} ^2 \|u\|_\lt ^ 2  = \left \|  \sum_{d=1}^D z^d  A_d^*   \right   \|_{\lt , \op} ^2.  
	\end{align*}
	The above display implies that 
	$$ \left \|  \sum_{d=1}^D z^d  a_d    \right   \|_{  \op}  \le \left \|  \sum_{d=1}^D z^d  A_d^*   \right   \|_{\lt , \op}   $$
	and so 
	\begin{align}
	\label{eq:invertible operators} \sup_{|z| \leq 1 }\left \|  \sum_{d=1}^D z^d  a_d   \right   \|_{ \op}  <\gamma_A < 1.
	\end{align}
	\
	\\
	{\bf Step 3.} By \Cref{lemma:subgaussian restricted eigenvalue}, it holds that for any $\{ \beta_d\}_{d=1}^D  \subset \mathbb R^N$,
	\begin{align}\label{eq:restricted eigenvalue in projected space}
	E  \left ( \sum_{d=1}^D y_{t -d} ^\top   \beta _d  \right) ^2 \ge  \frac{\kappa_\varepsilon}{(1-\gamma)^2}    \sum_{d=1}^D \|\beta_d\|_2^2 
	\end{align}
	where  
	$$ 
	\gamma := \sup_{|z| \le 1 }\left \|  \sum_{d=1}^D z^d  a_d    \right   \|_{\op} \quad \text{and} \quad 
	\kappa_\varepsilon : = \inf_{\beta \in \mathbb R^N, \| \beta\|_2=1 } E(\varepsilon_t^\top \beta)^2
	.$$  
	Then 
	\eqref{eq:restricted eigenvalue in projected space}
	implies that  for any $\{u_d\}_{d=1}^D \subset S $,
	\begin{align}\label{eq:restricted eigenvalue in subspace}
	E \left ( \int  \sum_{d=1}^D    u_d(s) X_{t -d}  (s) ds\right)^2  \ge \frac{\kappa_\varepsilon}{(1-\gamma)^2}  \sum_{d=1}^D\|u_d\|_{  \lt }^2  .  \end{align}
	Note that 
	$$ \kappa_\varepsilon  \ge  \inf_{v \in \lt, \| v \|_\lt =1 } E(\lb \epsilon_t, v\rb_\lt )^2 \ge \kappa_\epsilon   $$
	and that  \eqref{eq:invertible operators} gives  that $\gamma\le \gamma_A$. 
	So  
	\begin{align} 
	E \left ( \int  \sum_{d=1}^D    u_d(s) X_{t -d}  (s) ds\right)^2  \ge \frac{\kappa_\epsilon}{ \left( 1 - \gamma_A\right) ^2  }    \sum_{d=1}^D\|u_d\|_{  \lt }^2  .  
	\end{align}
	\
	\
	\
	\\
	{\bf Step 4.} 
	Let
	$$S =  \s  \{ \col(A^*_d)\cup   \row(A^*_d)  \}_{d=1}^D  \cup \{v_d\}_{d=1}^D. $$
	Then the desired results follows  immediate from \eqref{eq:restricted eigenvalue in subspace}.
\end{proof}

\begin{lemma} \label{lemma:subgaussian restricted eigenvalue}
	Let $\{ y_t\}_{t=-\infty }^\infty \subset \mathbb R^N$ be a     VAR(D) process such that such that 
	$$ y_{t}  = \sum_{d=1}^D a_d y_{t-d}  + \varepsilon_{t}  $$
	where $\{a_d \}_{d=1}^D \subset \mathbb R^{N\times N}$ and $\{\varepsilon_t\}_{t=-\infty}^\infty $ are i.i.d. sub-Gaussian random variables. Suppose in addition that there exist two constants
	$ 0<\gamma<1$ and $\kappa_{\varepsilon}>0$ where
	$$ 
	\gamma := \sup_{|z| \leq 1 }\left \|  \sum_{d=1}^D z^d  a_d    \right   \|_{\op}  <1 \quad \text{and}   \quad 
	\kappa_\varepsilon : = \inf_{\beta \in \mathbb R^N, \| \beta\|_2=1 } E(\varepsilon_t^\top \beta)^2
	.$$  
	Then $ \{ y_t\}_{t=-\infty}^\infty$ is stationary and invertible and it holds that for any $\{\beta_d\}_{d=1}^D \subset \mathbb R^N$
	\begin{align} \label{eq:restricted eigenvalue in subspace 2}
	E  \left ( \sum_{d=1}^D y_{t-d} ^\top   \beta _d  \right) ^2 \ge  \frac{\kappa_\varepsilon}{(1-\gamma)^2}    \sum_{d=1}^D \|\beta_d\|_2^2 .
	\end{align}
\end{lemma}
\begin{proof}
	Let $\{Y_{t}\}_{t=-\infty }^\infty \subset \mathbb R^{ND} $ be defined as 
	$$ Y_{t} =  [y_{t}^\top , y_{t-1}^\top , \ldots, y_{t-D+1}^\top ]^\top . $$ 
	Consider 
	\begin{align}
	B =\begin{bmatrix}
	a_{1} & a_2 &\ldots & a_{D-1} & a_D \\
	I_N   &    0 & \ldots & 0 & 0\\
	0 & I_N & \ldots & 0 &0 \\
	\vdots & \vdots & \ddots &   \vdots & \vdots \\
	0 & 0 & \ldots & I_N &  0
	\end{bmatrix}   \in \mathbb R^{ND \times ND}.
	\end{align}
	 It is well known (see for example \cite{basu2015regularized} and reference therein) that under the conditions on $ \gamma$ and $\kappa_\varepsilon$, $\{ Y_t\}_{t=-\infty }^\infty   $ is a stationary  and invertible VAR(1) process such that for any $w\in \mathbb R^{ND}$, it holds that 
	$$E(Y_t ^\top w)  \ge  \frac{\kappa_\varepsilon}{(1-\gamma)^2}\|w\|_2^2.$$
	This directly implies the desired result in \eqref{eq:restricted eigenvalue in subspace 2}.
	
\end{proof}

\section{Results Related to Sobolev Spaces}
\subsection{Bounds for $\gamma_n$}

Let $\{s_i \}_{i=1}^n $ be a collection of uniform random variables sampled from $[0,1]$ and $\{ \sigma_i\}_{i=1}^n$ is a collection of Rademacher random variables. The following theorem is Theorem 2.1  in \cite{bartlett2005local}, which is used for bounding $\gamma_n.$
\begin{theorem}\label{lemma:talagrand}
	Suppose $\mathcal F$ is a class of function that map $[0,1]$ into $[-b,b] $. Then for any $\delta  >0$, it holds that with probability $1-\exp(-c_\mathcal R \delta)$, 
	$$ \sup_{f\in \mathcal F _\alpha   }  \left (\int  f(s) ds  - \frac{1}{n} \sum_{i=1}^ n  f(s_i)  \right )   \le C_\mathcal R  \left ( E( \mathcal R_n \mathcal F_\alpha  )   +  \alpha  \sqrt { \frac{\delta  }{ n }} + b  \frac{\delta  }{ n }   \right)  $$
	where  $c_\mathcal R$, $C_\mathcal R$ are absolute constants,   $\mathcal F_\alpha  =  \{f\in \mathcal F : \| f\|_\lt  \le \alpha \} $, and 
	$$ \mathcal R_n \mathcal F_\alpha    : = \sup_{f \in \mathcal F_\alpha   }   \mathcal R_n f  \ \text{ and }  \  \mathcal R_n f  =  \frac{1}{n} \sum_{i=1}^ n \sigma_i f(s_i)  .$$
\end{theorem} 
\
\\
To use \Cref{lemma:talagrand} for bounding $\gamma_n$, we define 
$$\zeta_n :  = \inf  \left \{\zeta  \ge \sqrt { \frac{  \log(n)}{ n }}  :  E \sup_{ \|f  \|_{ \h  } \le b, \|  f \|_\lt \le \delta  } \mathcal R_n f    
\le  \zeta  \delta +b  \zeta^2  \text{ for all  } \delta \in (0,1]\right \} . $$ 

\Cref{lemma:restricted eigenvalue 1} provides the probability bound for $\gamma_n'$ using $\zeta_n.$
\begin{lemma} \label{lemma:restricted eigenvalue 1}
	Suppose $ b $ is any bounded constant. 
	Then
	it holds that 
	$$ P\left(  \left |  \int  f(s) ds  - \frac{1}{n} \sum_{i=1}^ n  f(s_i)   \right |    \le C_\zeta  (  \zeta_n  \|f\|_\lt     + (b+1)  \zeta_n^2   ) 
	\text{ for all } f \text{ such that }  \| f\|_{\h }\le b  
	\right)    \ge 1 -  1/n^4.  $$
\end{lemma}
\begin{proof} 
	It suffices to show that 
	$$ P\left(  \left |  \int  f(s) ds  - \frac{1}{n} \sum_{i=1}^ n  f(s_i)   \right |    \le C_\zeta   (  \zeta_n  \|f\|_\lt    +( b +1)  \zeta_n^2   ) 
	\text{ for all } f \text{   that }  \| f\|_{\h }\le b  ,  \| f\|_{\lt }   \le 1 
	\right)    \ge 1 -  1/n^4.  $$
	Let $\delta  = C_\delta \log(n)$ such that  $ \exp(-c_\mathcal R \delta) \le  n^{-5}$, where 
	$c_{\mathcal R}$ is defined as in \Cref{lemma:talagrand}. Let 
	$J \in \mathbb Z^+$ be such that 
	$$2^{-J } \le  \zeta_n  \le 2^{-J+1}. $$ 
	So $J \le \log(n)$.  For any 
	$ 1\le j \le J$,  it holds that with probability at least $1-\exp(-c_\mathcal R\delta ) \ge 1-n^{-5}$,  for all $f$ such that$ \| f\|_{\h} \le b ,  
	\  2^{-j} \le \| f\|_\lt \le 2^{-j+1} $
	\begin{align*}
	\left |  \int  f(s) ds  - \frac{1}{n} \sum_{i=1}^ n  f(s_i)   \right |   
	\le  & C _\mathcal R \left ( E   \sup_{ \|f  \|_{ \h  } \le b, \|  f \|_\lt  \le 2^{-j+1 }   } \mathcal R_n f     +  2^{-j+1 }  \sqrt { \frac{  \delta    }{ n }} + b  \frac{\delta   }{ n }   \right) 
	\\
	\le  & C _\mathcal R\left (  \zeta_n 2^{-j+1}  + b \zeta_n^2  +  2^{-j+1 }  \sqrt { \frac{  \delta     }{ n }} + b  \frac{\delta   }{ n }   \right) 
	\\
	\le &  C_\mathcal R  \left ( 2  \zeta_n \|f\|_\lt    +  b   \zeta_n ^2  +2 C_\delta \zeta_n \| f\|_\lt  + b C_\delta^2 \zeta_n^2 \right) ,
	\end{align*} 
	where the first inequality follows from  \Cref{lemma:talagrand}, the second inequality follows from definition of $\zeta_n$ and the last inequality follows from $\| f\|_\lt \ge 2^{-j}  $ and the fact that $C_\delta^2 \zeta_n^2 \ge \delta/n $. Therefore 
	with probability at least $1- J n^{-5}  \ge 1-  \log(n)n^{-5} $
	\begin{align*}
	\sup_{  f \in \mathcal H , \|f\|_\h \le b  ,  
		\   \| f\|_\lt \ge 2^{-J  }   } \left |  \int  f(s) ds  - \frac{1}{n} \sum_{i=1}^ n  f(s_i)   \right |   
	\le  &  C_1' \left (\zeta_n   \| f\|_\lt   + b   \zeta_n ^2    \right) . 
	\end{align*} 
	In addition,  by \Cref{lemma:talagrand} with probability at least $1- n^{-5 }$, 
	\begin{align*}
	\sup_{f \in \mathcal H , \|f\|_\h \le b  ,   
		\   \| f\|_\lt \le 2^{-J  }   } \left |  \int  f(s) ds  - \frac{1}{n} \sum_{i=1}^ n  f(s_i)   \right |  
	\le  & C_ \mathcal R  \left ( \zeta_n    2^{-J}   + b \zeta_n^2      \right)   \le  
	C_\mathcal R  (1+b) \zeta_n^2  ,
	\end{align*}  
	where the  last inequality follows from the choice that $2^{-J} \le \zeta_n$.  Therefore 
	it suffices to choose 
	$$C_\zeta =\max\{ 4C_\mathcal R (1+ C_\delta) , 4C_\mathcal R C_\delta^2 \}. $$
\end{proof}

\Cref{lemma:restricted eigenvalue 2} provides the probability bound for $\gamma_n''$ using $\zeta_n.$
\begin{lemma} \label{lemma:restricted eigenvalue 2}Suppose $b$ is any bounded constant. 
	Then    
	it holds that 
	\begin{align*}
	&  P\left ( \|f \|_\lt ^2  \le  2 \|f \|_n^2 +C_\zeta'     b^2  \zeta_n^2   \text{ for all } f \text{ such that } \| f\|_\h \le  b   \right)  \ge 1-1/n^4;
	\\
	&  P\left ( \|f \|_n^2  \le  2 \|f \|_\lt ^2 +C_\zeta'    b^2  \zeta_n^2   \text{ for all } f \text{ such that } \| f\|_\h \le  b   \right)  \ge 1-1/n^4,
	\end{align*}
	where $ \|f \|_n^2  = \frac{1}{n}  \sum_{i=1}^n f^2(s_i)$
\end{lemma}
\begin{proof}
{\bf Step 1.} It suffices to  show that 
	$$ P\left ( \|f \|_\lt ^2  \le  2 \|f \|_n^2 +C' _\zeta     b^2  \zeta_n^2   
	\text{ for all } f \text{ such that } \| f\|_\h \le  b  , \| f\|_\lt \le 1   \right).$$  
	This is because for any $g\in \h$ such that $\|g\|_\lt> 1 $ and that $\| g \|_\h \le b  $,  one can apply the above probability bounds to the function $ h: =\frac{ g}{\|g\|_\lt } $ and observe that 
	$\|h\|_\lt \le 1$, $\|h\|_\h\le \frac{ b}{\| g\|_\lt }$. 
	\\
	\\
	{\bf Step 2.}
	Observe that the function $\phi(x) = x^2/b $ is a contraction on the interval $[-b,b]$. Therefore 
	$$E  ( \mathcal R_n ( \phi \circ  \mathcal F  ))  \le E( \mathcal R_n   \mathcal F  ) .   $$
	The  same calculations  in  the proof of  \Cref{lemma:restricted eigenvalue 1} shows that 
	$$ P\left( \left |  \int  f^2(s) ds  - \frac{1}{n} \sum_{i=1}^ n  f^2(s_i)   \right |   
	\le C \left ( b \zeta_n \|f\|_\lt      +   b^2 \zeta_n^2 \right)  \text{ for all } f \text{   that } \| f\|_\lt \le 1 , \|f\|_\h \le b\right)    \ge 1 -  1/n^4, $$ 
	where $C$ only depends on $C_\mathcal R$ and $c_\mathcal R$. 
	Since
	$$ b \zeta_n \|f\|_\lt      +   b^2 \zeta_n^2 \le \frac{1}{2}\|f \|^2_\lt +3  b^2 \zeta_n^2 , $$
	The desired result follows immediately. 
\end{proof}

Based on \Cref{lemma:restricted eigenvalue 1} and \Cref{lemma:restricted eigenvalue 2}, we provide a probability bound for $\gamma_n$ in \Cref{corollary:explicit rate}.
\begin{corollary}\label{corollary:explicit rate}
	Suppose that $\h = W^{\alpha,2 } $ and  that $\{s_i \}_{i=1}^n $  is a collection of uniform random variables sampled from $[0,1]$. Let $\gamma_n'$ and $\gamma_n''$ be defined as in \eqref{eq:gamma 1} and \eqref{eq:gamma 2} respectively.    Then with probability at least $1-1/n^4 $, 
	it holds that 
	$$  \max \{ \gamma_n',\gamma_n''\} \le C_\alpha n^{-\alpha/(2\alpha+1) }, $$
	where $C_\alpha$ is some constant independent of $n$. 
\end{corollary}  
\begin{proof}
	Suppose   $\mathcal H  =  W^{ \alpha,2 }    $. Then from \cite{mendelson2002geometric},   it  holds that $\zeta_n \le C_\alpha  n^{-\alpha /(2\alpha+1) }$.   The desired results follow directly from 
	\Cref{lemma:restricted eigenvalue 1} and \Cref{lemma:restricted eigenvalue 2}.
\end{proof}

\subsection{Bounds for $\delta_n$}

\Cref{lemma:rec ar 2 no peal} and \Cref{lemma:covering for small radius} are used for proving \Cref{lemma:rec ar 3}, which provides the probability bound for $\delta_T'.$
\begin{lemma}    \label{lemma:rec ar 2 no peal}
Suppose $\h = W^{\alpha,2}$ and $0< \beta\le1 $ is any  constant. Under the conditions in \Cref{example:farD}, it holds that 
\begin{align*}
 & P  \Bigg(        \left| 
\frac{1}{T  }\sum_{t=1}^T  
\left(   \sum_{d=1}^D    \int v_d(  r) X_{t -d} (r) dr   
 \right)  ^2    -
 E \left(   \sum_{d=1}^D    \int v_d(  r) X_{t -d} (r) dr   
 \right)  ^2  
\right|   
 \le     C_ w \beta    T ^{  \frac{-\alpha}{2\alpha +1 } }  
 \\
 &   \text{ for all }   \sup_{1\leq d \leq D} \|v_d\|_\h \leq 1  \text{ such that } 
 \sum_{d=1}^D \| v_d\|_\lt^2  \le \beta ^2
 \Bigg) 
\ge 1-  2T\exp \left (-c_w  T ^{\frac{1}{2\alpha + 1}}     \right  ) ,
\end{align*}
 where $c_w$ and $C_w$ are two absolute constants independent of $T $.
 \end{lemma}  
 
 \begin{proof}
Let  
$S$  be a subspace  of $\lt$ such that 
$$  S \supset \s \{ \col(A^*_d)\cup   \row(A^*_d)  \}_{d=1}^D  $$
and let $\{ w_i\}_{i=1}^N$ be the orthonormal sub-basis in $\lt$  of $S$.  
\\
\\
 {\bf Step 1.} Let $\{ a_d\}_{d=1}^D \subset  \mathbb R^{N \times N}$, $\{ y_t\} \subset \mathbb R^N $  and 
 $\{ \varepsilon_t\} \subset \mathbb R^N $   be defined as in \Cref{eq:fard restricted eigenvalue}.
Then it holds that 
$$ y_{t}  = \sum_{d=1}^D a_d y_{t-d}  + \varepsilon_{t}$$ 
and thus $\{y_t\}_{t=1}^T $ is a VAR(D) process in $\mathbb R^N$.  Observe that \eqref{eq:invertible operators} implies that 
$$ \text{det} \left ( I_{N} - \sum_{d=1}^D z^d  a_d \right)\not = 0 $$
for all $z\in \mathbb C$ such that $|z|\leq 1$. Thus from {\bf Step I}  of  the proof of Proposition 8 in \cite{wong2017lasso} , it holds that for any fixed $\{ w_d\}_{d=1}^D \subset \mathbb R^N$,
\begin{align} \label{eq:projection deviation bound for variance}  P \left ( \left| 
\frac{1}{T  }\sum_{t=1}^T  
\left(   \sum_{d=1}^D y_{t -d } ^\top    w_d 
 \right)  ^2    -
 E \left(   \sum_{d=1}^D     y_{t -d }  ^\top  w_d   
 \right)  ^2  
\right|   \ge \eta \sum_{d=1}^D\| w_d \|_2^2 \right)  \le 2T \exp \left ( -c_\mathcal E  T \eta ^2\right ) ,\end{align}
where $c_\mathcal E $ depends on $C_\epsilon$ and $\gamma_A$ only. 
 \
 \\
 \\
 {\bf Step 2.} Let $  \{  v_d \}_{d=1}^D \subset  \h$ be any deterministic functions such that 
  $\sup_{1\leq d\leq D}\| v_d\|_\h \le 1$. Suppose  $  S \supset \s \{ \col(A^*_d)\cup   \row(A^*_d)  \}_{d=1}^D\cup   \{ v_d \}_{d=1}^D.  $
 Then \eqref{eq:projection deviation bound for variance} implies that if $\sum_{d=1}^D \| v_d\|_\lt^2 \le \beta^2 $
\begin{align*} P & \left (         \left| 
\frac{1}{T  }\sum_{t=1}^T  
\left(   \sum_{d=1}^D    \int v_d(  r) X_{t -d} (r) dr   
 \right)  ^2    -
 E \left(   \sum_{d=1}^D    \int v_d(  r) X_{t -d} (r) dr   
 \right)  ^2  
\right|   \ge   \gamma  \beta^2  \right) 
\\
 \le
 &  2T \exp\left ( -c _ \mathcal E  T \gamma^2\right )  
 \end{align*}
\
\\
{\bf Step 3.} Let $\{ u_j \}_{j=1}^\mathcal N$ be chosen as in  \Cref{lemma:covering for small radius} with   $\mathcal  N \le \exp \left( \frac{C_\mathcal N}{\delta ^{\frac{1}{\alpha}}  }\right) $. 
Observe that  
  \begin{align*} & P  \left (      \sup_{ \{ w_d \} _{d=1} ^D \subset  \{ u_{j}\} _{j=1}^{\mathcal N} }    \left| 
\frac{1}{T  }\sum_{t=1}^T  
\left(   \sum_{d=1}^D    \int w_d(  r) X_{t -d} (r) dr   
 \right)  ^2    -
 E \left(   \sum_{d=1}^D    \int w_d(  r) X_{t -d} (r) dr   
 \right)  ^2  
\right|   \ge  4 \gamma D \beta^2  \right) 
\\
 \le
 & \mathcal N^D  \sup_{ \{ w_d \} _{d=1} ^D \subset  \{ u_{j}\} _{j=1}^{\mathcal N} }  P   \left (        \left| 
\frac{1}{T  }\sum_{t=1}^T  
\left(   \sum_{d=1}^D    \int w_d(  r) X_{t -d} (r) dr   
 \right)  ^2    -
 E \left(   \sum_{d=1}^D    \int w_d(  r) X_{t -d} (r) dr   
 \right)  ^2  
\right|   \ge  4  \gamma D \beta^2  \right) 
\\
\le & 2 T   \exp \left( \frac{C_\mathcal N D}{\delta ^{\frac{1}{\alpha}}  } -c_  \mathcal E  T \gamma^2\right )   ,
 \end{align*}
 where the last inequality follows because $\sum_{d=1}^D \|w_d\|_\lt ^2\le 4D\beta^2  $ and {\bf Step 2}.
For any fixed $  \{ v_d\}_{d=1}^D $ such that  $\sup_{1\leq d\leq D}\| v_d\|_\h \le 1$ and $\sum_{d=1}^D \| v_d\|_\lt^2  \le \beta ^2$, by the choice of $\{ u_j \}_{j=1}^\mathcal N$,   one can assume without loss of generality that for any $1 \le d \le D$, 
$$   \| v_d-u_d  \|_\lt \le \delta \wedge \beta  .$$    
Therefore 
\begin{align*}
 &  \left| E \left(   \sum_{d=1}^D    \int u_d(  r) X_{t -d} (r) dr   
 \right)  ^2    
 -     E \left(   \sum_{d=1}^D    \int v_d(  r) X_{t -d} (r) dr   
 \right)  ^2      \right| 
 \\
  \le
& \sum_{1\le d,e\le D }  \left|   \iint  ( v_d(r)-u_d (r) )   E   \left(  X_{t -d} (r )X_{t -e} (s ) \right) 
  v_e  (s) drds   \right| 
 \\  + 
  & \sum_{1\le d,e\le D }  \left|   \iint    u _d(r)      E   \left(  X_{t -d} (r )X_{t -e} (s ) \right)   ( v_e (s)-u_e(s)) drds   \right|  
  \\
  \le & \sum_{1\le d , e\le D} \| v_d - u_d\|_\lt   E \left(  \|X_{t -d}  \|_{\infty } \|X_{t -e}  \|_{\infty }\right)  \| v_e\|_\lt 
  \\
  +&  \sum_{1\le d , e\le D} \| u_d\|_\lt  E \left(  \|X_{t -d}  \|_{\infty } \|X_{t -e}  \|_{\infty }\right)  \| v_e-u_e\|_\lt  
  \\
  \le & 2C_X^2 D^2  ( \delta \wedge  \beta  )\beta  \le 2C_X^2 D^2      \delta \beta.
 \end{align*}
Similarly, 
\begin{align*}
  \left|  \frac{1}{T  }\sum_{t=1}^T  
\left(   \sum_{d=1}^D    \int v_d(  r) X_{t -d} (r) dr   
 \right)  ^2 
 -    \frac{1}{T  }\sum_{t=1}^T  
\left(   \sum_{d=1}^D    \int w_d(  r) X_{t  -d} (r) dr   
 \right)  ^2    \right| 
  \le  
   2C_X^2 D^2   \delta  \beta . 
 \end{align*}
  So by standard covering argument
  
  \begin{align*}
   & P  \Bigg(        \left| 
\frac{1}{T  }\sum_{t=1}^T  
\left(   \sum_{d=1}^D    \int v_d(  r) X_{t -d} (r) dr   
 \right)  ^2    -
 E \left(   \sum_{d=1}^D    \int v_d(  r) X_{t -d} (r) dr   
 \right)  ^2  
\right|   
 \ge    4  \gamma D \beta^2  + 4 C_X^2 D^2   \delta  \beta  
 \\
 &   \text{ for all }   \sup_{1\leq d \leq D} \|v_d\|_\h \leq 1  \text{ such that } 
 \sum_{d=1}^D \| v_d\|_\lt^2  \le \beta ^2
 \Bigg) 
\le  2T\exp \left (-c_\mathcal E \gamma^2   T+  \frac{C_\mathcal N D}{\delta ^{\frac{1}{\alpha}}  } \right  ).
  \end{align*}
  The desired result follows by picking 
  $$\gamma =\delta =C     T  ^{\frac{-\alpha}{2\alpha+ 1}}  $$
  for sufficiently large $C$ depending on $D ,C_X ,c_\mathcal  E$ and $C_\mathcal N$.
 \end{proof}
 
\begin{lemma}\label{lemma:covering for small radius}
For any fixed $\beta>0$ and $\delta>0$, there exists a collection of functions $ \{ u_j\}_{j=1}^\mathcal N  \subset B_{W^{\alpha,2}} (0,1)\cap B_{\lt}(0, 2\beta)$ 
with $\mathcal N \le \exp \left(\frac{C _\mathcal N   }{\delta^{\frac{1}{\alpha}}} \right)$
such that for any
$v \in  B_{W^{\alpha,2}} (0,1)\cap B_{\lt}(0, \beta)$, it holds that 
$$\min_{1\le j \le \mathcal N}   \| v- u_j\|_\lt \le \delta \wedge \beta .$$

\end{lemma} 
 \begin{proof}
   Let $\{ u_j\}_{j=1}^\mathcal N   \subset  B_{W^{\alpha,2}} (0,1)        $ be a $\lt$ cover  of $B_{W^{\alpha,2}} (0,1) $. This means that   for any $v\in B_{W^{\alpha,2}} (0, 1) $, it holds that 
$$\min_{1\le j\le \mathcal N } \| v-u_j \|_ \lt  \le \delta $$
Then by the classical  result (see, e.g., \cite{nickl2007bracketing} and reference therein), there exists a constant $C_\mathcal N$ independent of $\delta$ such that $\{ u_j\}_{j=1}^\mathcal N $ can be picked so that 
  $\mathcal  N \le \exp \left( \frac{C_ \mathcal N }{\delta ^{\frac{1}{\alpha}}  }\right) $.   Without loss of generality assume that $ 0\in \{ u_j\}_{j=1}^\mathcal N$.
  \\
  \\
  {\bf Case 1.} Suppose $ \delta < \beta$. Then any $u_j$ with 
  $\| v-u_j \|_\lt  \le \delta$ must satisfy  $\|  u_j \|_\lt       \le  2\beta    $.  So it suffices to take the covering set to be $\{ u_j\}_{j=1}^\mathcal N \cap B_\lt (0,2\beta) $.
 \\
 \\
  {\bf Case 2.} Suppose $ \delta \ge \beta$.  Since $  0\in \{ u_j\}_{j=1}^\mathcal N$,
  $$\min_{1\le j\le \mathcal N } \| v-u_j \|_\lt  \le   \| v-0 \|_\lt  \le \beta.$$
 \end{proof}

\Cref{lemma:rec ar 3} provides the probability bound for $\delta_T'.$  
Note that   the definition  of $\delta_T'$  in   \eqref{eq:delta 2}  does not have  the additional condition  
 $\sum_{d=1}^D \|v_d \|_\lt^2 \le 1 $ which used in   \Cref{lemma:rec ar 3}. 
However,  since 
$  \|v_d \|_\h \le 1   $ implies that 
$   \|v_d \|_\lt \le 1 $. Therefore 
$ \sup_{1\le d\le D}\|v_d \|_\h \le 1$ implies that 
$  \sum_{d=1}^D \|v_d \|_\lt^2 \le D $. This inequality together with 
a simple rescaling argument  (such as that in the proof of \Cref{lemma:restricted eigenvalue 2})  can 
 straightforwardly show that     \Cref{lemma:rec ar 3} implies that \begin{align*}
 & P  \Bigg(        \left| 
\frac{1}{T  }\sum_{t=1}^T  
\left(   \sum_{d=1}^D    \int v_d(  r) X_{t -d} (r) dr   
 \right)  ^2    -
 E \left(   \sum_{d=1}^D    \int v_d(  r) X_{t -d} (r) dr   
 \right)  ^2  
\right|   
 \le     C _w' \sqrt {D}    T ^{  \frac{-\alpha}{2\alpha +1 } }    \sqrt {\sum_{d=1}^D \| v_d\|_\lt ^2}  
 \\
 &   \text{ for all }   \{ v_d\}_{d=1}^D   \text{ such that } 
  \sup_{1\le d\le D}\|v_d \|_\h \le 1
 \Bigg) 
\ge 1-  2T^2\exp \left (-c_w'   T ^{\frac{1}{2\alpha + 1}}     \right  ) . 
\end{align*}
 
\begin{lemma}    \label{lemma:rec ar 3}
Suppose $\h = W^{\alpha,2}$. Under the conditions in \Cref{example:farD}, it holds that 
\begin{align*}
 & P  \Bigg(        \left| 
\frac{1}{T  }\sum_{t=1}^T  
\left(   \sum_{d=1}^D    \int v_d(  r) X_{t -d} (r) dr   
 \right)  ^2    -
 E \left(   \sum_{d=1}^D    \int v_d(  r) X_{t -d} (r) dr   
 \right)  ^2  
\right|   
 \le     C _w'   T ^{  \frac{-\alpha}{2\alpha +1 } }    \sqrt {\sum_{d=1}^D \| v_d\|_\lt ^2}  
 \\
 &   \text{ for all }   \{ v_d\}_{d=1}^D   \text{ such that } 
  \sup_{1\le d\le D}\|v_d \|_\h \le 1 \text{ and that } 
 \sum_{d=1}^D \|v_d \|_\lt^2 \le 1 
 \Bigg) 
\ge 1-  2T^2\exp \left (-c_w'   T ^{\frac{1}{2\alpha + 1}}     \right  ) . 
\end{align*}
 
 \end{lemma}  
\begin{proof}
 Let 
$J \in \mathbb Z^+$ be such that 
$2^{-J } \le  T  ^{\frac{-2\alpha}{2\alpha+ 1}}   \le 2^{-J+1}$. So $J \le \log(T) $.  For any 
$ 1\le j \le J$,  it holds that with probability at least $1-2T\exp \left (-c_w   T ^{\frac{1}{2\alpha + 1}}   \right)   $,  for all $\{v_d\}_{d=1}^D$ such that  $\sup_{1\le d\le D}\|v_d \|_\h \le 1$ and
 $  2^{-j} \le \sum_{d=1}^D \|v_d\|_ \lt ^2 \le 2^{-j+1} $, 
\begin{align*} 
& \left| 
 \frac{1}{T  }\sum_{t=1}^T  
\left(   \sum_{d=1}^D    \int v_d(  r) X_{t -d} (r) dr   
 \right)  ^2    -
 E \left(   \sum_{d=1}^D    \int v_d(  r) X_{t -d} (r) dr   
 \right)  ^2  
\right|   
 \\
 \le  &   C _w   T ^{  \frac{-\alpha}{2\alpha +1 } }    \sqrt {2^{-j+1} }\le   C_w    T ^{  \frac{-\alpha}{2\alpha +1 } }    \sqrt { 2 \sum_{d=1}^D \| v_d\|_\lt ^2}  
 \end{align*}
 where the first inequality follows from \Cref{lemma:rec ar 2 no peal}.
So by union bound,  with probability at least $1-2J T \exp(-c_w T  ^{\frac{1 }{2\alpha+ 1}} )$, for all
 $ \sum_{d=1} ^D \|v_d\|_\lt  ^2\ge 2^{-J} $,
\begin{align}\label{eq:convergent rate in T 1}
 &   \left| 
 \frac{1}{T  }\sum_{t=1}^T  
\left(   \sum_{d=1}^D    \int v_d(  r) X_{t -d} (r) dr   
 \right)  ^2    -
 E \left(   \sum_{d=1}^D    \int v_d(  r) X_{t -d} (r) dr   
 \right)  ^2   \right| 
    \le
  \sqrt {2} C_w     T  ^{\frac{-\alpha}{2\alpha+ 1}}  \sqrt {\sum_{d=1}^D \| v_d\|_\lt ^2}     . 
 \end{align} 
 In addition, observe that if 
 $ \sum_{d=1}^D \| v_d\|_\lt^2 \le 2^{-J} \le T  ^{\frac{-2\alpha}{2\alpha+ 1}}     $, 
\begin{align} \nonumber  
 \frac{1}{T  }\sum_{t=1}^T  
\left(   \sum_{d=1}^D    \int v_d(  r) X_{t -d} (r) dr   
 \right)  ^2    
 = & 
 \frac{1}{T  }\sum_{t=1}^T   \sum_{d,e=1}^D \int 
 v_d(r)X_{t -d}(r) dr \int 
 v_e(r)X_{t -e}(r) dr 
 \\ \nonumber 
 \le 
 &  \frac{1}{T  }\sum_{t=1}^T     \sum_{d,e=1}^D  \|v_d\|_\lt  \|v_e\|_\lt   C_X^2 
 \\\nonumber  
= &C_X^2 \left( \sum_{d=1}^D \| v_d\|_\lt \right)^2
\\\nonumber  
\le& DC_X^2  \sum_{d=1}^D \| v_d\|_\lt ^2. 
 \\
\le  &  DC_X^2 T^{\frac{- \alpha}{2\alpha+ 1} } \sqrt {  \sum_{d=1}^D \| v_d\|_\lt ^2.  } \nonumber  
 \end{align} 
 Similarly, if 
 $ \sum_{d=1}^D \| v_d\|_\lt^2 \le 2^{-J} \le T  ^{\frac{-2\alpha}{2\alpha+ 1}}     $, 
 $$ \left| 
 E \left(   \sum_{d=1}^D    \int v_d(  r) X_{t -d} (r) dr    \right|
 \right)  ^2  \le  DC_X^2 T^{\frac{- \alpha}{2\alpha+ 1} } \sqrt {  \sum_{d=1}^D \| v_d\|_\lt ^2.  }   $$
 Therefore if 
 $ \sum_{d=1}^D \| v_d\|_\lt^2 \le 2^{-J} \le T  ^{\frac{-2\alpha}{2\alpha+ 1}}     $, 
 \begin{align}  \label{eq:convergent rate in T 2} \left| 
 \frac{1}{T  }\sum_{t=1}^T  
\left(   \sum_{d=1}^D    \int v_d(  r) X_{t -d} (r) dr   
 \right)  ^2    -
 E \left(   \sum_{d=1}^D    \int v_d(  r) X_{t -d} (r) dr   
 \right)  ^2  
\right|   
 \le  &  2DC_X^2 T^{\frac{- \alpha}{2\alpha+ 1} } \sqrt {  \sum_{d=1}^D \| v_d\|_\lt ^2.  }  
 \end{align}
 The desired result follows from \eqref{eq:convergent rate in T 1} and \eqref{eq:convergent rate in T 2}. 
\end{proof}

\Cref{lemma:azuma hoeffding} provides the probability bound for $\delta_T''.$  
\begin{lemma}\label{lemma:azuma hoeffding}
 Suppose  $\h = W^{\alpha,2}$ and \Cref{assume:X} holds. Let  $\{s_i \}_{i=1}^n $ being  a collection of uniform random variables sampled from $[0,1]$ independent of $\{\epsilon_t\}_{t=1}^T$. Under the conditions in \Cref{example:farD}, it holds that 
 $$ P \left( \sup_{1\le d\le D,  r \in [0,1], s\in [0,1] } 
  \left|  \frac{1}{T}  \sum_{t=1}^T  X_{t -d} (r) \epsilon_{t } (s ) 
 \right|  \ge  3   C_X  C_\epsilon   \sqrt { \frac{\log(T )}{ T}} \right) \le 1/T^{ 3}  .  $$
 \end{lemma} 
\begin{proof} 
{\bf Step 1.} Let $r,s \in [0,1]$ be given. 
Let $ y_\tau  = \sum_{t=1}^\tau  X_{t-d } (r)  \epsilon_{t} (s ) $. One has 
$$ | y_{\tau } -y_{\tau -1 }  |  \le | X_{\tau-d } (r)  \epsilon_{ \tau} (s)  | \le C_X C_\epsilon.  $$
Let $ \mathcal F_\tau$ be  the sigma algebra generated by $\{ \epsilon_t\}_{t=1}^\tau$.   Then
$$  E(y_{\tau+1}  | \mathcal F_\tau  ) = y_\tau. $$ 
Therefore, by Azuma Hoeffding inequality, it holds that 
\begin{align} \label{eq:azuma hoeffding far deviation bound}P \left(  
  \left|  \frac{1}{T}  \sum_{t=1}^T  X_{t -d} (r ) \epsilon_{t } (s ) 
 \right|  \ge     \delta  \right) \le  2 \exp \left( - \frac{T^2 \delta ^2}{2 C_X^2 C_\epsilon ^2 } \right) .  
 \end{align}
\
\\
{\bf Step 2.} Let $\mathcal G : =\{   r_m \}_{m=1}^M  $ be a equally spaced grid on $[0,1]$.  Observe that  for any $s\in [0,1]$, there exist 
$m $ such that $|r_{m } - s|\le 1/M$. Therefore for any $f\in W^{\alpha,2}$ with $\alpha \ge 1$,  
$$ | f(r_m)  -f (s)|\le \|f'\|_\lt  |r_m -s | \le \|f \|_\h  |r_m -s |  \le \|f\|_\h /M.    $$
Therefore given any $ r,s \in [0,1] $, there exists 
$r_m, r_{m'}$ such that 
$$ | X_{t-d} (r) -X_{t-d} (r_m)| \le C_X /M \text{ and } | \epsilon_{t} (s) -\epsilon_{t} (r_m')| \le C_\epsilon /M $$
 and therefore it holds that 
\begin{align} \nonumber 
& |X_{t-d} (r) \epsilon_{t}(s)   -X_{t-d} (r_{m}) \epsilon_{t}(r_{m'})  | 
\\\nonumber 
 \le  &
 |X_{t-d} (r) \epsilon_{t}(s)   -X_{t-d} (r_{m}) \epsilon_{t}(s )  | 
 + |X_{t-d} (r_m ) \epsilon_{t}(s)   -X_{t-d} (r_{m}) \epsilon_{t}(r_{m'})  | 
\\ \nonumber 
 \le
 & \|X_{t-d}\|_\infty C_\epsilon /M  + \|\epsilon_{t}\|_\infty C_X/M 
 \\ \label{eq:grid approx}
 \le & 2 C_XC_\epsilon/M. 
\end{align} 
So 
\begin{align*}   & P \left(   \sup_{r,s \in [0,1]}
  \left|  \frac{1}{T}  \sum_{t=1}^T  X_{t -d} (r ) \epsilon_{t } (s ) 
 \right|  \ge     \delta  + 2C_XC_\epsilon/ M  \right)  
 \\
\le & P \left(   \sup_{r_m, r_m' \in \mathcal G }
  \left|  \frac{1}{T}  \sum_{t=1}^T  X_{t - d} (r_m ) \epsilon_{t } (r_{m'}   ) 
 \right|  \ge     \delta   \right)     
 \\
 \le & M^2 2 \exp \left( - \frac{T^2 \delta ^2 }{2 C_X^2 C_\epsilon ^2 } \right) .
 \end{align*} 
 where the first inequality follows from \eqref{eq:grid approx} and the second inequality follows from \eqref{eq:azuma hoeffding far deviation bound} and union bound. So by union bound again,
 $$ P \left( \sup_{1\le d\le D,  r \in [0,1], s\in [0,1] } 
  \left|  \frac{1}{T}  \sum_{t=1}^T  X_{t -d} (r) \epsilon_{t } (s ) 
 \right|  \ge    \delta  + 2C_XC_\epsilon/ M  \right)   \le   D M^2 2 \exp \left( - \frac{T^2 \delta ^2 }{2 C_X^2 C_\epsilon ^2 } \right) .  $$
 The desired result follows by taking  $\delta = C_XC_\epsilon  \sqrt  { \frac{  \log(T) }{ T }}  $ and  $M =  \sqrt {T}. $
 
 \end{proof}

The probability bound for $\delta_T'$ in \Cref{lemma:rec ar 3} and for $\delta_T''$ in \Cref{lemma:azuma hoeffding} imply a probability bound on $\delta_T$, summarized in the following corollary.
\begin{corollary}\label{corollary:delta explicit rate}
Suppose  $\h = W^{\alpha,2 } $ and  \Cref{assume:X} holds.  Let $\delta_n'$ and $\delta _n''$ be defined as in \eqref{eq:delta 2} and \eqref{eq:delta 1} respectively.  Under the conditions in \Cref{example:farD},   there exists constants $c'_w, C'_w$ 
 such that  
\begin{align*}  
   P \left(  \delta_T '  \ge    C _w'   T ^{  \frac{-\alpha}{2\alpha +1 } }   \right ) \le 2T^2\exp \left (-c_w'   T ^{\frac{1}{2\alpha + 1}}     \right  )   \quad 
 \text{and}    \quad    
   P \left(  \delta_T'' \ge 3 C_XC_\epsilon   \sqrt { \frac{\log(T )}{ T}}    \right ) \le   T^{-3} . 
\end{align*}
  
\end{corollary}

\section{Proof of \Cref{lemma:representer}}

\begin{proof}[Proof of \Cref{lemma:representer} ]

Let $S_1, S_2 \subset \mathcal H$. Denote 
$$ A\vert _{S_1 \times S_2 } [f,g] =  A [\mathcal P  _{S_1}f, \mathcal P  _{S_2}g].$$
Let $S =\text{span}\{ \mathbb K(s_i,\cdot)  \}_{i=1}^n $ and that $S^\perp $ is the orthogonal complement of  $S$ in $\h$. Then
$$  A =   A\vert _{S  \times S  }  + A\vert _{S  \times S^\perp  } +  A\vert _{S^\perp \times S  } +  A\vert _{S^\perp  \times S ^\perp  } .$$
From \Cref{lemma:constrained operator}, $ A\vert  _{S  \times S  } $ can be written as 
$$   A\vert  _{S  \times S  }  [f, g] = \sum_{1\le i,j \le n}    a    _{ij } \lb \mk (  s_i ,   ), f \rb_{\mathcal H } \lb \mk (s_j ,  ),g \rb_{\mathcal H }  . $$
{\bf Step 1.}  In this step, it is shown that  $ \{ A(s_i, s_j)\}_{i,j=1} ^n  $ only depend  on $ A\vert _{S  \times S  }$.
Observe that 
$$   A\vert_{S\times S^\perp  } (x_i,x_j)=     A\vert_{S\times S^\perp  }[ \mathbb K(s_i, \cdot) , \mathbb K(s_j, \cdot) ]  
=     A [  \mathcal P_S \mathbb K(s_i, \cdot) ,  \mathcal P_{S^\perp }  \mathbb K(s_j, \cdot) ]  =0.  $$
Similarly
$   A\vert_{S^\perp\times S  }[ \mathbb K(s_i, \cdot) , \mathbb K(s_j, \cdot) ]  = 0 $
and 
$   A\vert_{S^\perp\times S^\perp   }[ \mathbb K(s_i, \cdot) , \mathbb K(s_j, \cdot) ]  = 0 $ for all $1\le i,j\le n$. 
\\
\\
{\bf Step 2.}    Let 
  $\{\widehat B_d \}_{d=1}^D  $   be  a    solution to \eqref{eq:approx A 2} and let  $ \widehat {A } _d= \widehat B_d| _{S\times S}    $. Then by {\bf Step 1}   it holds that 
$\widehat A _d (s_i,s_j) =\widehat B_d (s_i,s_j)    $  
  for all $1\le i,j\le n$ and all $1 \le d \le D $. Therefore 
 \begin{align*} &  \sum_{t=1}^T    \sum_{i=1}^ n   \left  (X_{t  } (s_i) -  \sum_{d=1}^D  \frac{1}{n }\sum_{j=1} ^n \widehat A_d  (s_i, s_j)  X_{t -d } (s_j) \right  )^2    
\\ 
=   & 
 \sum_{t=1}^T    \sum_{i=1}^ n   \left  (X_{t  } (s_i) -  \sum_{d=1}^D  \frac{1}{n }\sum_{j=1} ^n \widehat B_d  (s_i, s_j)  X_{t -d } (s_j) \right  )^2     .   
 \end{align*}
 From  \Cref{lemma:inequality of norm} it holds that  $\|    \widehat B_d \|_{\h ,*}  \ge \| \widehat {A } _d\|_{\h,*} $. 
 As a result,   $\{\widehat A_d \}_{d=1}^D  $   is also a  solution to \eqref{eq:approx A 2}. 
 So by \Cref{lemma:constrained operator}, 
$$ \lb \widehat {    A }  _d  [f], g\rb_\h  = \sum_{1\le i,j \le n}    \widehat a    _{d,ij } \lb \mk (  s_i ,   ), f \rb_{\mathcal H } \lb \mk (s_j ,  ),g \rb_{\mathcal H }   $$
as desired. 
\end{proof}

\begin{lemma}\label{lemma:constrained operator}
Let $   A : \mathcal H \to \mathcal H   $ be any linear compact operator  and let 
$S$ be any subspace of $\h$  spanned by $\{ v_1,\ldots, v_m\}$. 
Then there exists $\{a_{ij}\}_{i,j=1}^m$ not necessarily unique such that 
$$ \lb  A\vert_{S\times S} [f] ,  g \rb _\h  = \sum_{i,j=1}^m a_{ij }\lb v_i, f\rb_{\mathcal H }\lb v_j, g\rb_{\mathcal H }   .$$ 

\end{lemma}
\begin{proof} 
 Let $\{ u_i\}_{i=1}^m$ be the orthonormal basis of $S \subset \h  $. Since each $u_i$ can be written as linear combination of 
 $\{ v_1,\ldots, v_m\} $, it suffices to show that 
 $$  \lb A\vert_{S\times S} [f] ,g \rb_\h  =  \sum_{i,j=1}^m a_{ij }\lb u_i, f\rb_{\mathcal H }\lb u_j, g\rb_{\mathcal H }  .$$
 Since $S$ is a linear subspace of $\mathcal H$, there exists $\{u_{i}\}_{i=m+1}^\infty$ such that $\{ u_i\}_{i=1}^\infty $ is the basis of $\mathcal H$ and that 
 $$  \lb A\vert_{S\times S} [f] ,g \rb_\h  =  \sum_{i,j=1}^\infty  a_{ij }\lb u_i, f\rb_{\mathcal H }\lb u_j, g\rb_{\mathcal H }   .$$
 where $a_{ij}=  \lb A[u_i], u_j\rb _\h $. Therefore
\begin{align*} 
  \lb A\vert_{S\times S} [f] ,g \rb_\h  
= &  \sum_{i,j=1}^\infty  a_{ij }\lb u_i, \mathcal P_S f\rb_{\mathcal H }\lb u_j, \mathcal P_S g\rb_{\mathcal H }    
\\
=  &  \sum_{i,j=1}^\infty  a_{ij }\lb \mathcal P_S u_i,  f\rb_{\mathcal H }\lb \mathcal P_S u_j,  g\rb_{\mathcal H }    
\\
= 
& \sum_{i,j=1}^m a_{ij }\lb  u_i,  f\rb_{\mathcal H }\lb  u_j,  g\rb_{\mathcal H }    .
\end{align*}
 \end{proof} 
\begin{lemma} \label{lemma:inequality of norm}
Let $A : \h  \to  \h     $ be any  linear compact operator and $S$ be any finite-dimensional subspace of $  \h $. Then $\| A \vert_{S\times S}\|_{\h , *} \le \| A\|_{\h ,*}  $.

\end{lemma}
\begin{proof}
Let $S$ be of dimension $K$. Observe that $A^\top A $ is a self-adjoint and compact operator from  $\h \to \h $. Let  $\lambda_1 \ge \lambda_2 \ge \ldots \ge 0 $ be the eigenvalues of $A^\top A $, 
let $\nu_1 \ge \nu_2   \ge \ldots \ge 0 $ be the eigenvalues of $ \mathcal P_SA^\top A \mathcal P_S $ and let
$ \tau_1 \ge \tau_2 \ge \ldots  \ge 0 $ be the eigenvalues of $ \mathcal P_SA^ \top \mathcal  P_S \mathcal  P_S A\mathcal  P_S $.
\\
\\
 Observe that 
$$   \| A  \|_{\h, *}  = \sum_{i=1}^\infty   \sqrt \lambda_i \quad \text{and}
\quad 
\| A\vert_{S\times S} \|_{\h,*}  = \sum_{i=1}^K \sqrt  \tau_i  .$$
Therefore  it suffices to show that for all $i$ 
$$\lambda_i \ge \tau_i . $$
Since by \Cref{coro:min-max 1},
$\lambda_i \ge \nu_i$ and by \Cref{coro:min-max 2},
$\nu_i \ge \tau_i$, the  desired result follows. 
\end{proof}

\begin{theorem}[Min-Max Theorem] Let $B$ be a compact, self-adjoint operator on  $\h $ with 
$\lambda_ 1 \ge \lambda_2 \ge \ldots $ being the   eigenvalues of $B$. Then
\begin{align*}
\max_{T_k \subset \h } \ \min_{v\in T_k , \|v\|_{{\h } } =1 } \lb B [v] , v\rb_{\h  } &  = \lambda_k \text{ and } 
\\
 \min_{T_{k-1} \subset \h     } \ \max_{v\in T_{k-1}^\perp , \|v\|_{{\h} } =1 } \lb B [v] , v\rb_{\h} & =\lambda_k,
\end{align*}
where $  {T} _k$ denote any  subspace of $ \h $ of dimension $k$. 

\end{theorem}
 This is a well-known results for operators. 
\begin{corollary} \label{coro:min-max 1}
Let $B$ be compact self-adjoint positive definite operator on $\h \to \h$ and $S$ be any $M$ dimensional  subspace of $\h$.  Let 
$ \lambda_1 \ge \lambda_2 \ge \ldots\ge 0$ be the eigenvalues of $B$ and $\nu_1\ge \nu_2 \ge \ldots$ be the eigenvalues of 
$\mathcal P_SB \mathcal P_S$.  Then for all $i$, it holds that
$$ \lambda_i \ge \nu_i . $$ 
\end{corollary}
\begin{proof} 
Since $ S $ is of dimension $M$, $\nu_{M+1} = \nu_{M+2}=\ldots=0$. Let $\{u_i\}_{i=1}^M$ denote the eigenvectors of $ \mathcal P_SB \mathcal P_S$ corresponding to $\{ \nu_i\}_{i=1}^M$ and
let $S_k =\text{span} \{ u_i\}_{i=1}^k $. 
Then for $k \le m$,  it holds that
$$   \lambda_ k  \ge     \min_{v\in S_k, \|v\|_{  \h  } =1 } \lb B [v] ,v\rb_{\h  }  = \min_{v\in S_k, \|v\|_{ \h  } =1 } \lb \mathcal P_S B\mathcal P_S [ v] ,v\rb_{\h  }= \nu_k ,  $$
where the first inequality follows from the Min-max theorem.
For $k>m$, 
$$  \lambda_k \ge 0 = \nu_i.$$
\end{proof}

\begin{corollary}\label{coro:min-max 2}
Let $A$ be any compact operator and $S$ be any  $M$ dimensional  subspace of $\h  $.  Let 
$ \lambda_1 \ge \lambda_2 \ge \ldots$ be the eigenvalues of $\mathcal  P_SA^\top A \mathcal  P_S$ and $\nu_1\ge \nu_2 \ge \ldots$ be the eigenvalues of 
$\mathcal P_SA^\top  \mathcal P_S \mathcal P_SA \mathcal P_S$.  Then for all $i$, it holds that
$$ \lambda_i \ge \nu_i . $$ 
\end{corollary}
\begin{proof}
Let $B=AP_S $. Therefore $Bv=0$ for any $v\in S^\perp$. So
\begin{align*}
\lambda_k &= \max_{T_k  \subset  {\h}   } \min_{v\in T_k, \|v\|_{\h }  =1 } \lb B^\top B[v] , v  \rb_{\h }    = \max_{T_k  \subset  \mathcal H  } \min_{v\in T_k,\|v\|_{\h } =1 } \|Bv\|_{\h}  ^2 ,
\\
\nu _k &= \max_{T_k  \subset  {\h}  } \min_{v\in T_k, \|v\|_{\h}  =1} \lb B^\top \mathcal  P_S  \mathcal   P_S B [v] , v  \rb_{\h}    = \max_{T_k  \subset  {\h}  } \min_{v\in T_k, \|v\|_{\h}  =1} \| \mathcal  P_SBv\|_{\h}  ^2. 
\end{align*}
For any $v$, it holds that $\|Bv\|_{\h } \ge  \|   \mathcal  P_S Bv\|_{\h }  ^2$ as $\mathcal  P_S$ is the projection operator onto $S$. Therefore $\lambda_k \ge v_k$ as desired. 
\end{proof}

\section{Proof of \Cref{theorem:far rate}} 
Remark: Equations \eqref{eq:gamma 1}, \eqref{eq:gamma 2} and \eqref{eq:delta 2} are repeatedly used in the proof of \Cref{theorem:far rate}. Note that by rescaling, \eqref{eq:gamma 1}, \eqref{eq:gamma 2} and \eqref{eq:delta 2} can be applied to any function with bounded norm in $\h$. For example, consider $g\in \h$, $ \|g\|_\h > 1 $. Let $f = {g}/{\|g\|_\h} $ and apply \eqref{eq:gamma 1} to $f$, we have 
\begin{align*}
	\left |  \int   g(s)  ds  - \frac{1}{n} \sum_{i=1}^ n  g(s_i )  \right |    \le  \gamma  \|g\|_\lt    +  \gamma ^2  \|g\|_\h.
\end{align*}
This rescaling calculation is applied repeatedly in the proof of \Cref{theorem:far rate}.

\begin{proof}[Proof of \Cref{theorem:far rate}]

Note that we have $C_A>\max_{1\leq d\leq D} \tau_d$. Observe that $\{ A^*_d\}_{d=1}^D  \in  \mathcal C _{\bm \tau} $ since $\|A^*_d\|_{\h,*} \leq \tau_d$ for $d=1,\cdots, D$. Thus we have
\begin{align*}
& \sum_{t=1}^T \sum_{i=1}^n    \frac{1}{Tn } \left (X_{t } (s_i)  -  \sum_{d=1}^D \frac{1}{n } \sum_{j=1}^n  \widehat A_d(s_i, s_j) X_{t -d} (s_j) \right)^2   
\\
\le& 
\sum_{t=1}^T \sum_{i=1}^n    \frac{1}{Tn } \left (X_{t } (s_i)  -\sum_{d=1}^D  \frac{1}{n } \sum_{j=1}^n  A _d ^* (s_i, s_j) X_{t -d} (s_j) \right)^2  
 \end{align*}
Denote $ \Delta_d = A^*_d -\widehat A_d$.  Then standard calculation gives
\begin{align} \label{eq:nf step 0 term 1}
& \frac{1}{Tn } \sum_{t=1}^T \sum_{i=1}^n \left(   \sum_{d=1}^D   \frac{1}{n}\sum_{j=1}^n \Delta _d (s_i, s_j)   X_{t -d} (s_j ) \right)  ^2  
  \\ \label{eq:nf step 0 term 2}
   \le &  \frac{2}{Tn } \sum_{t=1}^T \sum_{i=1}^n \left( \sum_{d=1}^D   \frac{1}{n}\sum_{j=1}^n  \Delta_d   (s_i, s_j) X_{t -d}  (s_j ) \right)  
   \left(  X _{t } (s_i) - \sum_{d=1}^D  \frac{1}{n} \sum_{j=1}^n  A_d ^*   (s_i, s_j) X_{t -d}  (s_j )    \right)      
  \end{align} 
  \
  \\
{\bf Step 1.}   
Observe that
\begin{align}\label{eq:regular estimators} \sup_{1\le i\le n } \| \Delta_d(s_i, \cdot)\|_\h  \le \sup_{r\in [0,1]}  \left ( \| A^*_d (r,\cdot) \|_\h +   \| \widehat A_d  (r,\cdot) \|_\h  \right) \le  \| \widehat A_d\|_{\h , * }  + \|  A^*_d \|_{\h , *}  \le  2C_A,
\end{align} 
where the second to last inequality follows from \Cref{lemma:bound of A 1}. 
 Therefore 
for any fixed $ t$ and $i$,
  \begin{align} \nonumber 
   &
\left(   \sum_{d=1}^D    \frac{1}{n}\sum_{j=1}^n  \Delta_d  (s_i, s_j)   X_{t -d}  (s_j ) \right)  ^2   
\\ \nonumber 
\ge  &
   \frac{1}{2 } 
\left(   \sum_{d=1}^D    \int \Delta_d(s_i, r) X_{t -d} (r) dr   \right)  ^2    - 
 2   
\left(   \sum_{d=1}^D    \int \Delta_d(s_i, r) X_{t -d}  (r) dr  - \frac{1}{n}\sum_{j=1}^n  \Delta_d  (s_i, s_j)   X_{t-d}  (s_j )  \right)  ^2 
 \\ \nonumber 
   \ge &  \frac{1}{2 } 
\left(   \sum_{d=1}^D    \int \Delta_d(s_i, r) X_{t -d} (r) dr   \right)  ^2      -   \left( \gamma_n 
\sum_{d=1}^D 
\|\Delta_d(s_i,\cdot)\|_\lt    +  2C_AC_X  \gamma_n^2 \right) ^2. 
\\ 
\nonumber 
   \ge &  \frac{1}{2 } 
\left(   \sum_{d=1}^D    \int \Delta_d(s_i, r) X_{t -d} (r) dr   \right)  ^2  
   -
       2 \gamma^2_n \left(
\sum_{d=1}^D 
\|\Delta_d(s_i,\cdot)\|_\lt     \right) ^2   - 8C_A^2C_X ^2 \gamma_n^4   
\\ 
\label{eq:fard restricted 1} 
   \ge &  \frac{1}{2 } 
\left(   \sum_{d=1}^D    \int \Delta_d(s_i, r) X_{t -d} (r) dr   \right)  ^2  
   -
       2\gamma^2_n D\sum_{d=1}^D   
\|\Delta_d(s_i,\cdot)\|_\lt  ^2      -  8C_A^2C_X^2  \gamma_n^2
  \end{align} 
where the second  inequality follows from \eqref{eq:gamma 1} and the fact that 
$$ \sup_{1\le d\le D , 1\le i \le n}   
\| \Delta_d (s_i,\cdot ) X_{t -d} (\cdot)\|_\h 
\le \sup_{1\le i\le n } \| \Delta_d(s_i, \cdot)\|_\h   \| X_{t -d} \|_\h  \le 2C_A C_X ,$$
and $\gamma_n^4 \le \gamma_n^2$ is used in the last inequality.  
In addition,
\begin{align} \nonumber 
 &\frac{1}{T  }\sum_{t=1}^T  
\left(   \sum_{d=1}^D    \int \Delta_d(s_i, r) X_{t -d} (r) dr   \right)  ^2   
\\ \nonumber 
   \ge  & E   \left(   \sum_{d=1}^D    \int \Delta_d(s_i, r) X_{t -d} (r) dr   \right)  ^2         -\delta_T  C_A \sqrt {\sum_{d=1}^D \|\Delta_d (s_i, \cdot) \|_\lt  ^2} 
   \\ \nonumber  
   \ge & \kappa_X   \sum_{d=1}^D \|\Delta_d (s_i, \cdot) \|_\lt  ^2  - \frac{2 \delta_T ^2C_A^2 }{\kappa_X }  -\frac{\kappa_X }{2} \sum_{d=1}^D \|\Delta_d (s_i, \cdot) \|_\lt  ^2   
   \\
   = 
   & \frac{\kappa_X }{2} \sum_{d=1}^D \|\Delta_d (s_i, \cdot) \|_\lt  ^2    - \frac{2 \delta_T ^2C_A^2 }{\kappa_X } ,
   \label{eq:fard restricted 2}
\end{align}
where the first inequality follows from \eqref{eq:delta 2}, the second inequality follows from \Cref{assume:X}.  Therefore 
 \begin{align} \nonumber  \eqref{eq:nf step 0 term 1} =
&\frac{1}{Tn} \sum_{t=1}^T \sum_{i=1} ^n \left(   \sum_{d=1}^D    \frac{1}{n}\sum_{j=1}^n  \Delta_d  (s_i, s_j)   X_{t -d}  (s_j ) \right)  ^2     
\\ \nonumber 
\ge 
&\frac{1}{Tn} \sum_{t=1}^T \sum_{i=1} ^n   \frac{1}{2 } 
\left(   \sum_{d=1}^D    \int \Delta_d(s_i, r) X_{t -d} (r) dr   \right)  ^2      -    \frac{2\gamma^2_n D}{n}\sum_{i=1}^n\sum_{d=1}^D   
\|\Delta_d(s_i,\cdot)\|_\lt  ^2      -  8C_A ^2C_X^2  \gamma_n^2   
\\ \nonumber 
\ge 
&\frac{\kappa_X }{4n}  \sum_{i=1}^n \sum_{d=1}^D \|\Delta_d (s_i, \cdot) \|_\lt  ^2    - \frac{2 \delta_T^2 C_A^2 }{\kappa_X } - \frac{2\gamma_n^2 D }{n}   \sum_{i=1} ^n     \sum_{d=1}^D \|\Delta_d(s_i,\cdot)\|_\lt   ^2  
-8C_A ^2C_X^2 \gamma_n^2 
\\ \label{eq:fard restricted 2.5}
\ge 
&\frac{\kappa_X }{8n}  \sum_{i=1}^n \sum_{d=1}^D \|\Delta_d (s_i, \cdot) \|_\lt  ^2    - \frac{2 \delta_T^2 C_A ^2 }{\kappa_X } 
-  8C_A ^2C_X^2 \gamma_n^2 
\end{align}  
where the first inequality follows from \eqref{eq:fard restricted 1},   the second inequality follows from \eqref{eq:fard restricted 2}, and the last inequality follows from the assumption that 
$\kappa_X \ge 64\gamma_n^2 D  $
\\
\\
{ \bf Step 2.} Observe that  
\begin{align*}
& \frac{1}{2} \cdot  \eqref{eq:nf step 0 term 2} 
  \\
   = 
   & 
\frac{1}{Tn } \sum_{t=1}^T \sum_{i=1}^n  \sum_{d=1}^D  \left(    \frac{1}{n}\sum_{j=1}^n  \Delta _d      (s_i, s_j) X_{t -d}  (s_j ) \right)  
     \epsilon_{t } (s_i)    
     \\
     +
     & \frac{1}{Tn } \sum_{t=1}^T \sum_{i=1}^n  \sum_{d=1}^D  \left(    \frac{1}{n}\sum_{j=1}^n  \Delta _d      (s_i, s_j) X_{t -d}  (s_j ) \right) \sum_{d=1}^D  \left(   \int A^*_d(s_i, s) X_{t -d}(s) ds  - 
 \  \frac{1}{n}  \sum_{j=1}^n  A^* _d   (s_i, s_j) X_{t -d}  (s_j )   \right) 
 \\
 \le 
 & 
 \sum_{d=1}^D  \frac{1}{n^2} \sum_{i,j =1}^n |  \Delta_{d} (s_i, s_j) | \sup_{1\le d\le D, 1\le i,j \le n } 
  \left|  \frac{1}{T}  \sum_{t=1}^T  X_{t -d} (s_j) \epsilon_{t } (s_i) 
 \right|
 \\
 +
 & 
  C_X \sum_{d=1}^D  \frac{1}{n^2} \sum_{i,j =1}^n   |  \Delta_{d} (s_i, s_j) | D \sup_{1\le t\le T , 1\le i \le n , 1\le d \le D  }
      \left |   \int A^*_d(s_i, s) X_{t -d}(s) ds  - 
 \  \frac{1}{n}  \sum_{j=1}^n  A^* _d   (s_i, s_j) X_{t -d}  (s_j )   \right | 
\end{align*}
From \eqref{eq:delta 1}, it holds that  $$ \sup_{1\le d\le D, 1\le i,j \le n } 
  \left|  \frac{1}{T}  \sum_{t=1}^T  X_{t -d} (s_j) \epsilon_{t  } (s_i) 
 \right|  \le \delta_T' \le \delta_T   .  $$
 In addition, since for any $1\le i \le n $ 
 $$  \| A_d^* (s_i  ,\cdot) X_{t - d  } (\cdot)\|_\lt \le \| A_d^* (s_i  ,\cdot) X_{t -d  } (\cdot)\|_\h \le \sup_{r\in [0,1 ]} \|A_d^* (r ,\cdot)   \|_\h \| X_{t+1-d  } \|_\h \le C_A C_X,$$
  \eqref{eq:gamma 1} implies that 
 $$  \sup_{1\le t\le T , 1\le i \le n , 1\le d \le D  }
      \left |   \int A^*_d(s_i, s) X_{t -d}(s) ds  - 
 \  \frac{1}{n}  \sum_{j=1}^n  A^* _d   (s_i, s_j) X_{t -d}  (s_j )   \right |   \le  C_A C_X (\gamma_n  +\gamma_n^2) \le 2 C_A C_X  \gamma_n. $$
Therefore 
\begin{align} \nonumber 
   \frac{1}{2} \cdot   \eqref{eq:nf step 0 term 2}
  \le      &  \sum_{d=1}^D  \frac{1}{n^2} \sum_{i,j =1}^n  |  \Delta_{d} (s_i, s_j) |   
 \left( \delta _T+ 2 D  C_A C_X  \gamma_n    \right )
 \le    \sum_{d=1}^D \sqrt {  \frac{1}{n^2} \sum_{i,j =1}^n    \Delta^2 _{d} (s_i, s_j)   }   \left( \delta _T   +  2  DC_A C_X  \gamma_n   \right ) 
 \\\nonumber  
 \le   & \frac{\kappa_X }{64}   \sum_{d=1}^D  \frac{1}{n^2} \sum_{i,j =1}^n    \Delta^2 _{d} (s_i, s_j)      
+  \frac{64}{\kappa_X} \left( \delta _T ^2    +  D C_AC_X \gamma_n^2   \right)
\\ 
 \le  &    \frac{\kappa_X }{32}   \sum_{d=1}^D  \frac{1}{n } \sum_{i =1}^n   \|  \Delta  _{d} (s_i,  \cdot )\|_\lt  ^2    + 
 \frac{\kappa_XD C_A }{16}  \gamma_n^2 
+  \frac{64}{\kappa_X} \left( \delta _T ^2    + D C_AC_X \gamma_n^2   \right) , \label{eq:theorem 1 step 2}
\end{align}
 where the second inequality follows from 
 $  \frac{1}{n^2} \sum_{i,j =1}^n  |  \Delta_{d} (s_i, s_j) |     \le \sqrt {  \frac{1}{n^2} \sum_{i,j =1}^n    \Delta^2 _{d} (s_i, s_j)   }  ,$
the third inequality follows from H\ou lder's  inequality  and the last inequality follows from \eqref{eq:gamma 2} and the fact that 
 $$ \| \Delta_d (s_i, \cdot)   \|_\h \le   \| \Delta_d  \|_{\h, *  } \le 2C_A.$$
 \\
{ \bf Step 3.}  Combining  \eqref{eq:fard restricted 2.5} and  \eqref{eq:theorem 1 step 2}, one has 
\begin{align*}  
 & \frac{\kappa_X}{8 }  \sum_{d=1}^D  \frac{1}{n } \sum_{i =1}^n   \|  \Delta  _{d} (s_i,  \cdot )\|_\lt  ^2   - \frac{2\delta_T^2C_A^2}{\kappa_X } 
-  8C_A ^2C_X^2 \gamma_n^2  
\\
\le
&  
\frac{\kappa_X }{32}   \sum_{d=1}^D  \frac{1}{n } \sum_{i =1}^n   \|  \Delta  _{d} (s_i,  \cdot )\|_\lt  ^2    + 
 ( \frac{\kappa_XD C_A }{16}  +D C_AC_X) \gamma_n^2 
+  \frac{64}{\kappa_X}   \delta _T ^2       , 
\end{align*}
which implies that 
\begin{align} \label{eq:standard lasso fard 1}
& \frac{\kappa_X}{16 }  \sum_{d=1}^D  \frac{1}{n } \sum_{i =1}^n   \|  \Delta  _{d} (s_i,  \cdot )\|_\lt  ^2  
\le 
       C_1'  \left( \gamma_n^2  + \delta_T^2  \right) . 
\end{align}
for some $C_ 1 '$ only depending on $\kappa_X,D, C_A,C_X$.
 Since $$  \left  \| \int \Delta_d(\cdot, r)^2dr \right\|_\h  \le   \int  \sup_{r\in [0,1] } \|\Delta_d (\cdot, r)  \|_\h ^2  dr  
\le   \| \Delta_d\|_{\h , \op}^2  \le 4 C_A^2 , $$ 
  \eqref{eq:gamma 2} 
 implies that 
$$  \frac{1}{n } \sum_{i =1}^n   \|  \Delta  _{d} (s_i,  \cdot )\|_\lt  ^2   \ge \frac{1}{2} \| \Delta_{d}\|_\lt^2  - C_A^2  \gamma_n^2.  $$  
 Therefore 
\eqref{eq:standard lasso fard 1}  and the above display give
\begin{align*}  
\frac{\kappa_X }{32} \| \Delta_d\|_{\lt }^2 \le 
  \frac{\kappa_X}{16 }  \sum_{d=1}^D  \frac{1}{n } \sum_{i =1}^n   \|  \Delta  _{d} (s_i,  \cdot )\|_\lt ^2   + \gamma_n^2
\le &
       C  _1' \left(     \gamma_n^2  + \delta_T^2  \right)   . 
\end{align*}
The above equation directly implies the desired result. 
\end{proof}  

\begin{proof}[Proof of \Cref{corollary:farD}]
	Suppose $\h=W^{\alpha,2}$. 
	From \Cref{corollary:explicit rate} it holds that  with probability at least $1-1/n^4$,  
	$$  \max \{ \gamma_n',\gamma_n''\} \le C_\alpha n^{-\alpha/(2\alpha+1) }.$$
	where $C_\alpha$ is some constant independent of $n$.  From  \Cref{corollary:delta explicit rate},  it holds that 
	\begin{align*}  
	P \left(  \delta_T '  \ge    C _w'   T ^{  \frac{-\alpha}{2\alpha +1 } }   \right ) \le 2T^2\exp \left (-c_w'   T ^{\frac{1}{2\alpha + 1}}     \right  )   \quad 
	\text{and}    \quad    
	P \left(  \delta_T'' \ge 3 C_XC_\epsilon   \sqrt { \frac{\log(T )}{ T}}    \right ) \le   T^{-3} . 
	\end{align*} 
	The result immediately follows from  \Cref{theorem:far rate}.
\end{proof}

\subsection{Proof of \Cref{prop:prediction rate}}
\begin{proof}[Proof of \Cref{prop:prediction rate}]
Define $$ \widetilde X_{T+1} (r) : = \sum_{d=1}^D \int   \widehat  A_d  (r  ,s)X_{T+1-d} (s)  ds. $$
{\bf Step 1.}
Since
\begin{align*}
&\left \| \int \widehat A_d(\cdot ,s)X_{t+1-d} (s)  ds - \int   A_d ^* (\cdot ,s)X_{t+1-d} (s)  ds \right \|_\lt
\\
=&
 \left \| \int  (\widehat A_d(\cdot ,s) -  A_d ^* (\cdot ,s) ) X_{t+1-d} (s)  ds \right \|_\lt
 \\
 \le&  \int  \left \| \widehat  A_d(\cdot ,s) -  A_d ^* (\cdot ,s)  \right \|_\lt  X_{t+1-d} (s)  ds
 \\
 \le 
 &  \left \| \widehat  A_d  -  A_d ^*   \right \|_\lt \| X_{t+1-d} \|_{\lt } ,
\end{align*}
and therefore
\begin{align*}
\| E(X_{T+1 } \vert \{X_t\}_{t=1}^T  ) -  \widetilde X_{T+1}  \|_{\lt } 
=
&\left \|  \sum_{d=1}^D \int \widehat A_d(\cdot ,s)X_{t+1-d} (s)  ds - \sum_{d=1}^D \int   A_d ^* (\cdot ,s)X_{t+1-d} (s)  ds \right \|_\lt 
\\
\le & D  \sum_{d=1}^D 
 \left \| \int  (\widehat A_d(\cdot ,s) -  A_d ^* (\cdot ,s) ) X_{t+1-d} (s)  ds \right \|_\lt 
  \\
 \le 
 & D \sum_{d=1}^D    \left \| \widehat  A_d  -  A_d ^*   \right \|_\lt     \| X_{t+1-d} \|_{\lt }  
 \\
 \le
 &  C_1' D  C_X     \left (n^{\frac{- \alpha}{2\alpha+ 1}} + T^{\frac{- \alpha}{2\alpha+ 1}}   \right ),
\end{align*} 
where the last inequality follows from \Cref{corollary:farD}.

\
\\
{\bf Step 2.} Observe that for any $r\in [0,1]$, and any $t\in [1,\ldots,T]$, 
$$\| \widehat A_d (r, \cdot )  X_t(\cdot )\|_\h  \le  C_AC_X,$$
therefore by \eqref{eq:gamma 1}, it holds that for all $r\in [0,1]$,
\begin{align*} 
&\left |  \frac{1}{n}\sum_{j=1}^n\widehat{A}_d(r ,s_j)X_{T+1-d}(s_j)  -   \int \widehat{A}_d(r ,s)X_{T+1-d}(s ) ds  \right| 
\\
\le & \| \widehat A_d (r, \cdot )  X_t(\cdot )\|_\lt \gamma_n+\| \widehat A_d (r, \cdot )  X_t(\cdot )\|_\h\gamma_n^2 
\\
\le &
2 C_AC_X \gamma_n .
\end{align*}
Let $\widehat X_{T+1}(r)=\sum_{d=1}^D\frac{1}{n}\sum_{j=1}^n\widehat{A}_d(r,s_j)X_{T+1-d}(s_j) $.
Therefore
\begin{align*} 
&\| \widehat X_{T+1} -\widetilde X_{T+1}\|_\lt \le  \| \widehat X_{T+1} -\widetilde X_{T+1}\|_\infty 
\\
=
& \sup_{r\in [0,1] }\left |  \frac{1}{n}\sum_{j=1}^n\widehat{A}_d(r ,s_j)X_{T+1-d}(s_j)  -   \int \widehat{A}_d(r ,s)X_{T+1-d}(s ) ds  \right| 
\\
\le &2C_AC_X\gamma_n \le  2  C_AC_X C_\alpha n^{\frac{-\alpha}{2\alpha+ 1}},
\end{align*}
where the last inequality follows from \Cref{corollary:explicit rate}. The desired result follows from the  inequality that 
\begin{align*}&\| E(X_{T+1 } \vert \{X_t\}_{t=1}^T  ) -  \widehat X_{T+1}  \|_{\lt }  
\le   \| E(X_{T+1 } \vert \{X_t\}_{t=1}^T  ) -  \widetilde  X_{T+1}  \|_{\lt }  
+  \| \widehat   X_{T+1} -  \widetilde  X_{T+1}  \|_{\lt }  
\end{align*}

The proof for $\frac{1}{n}\sum_{j=1}^n \left(E(X_{T+1}(s_j) \vert \{X_t\}_{t=1}^T) - \widehat X_{T+1}(s_j)\right)^2\le C_1'' \left (n^{\frac{-2\alpha}{2\alpha+ 1}} + T^{\frac{-2\alpha}{2\alpha+ 1}}   \right )$ follows the same argument and thus is omitted.
\end{proof}
 
\section{Accelerated gradient method for nuclear norm penalization}\label{sec:AGM algorithm}
We use the accelerated gradient method~(AGM, Algorithm 2 in \cite{Ji2009}) to solve the trace norm minimization problem in \eqref{eq:trace_norm2}, i.e.
\begin{align*}
\argmin_{W} g(W) + \|W\|_* = \argmin_{W} \left\|X-\K W Z \right\|_F^2 + \|W\|_*.
\end{align*}

To implement AGM, we first calculate the gradient $\nabla g(W)$ for a given $W$. We have $g(W)=\left\|X-\K W Z \right\|_F^2=\left\|X-\sum_{d=1}^D\K_d W_d Z_d \right\|_F^2$. By matrix calculus, for the component-wise gradient, we have $\nabla g(W_d)=-2\K_d^\top (X-\sum_{d'=1}^D\K_{d'}W_{d'}Z_{d'})Z_d^\top$ for $d=1,\cdots, D.$ Due to the block diagonal structure of $W$, the gradient is $\nabla g(W)=\begin{bmatrix}
\nabla g(W_1) & &\\
& \ddots &\\ 
& & \nabla g(W_D)
\end{bmatrix}$.

Denote $W_{(k)}=\begin{bmatrix}
W_{(k)1} & &\\
& \ddots &\\ 
& & W_{(k)D}
\end{bmatrix}$ to be the value of $W$ at step $k$ of AGM. To update $W_{(k)}$ using AGM, two quantities to be calculated are equations (8) and (9) in \cite{Ji2009}.

Equation (8) in \cite{Ji2009} can be written as
\begin{align*}
&Q_{t_k}(W,W_{(k-1)}):= P_{t_k}(W,W_{(k-1)})+\|W\|_*\\
=&g(W_{(k-1)}) + \langle W-W_{(k-1)}, \nabla g(W_{(k-1)}) \rangle + \frac{t_k}{2}\|W-W_{(k-1)}\|_F^2+\|W\|_* \\
=&\left\|X-\sum_{d=1}^D\K_d W_{(k-1)d} Z_d \right\|_F^2+\sum_{d=1}^D \left( \langle W_d-W_{(k-1)d}, \nabla g(W_{(k-1)d})\rangle + \frac{t_k}{2}\|W_d-W_{(k-1)d}\|_F^2 + \|W_d\|_*\right),
\end{align*}
where $\langle A, B \rangle=tr(A^\top B)$ denotes the matrix inner product. Thus, equation (8) in \cite{Ji2009} can be calculated component-wisely for $W_1,\cdots, W_D.$

Equation (9) in \cite{Ji2009} can be written as
\begin{align*}
&\frac{t_k}{2}\left\|W-\left(W_{(k-1)}-\frac{1}{t_k}\nabla g(W_{(k-1)})\right)\right\|_F^2 + \|W\|_*\\
=&\sum_{d=1}^D \left( \frac{t_k}{2}\left\|W_d-\left(W_{(k-1)d}-\frac{1}{t_k}\nabla g(W_{(k-1)d})\right)\right\|_F^2 + \|W_d\|_*\right).
\end{align*}
Thus the minimization of equation (9) in \cite{Ji2009} can be performed on $W_1,\cdots, W_D$ separately using singular value decomposition as in Theorem 3.1 of \cite{Ji2009}.

\section{Additional simulation results}
Figures \ref{fig:FAR1_k2_PE} and \ref{fig:FAR1_k8_PE} give the boxplot of PE by Bosq, ANH and RKHS for FAR(1) under signal strength $\kappa=0.2$ and $0.8$ respectively. Figure \ref{fig:FAR_orderselection_PE_k1_0_k2_5} give the boxplot of PE by ANH and RKHS for FAR with autoregressive order selection under signal strength $(\kappa_1, \kappa_2)=(0,0.5)$.

\begin{figure}[h]
	\begin{subfigure}{0.32\textwidth}
		\includegraphics[angle=270, width=1.25\textwidth]{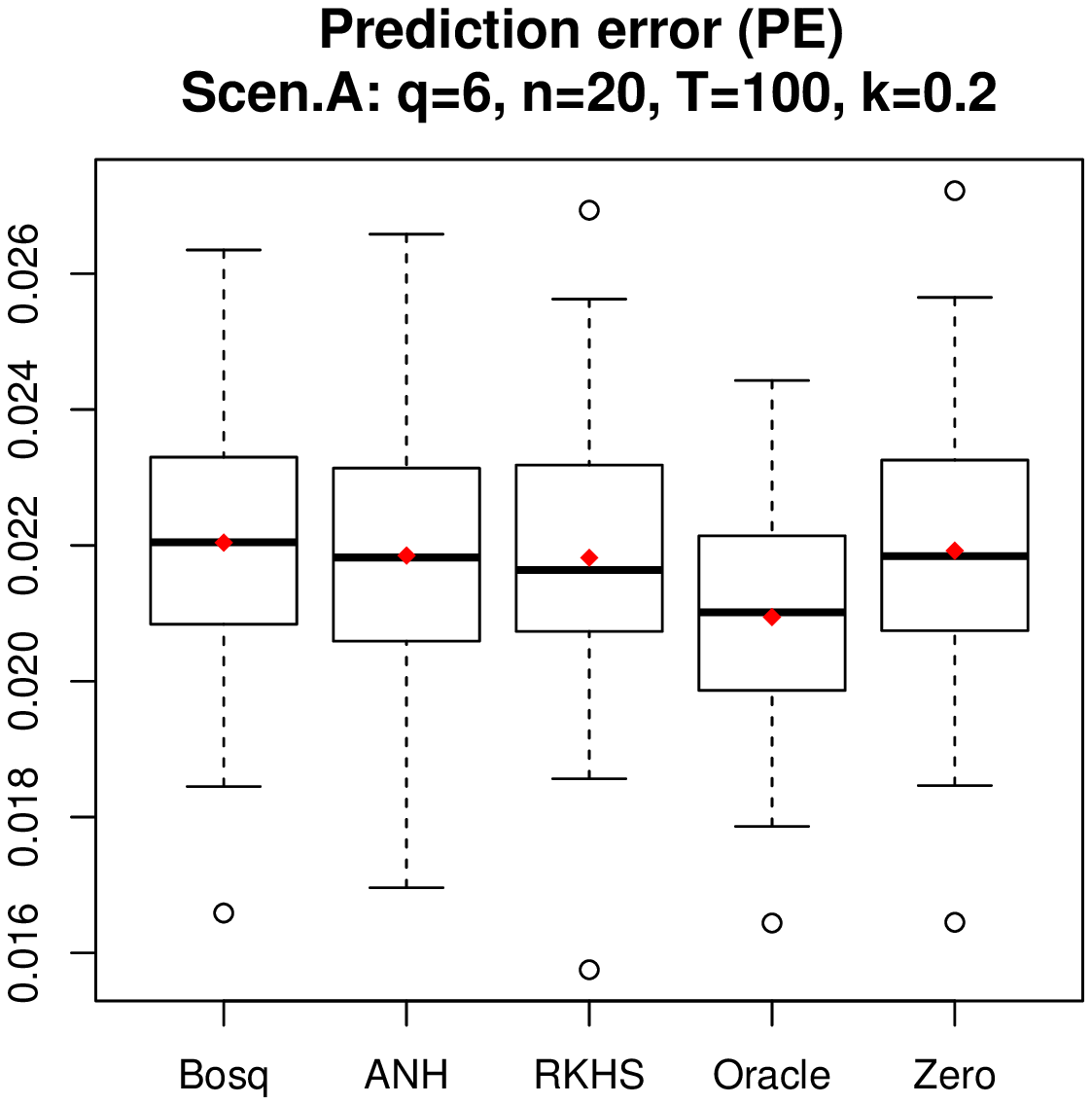}
		\vspace{-1cm}
	\end{subfigure}
	~
	\begin{subfigure}{0.32\textwidth}
		\includegraphics[angle=270, width=1.25\textwidth]{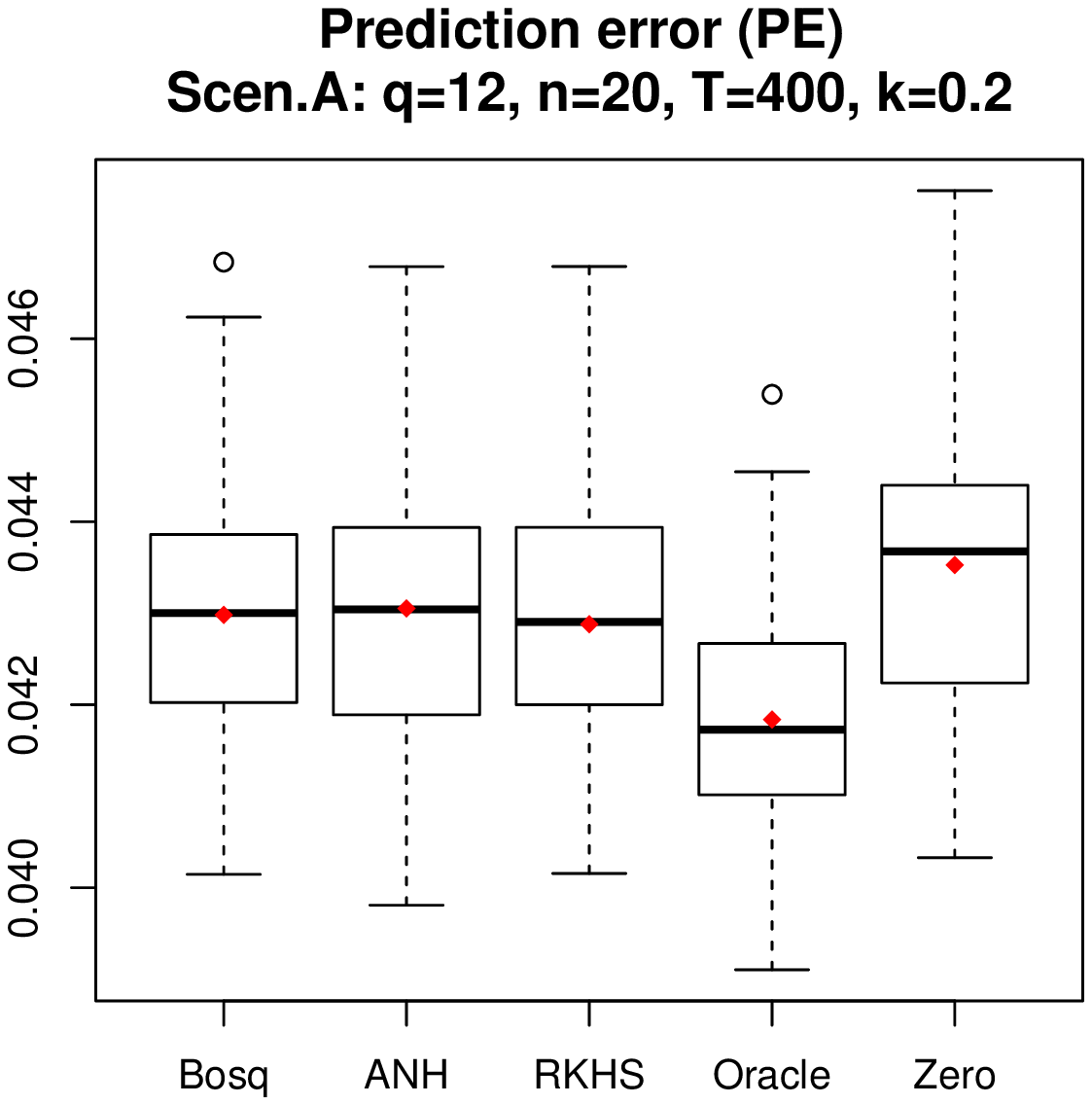}
		\vspace{-1cm}
	\end{subfigure}
	~
	\begin{subfigure}{0.32\textwidth}
		\includegraphics[angle=270, width=1.25\textwidth]{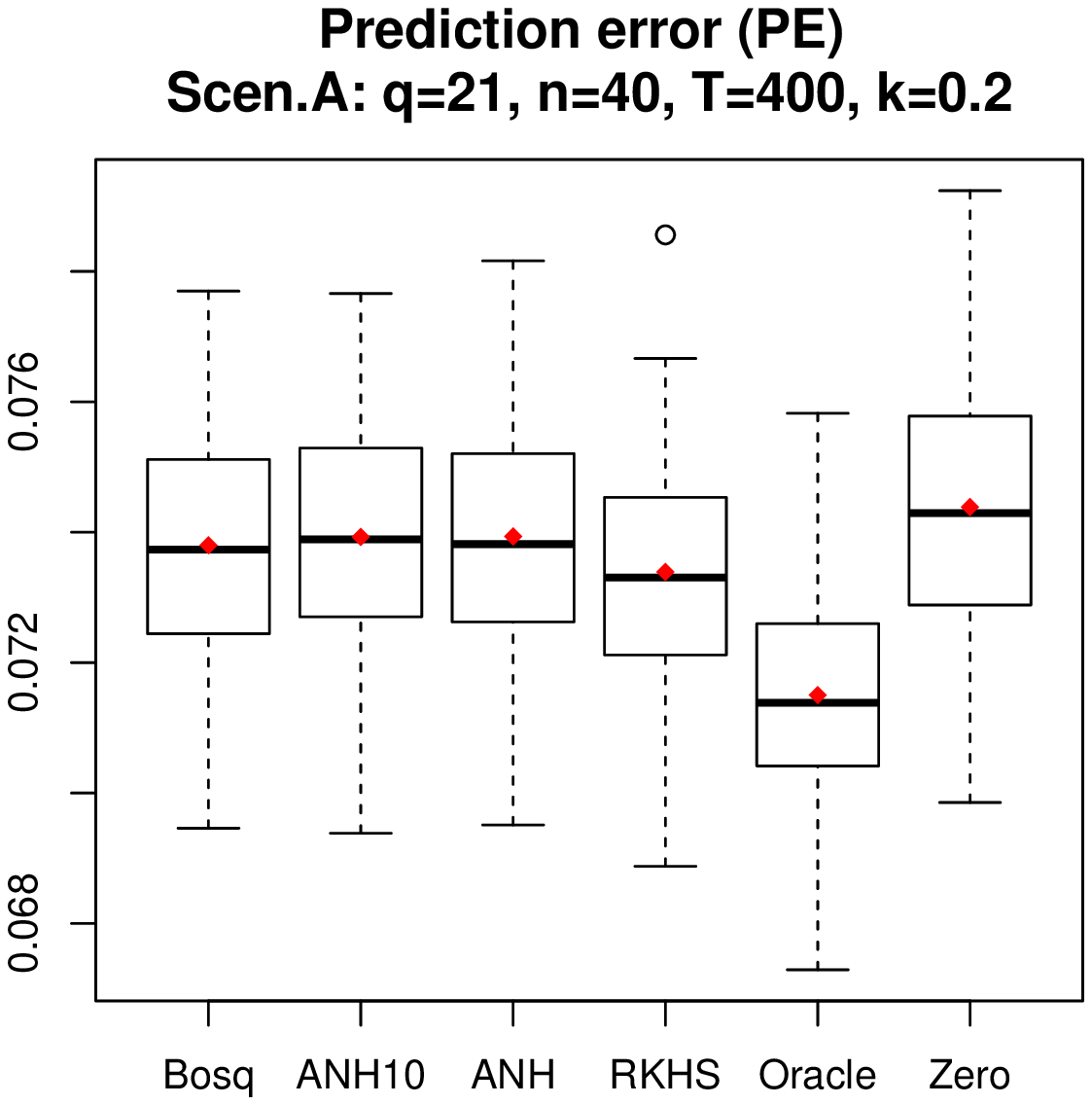}
		\vspace{-1cm}
	\end{subfigure}
	~
	\begin{subfigure}{0.32\textwidth}
		\includegraphics[angle=270, width=1.25\textwidth]{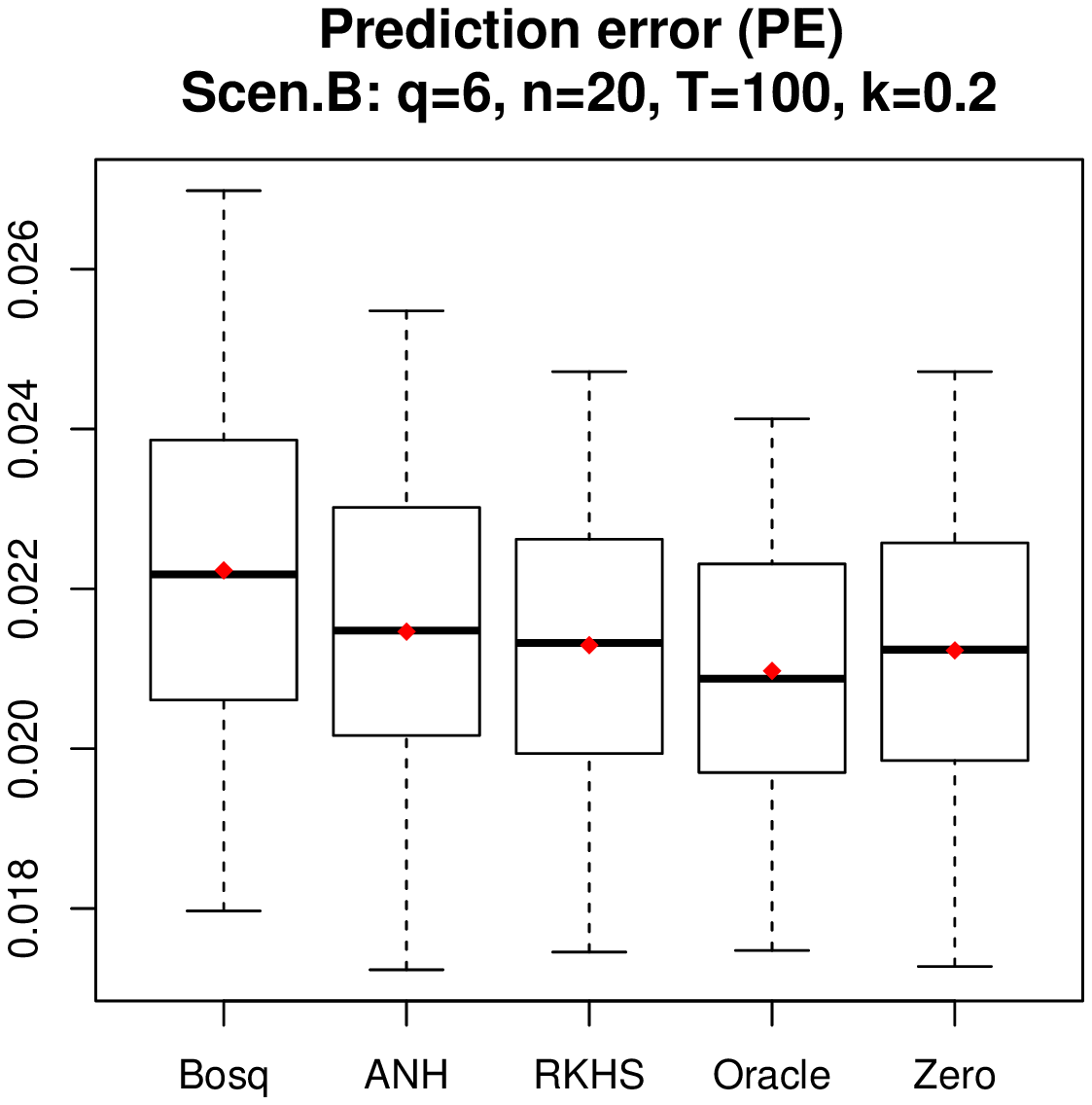}
		\vspace{-1cm}
	\end{subfigure}
	~
	\begin{subfigure}{0.32\textwidth}
		\includegraphics[angle=270, width=1.25\textwidth]{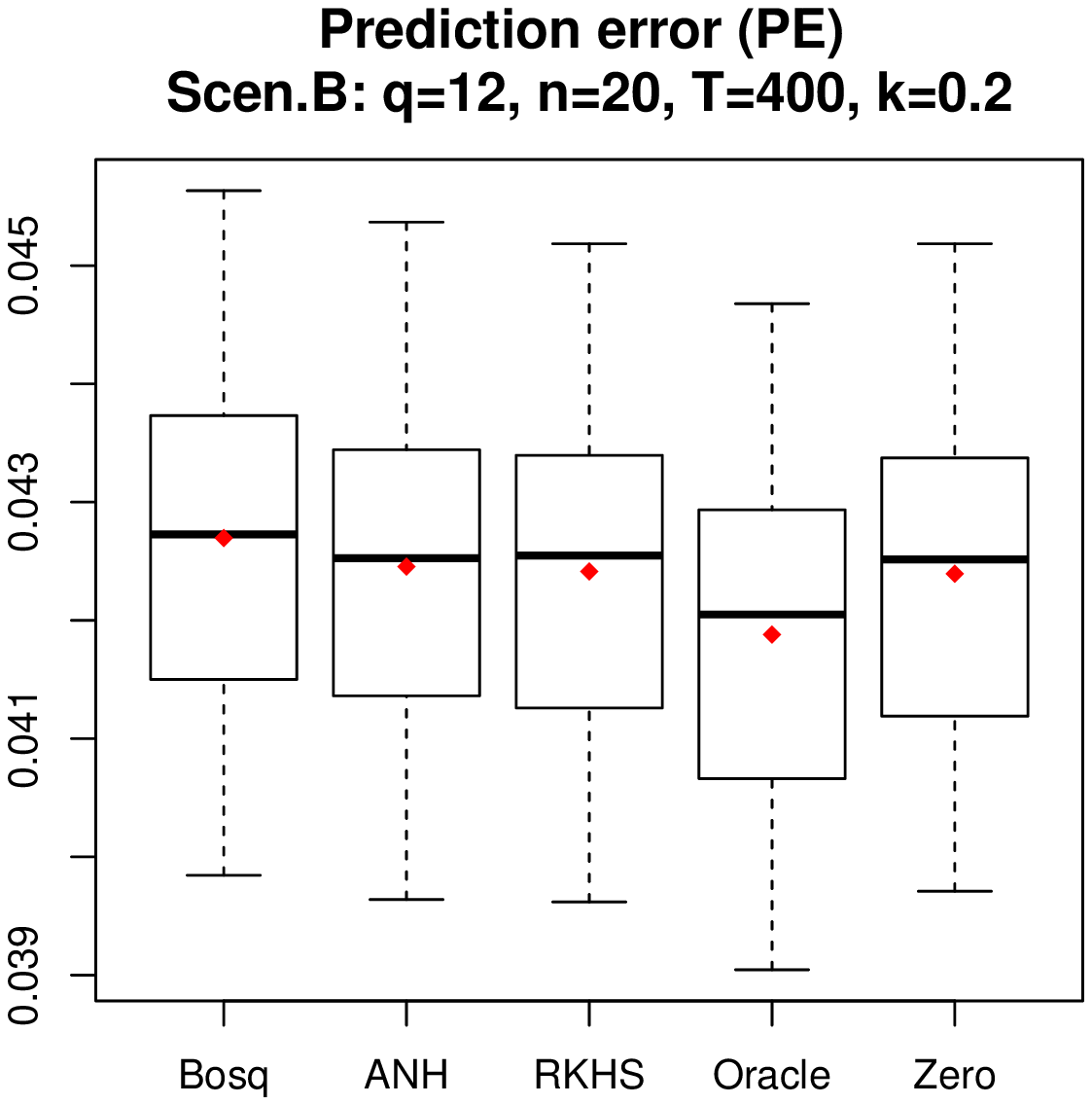}
		\vspace{-1cm}
	\end{subfigure}
	~
	\begin{subfigure}{0.32\textwidth}
		\includegraphics[angle=270, width=1.25\textwidth]{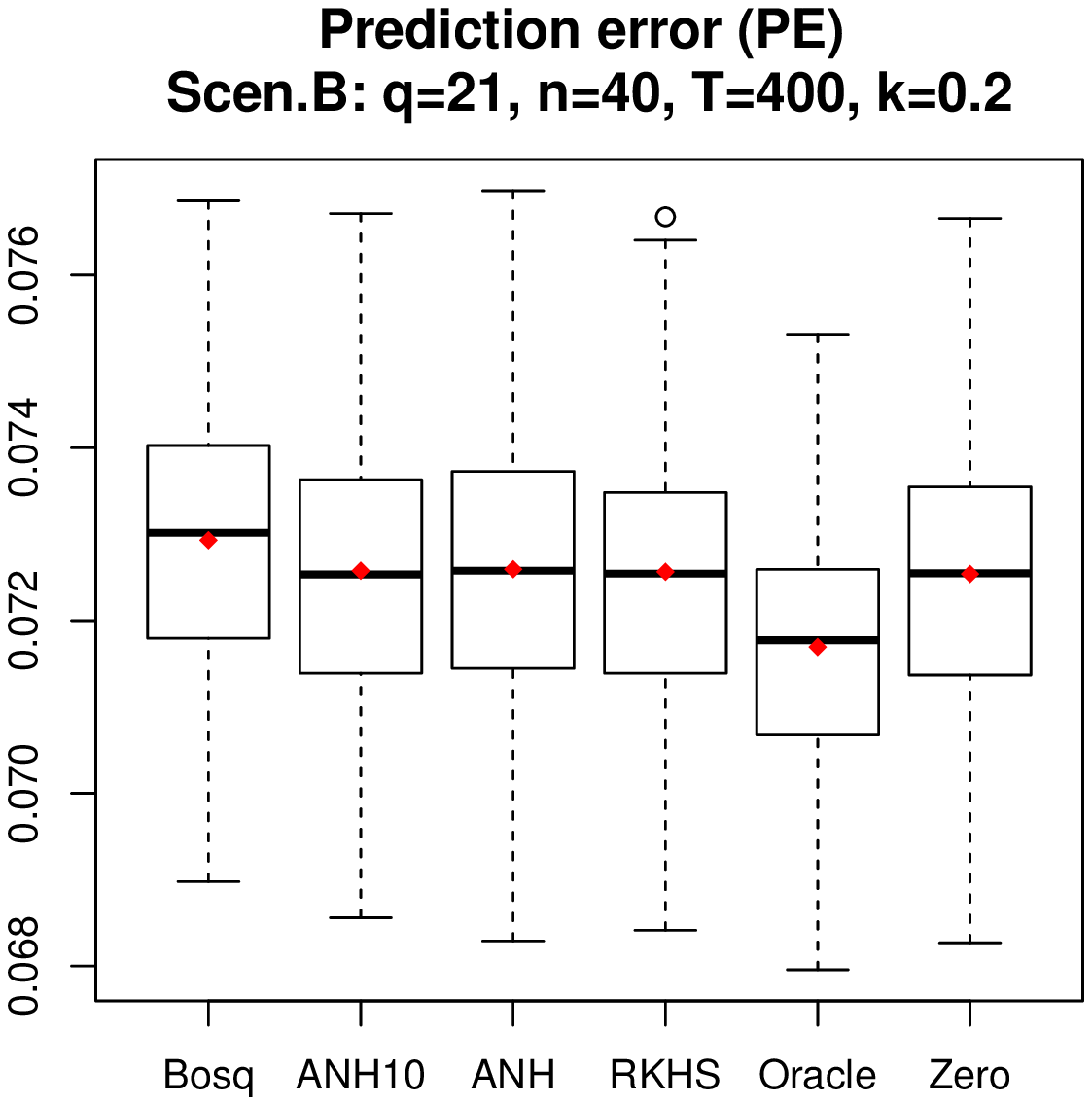}
		\vspace{-1cm}
	\end{subfigure}
	~
	\begin{subfigure}{0.32\textwidth}
		\includegraphics[angle=270, width=1.25\textwidth]{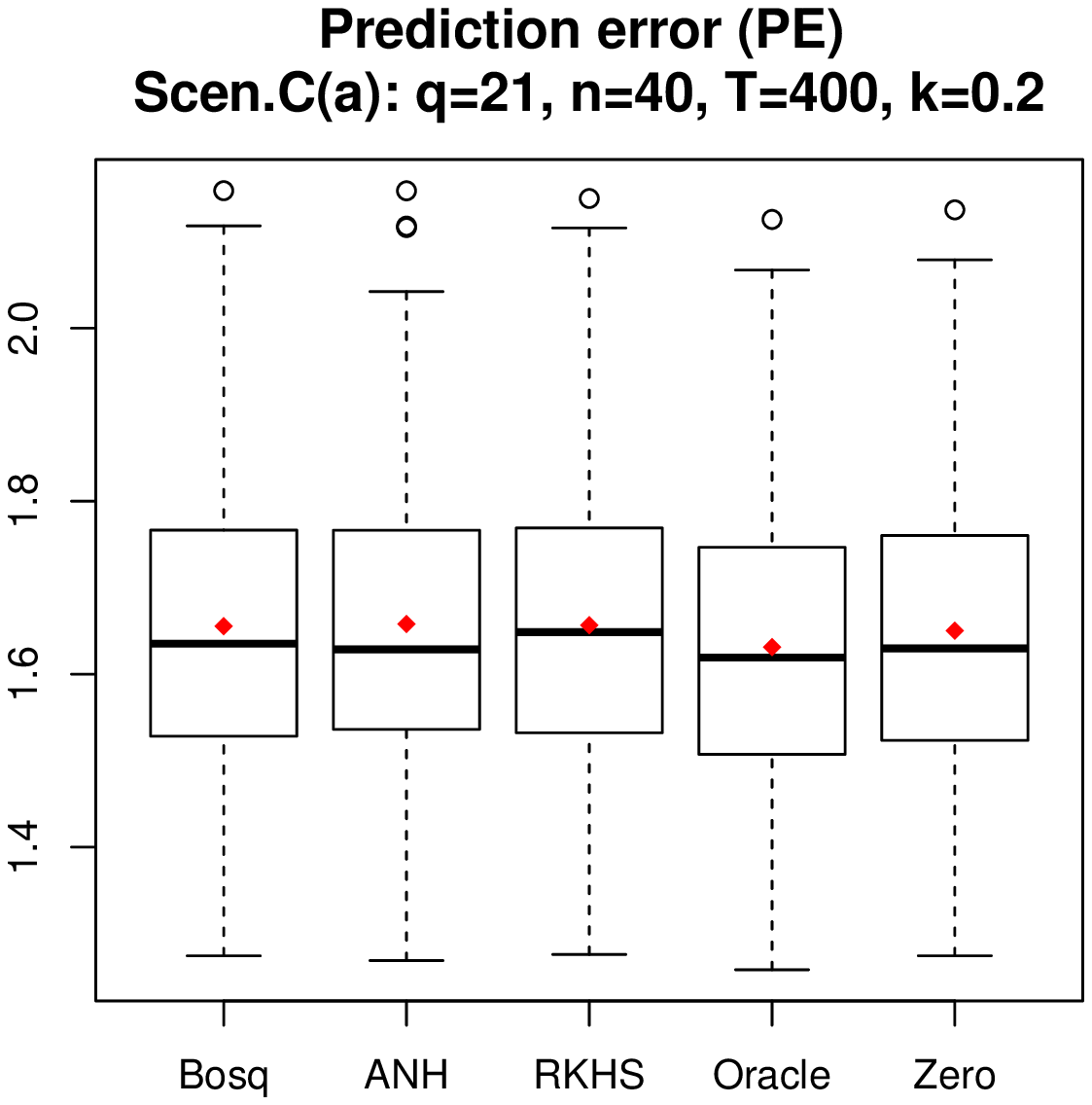}
		\vspace{-1cm}
	\end{subfigure}
	~
	\begin{subfigure}{0.32\textwidth}
		\includegraphics[angle=270, width=1.25\textwidth]{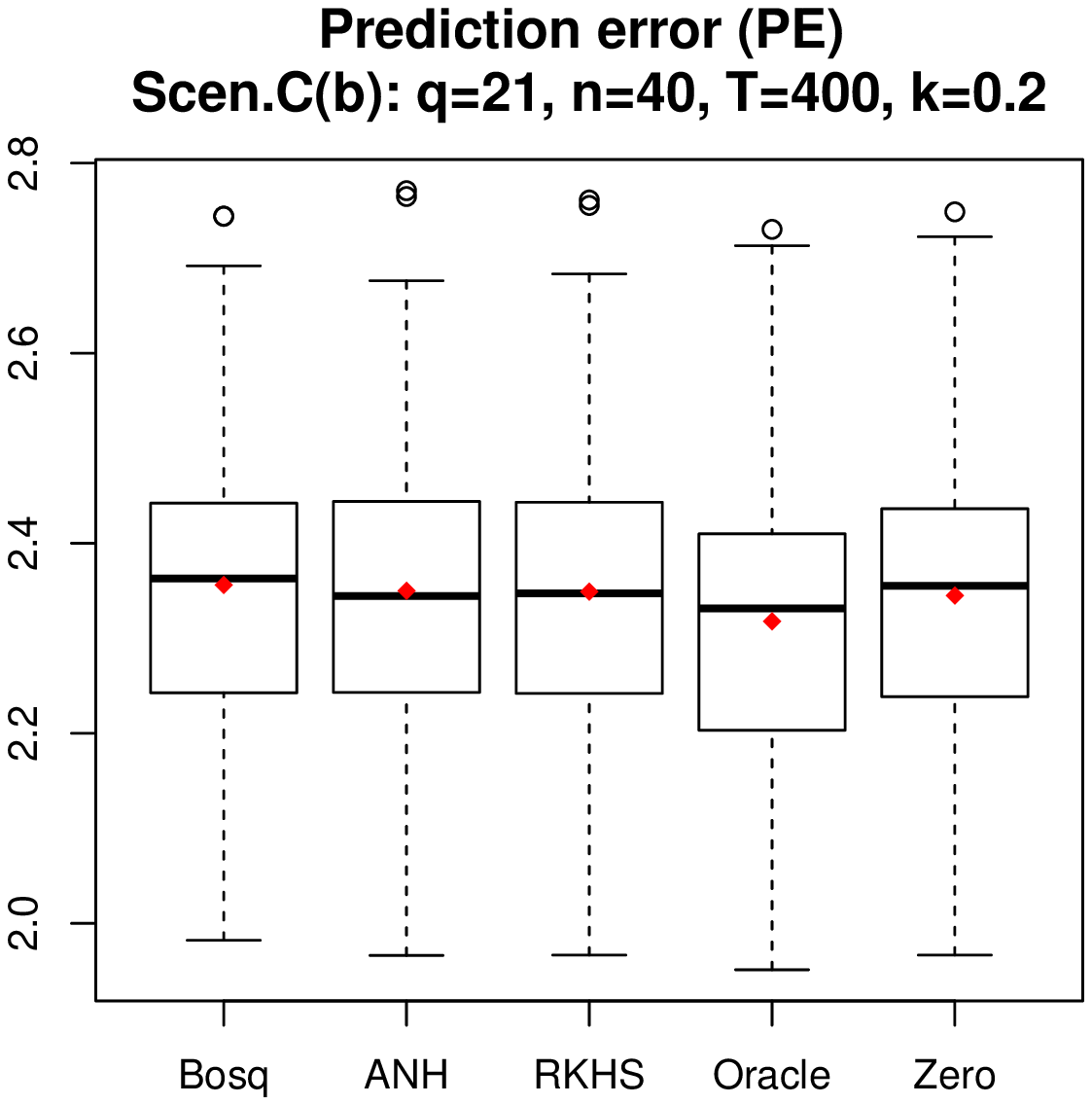}
		\vspace{-1cm}
	\end{subfigure}
	\caption{Boxplot of prediction error~(PE) for FAR(1) across 100 experiments with signal strength $\kappa=0.2.$ ANH10 stands for ANH based on 10 cubic B-splines under $q=21$.}
	\label{fig:FAR1_k2_PE}
\end{figure}

\begin{figure}[h]
	\begin{subfigure}{0.32\textwidth}
		\includegraphics[angle=270, width=1.25\textwidth]{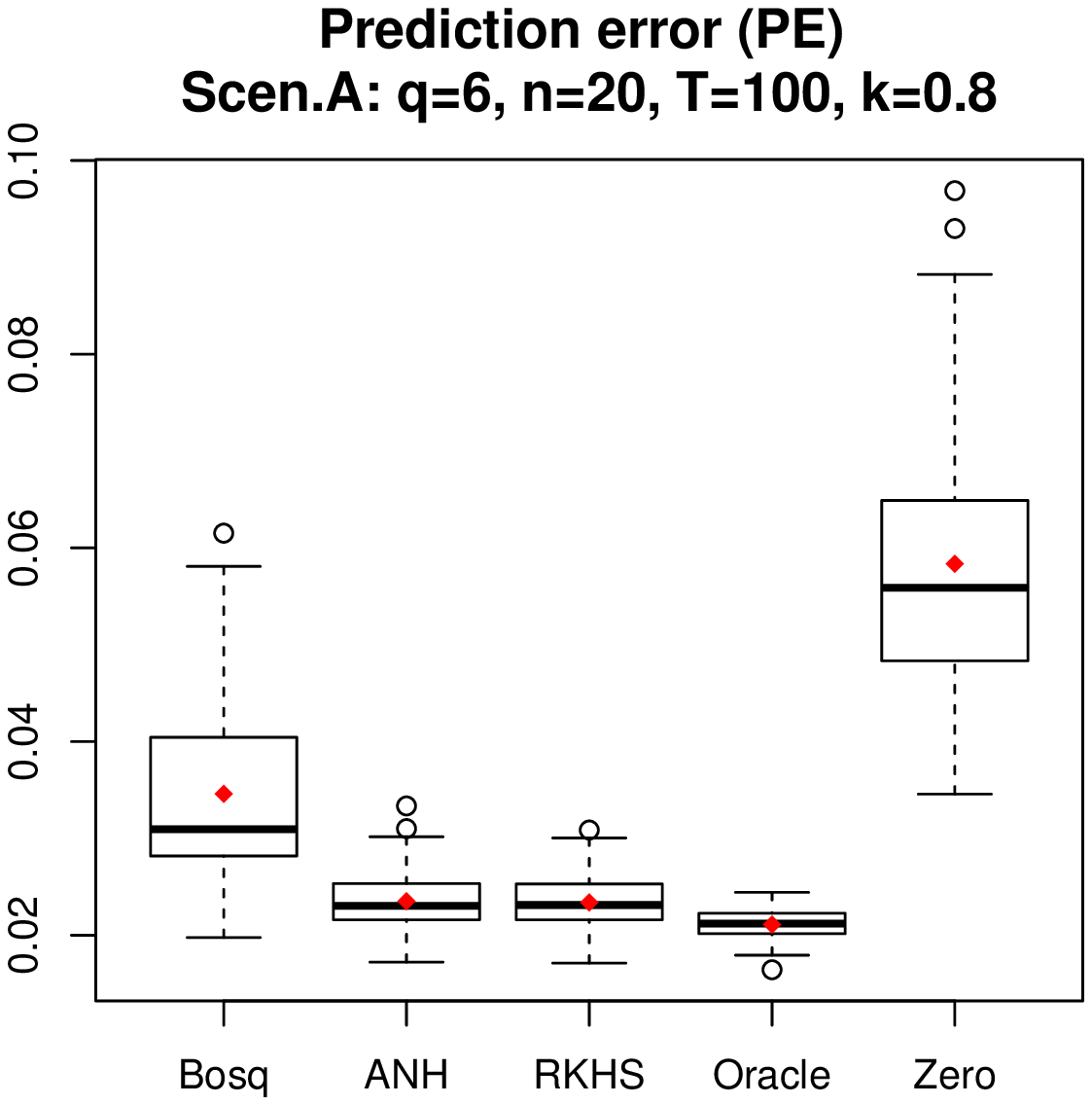}
		\vspace{-1cm}
	\end{subfigure}
	~
	\begin{subfigure}{0.32\textwidth}
		\includegraphics[angle=270, width=1.25\textwidth]{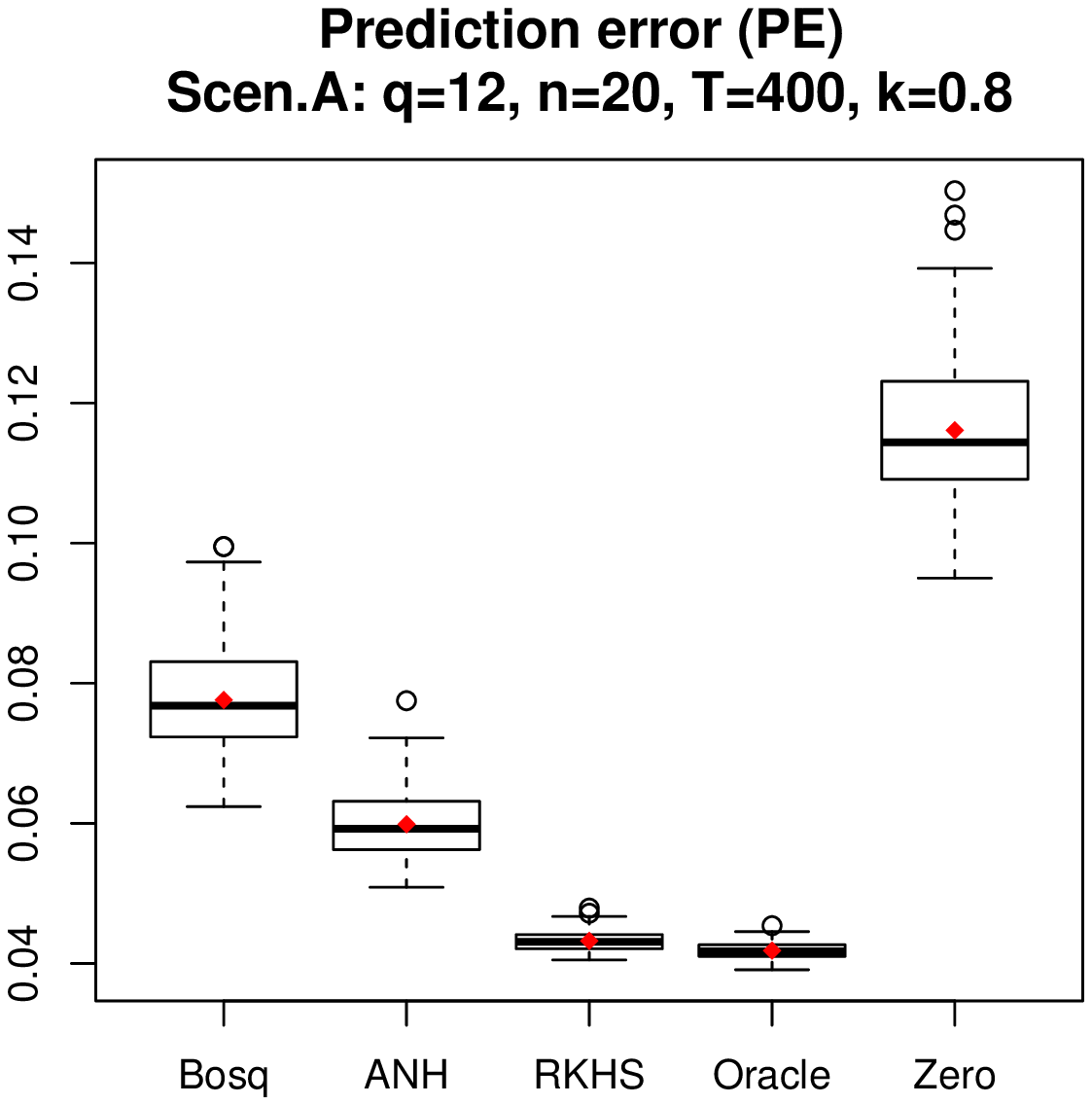}
		\vspace{-1cm}
	\end{subfigure}
	~
	\begin{subfigure}{0.32\textwidth}
		\includegraphics[angle=270, width=1.25\textwidth]{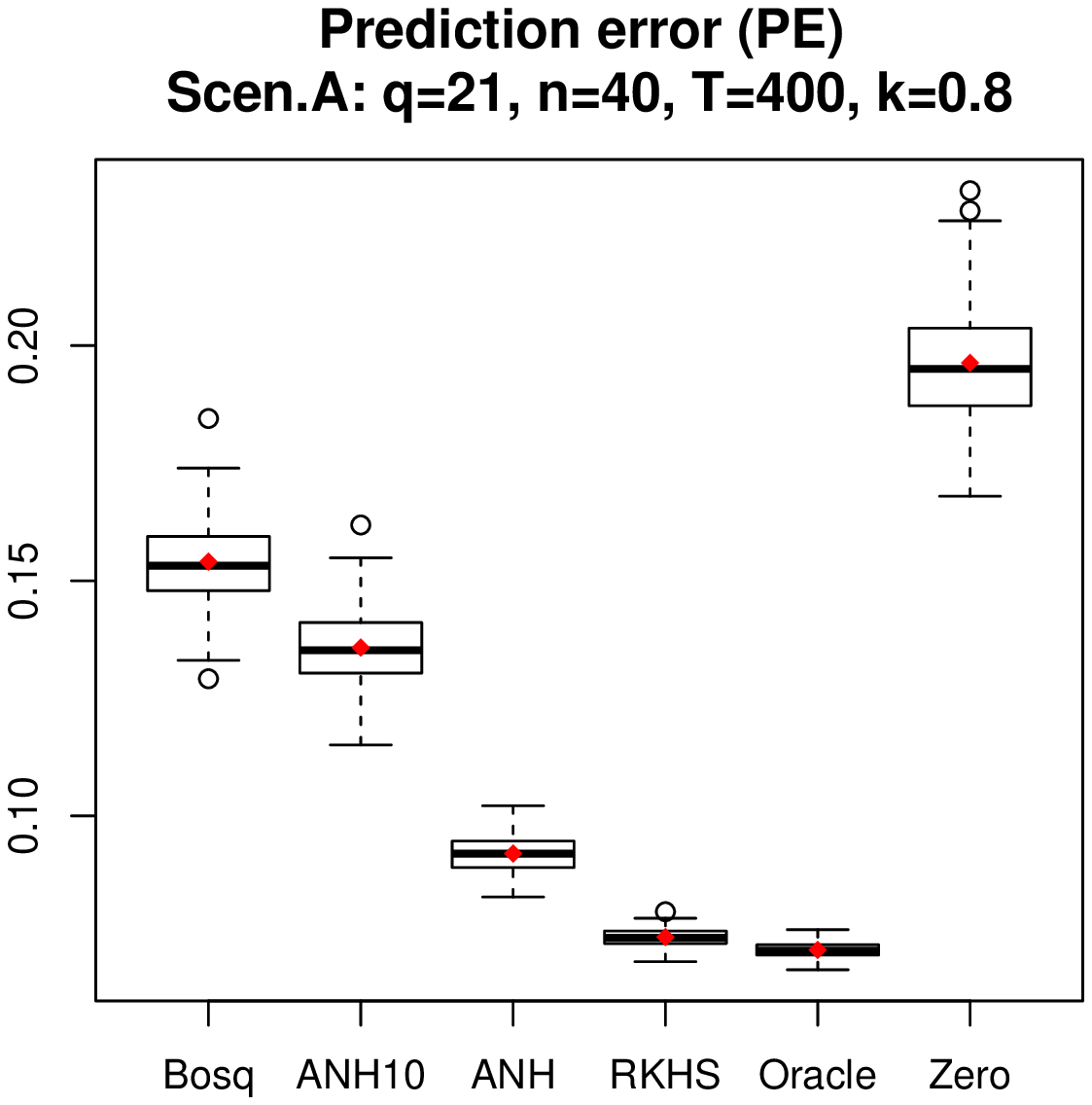}
		\vspace{-1cm}
	\end{subfigure}
	~
	\begin{subfigure}{0.32\textwidth}
		\includegraphics[angle=270, width=1.25\textwidth]{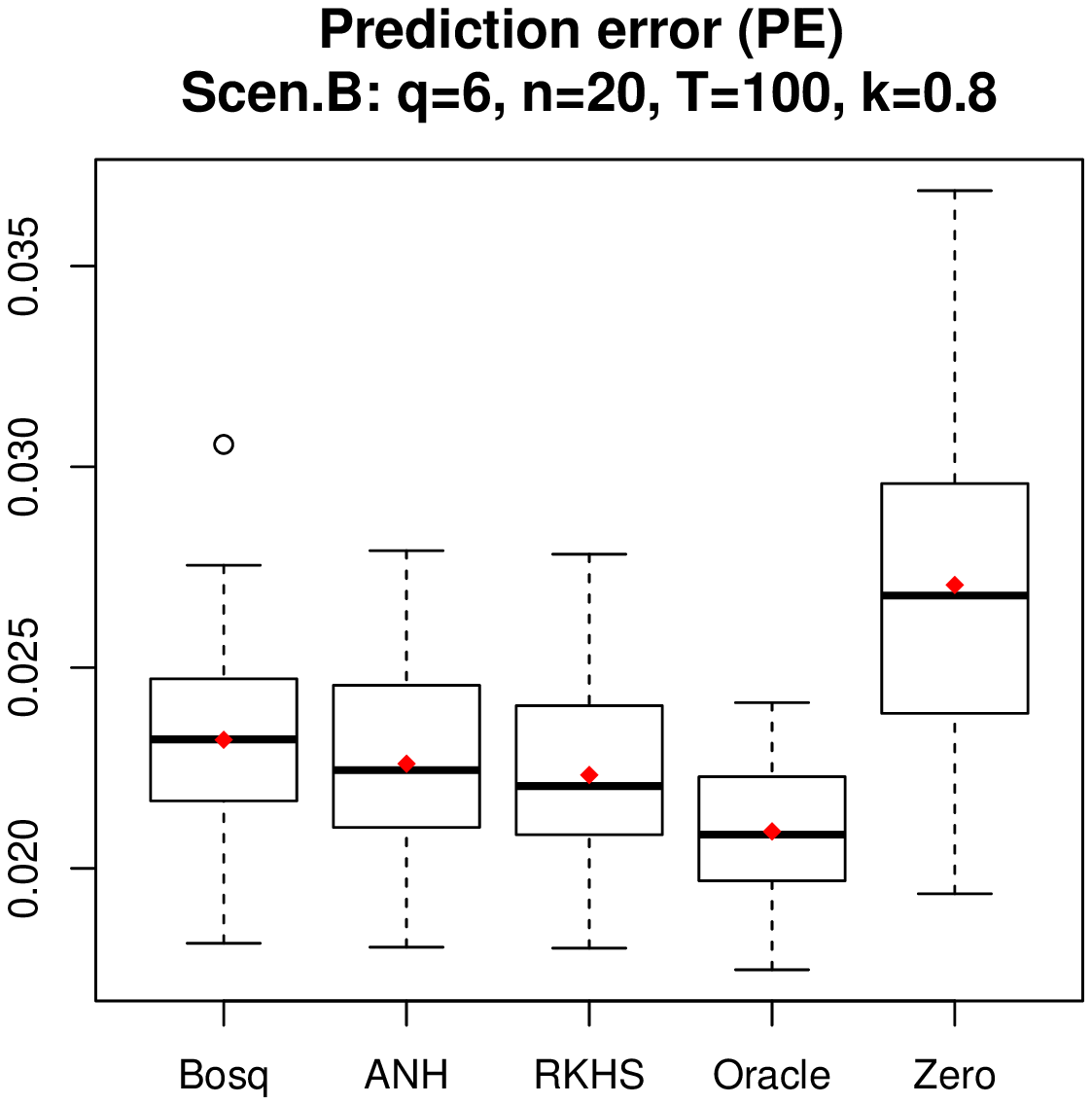}
		\vspace{-1cm}
	\end{subfigure}
	~
	\begin{subfigure}{0.32\textwidth}
		\includegraphics[angle=270, width=1.25\textwidth]{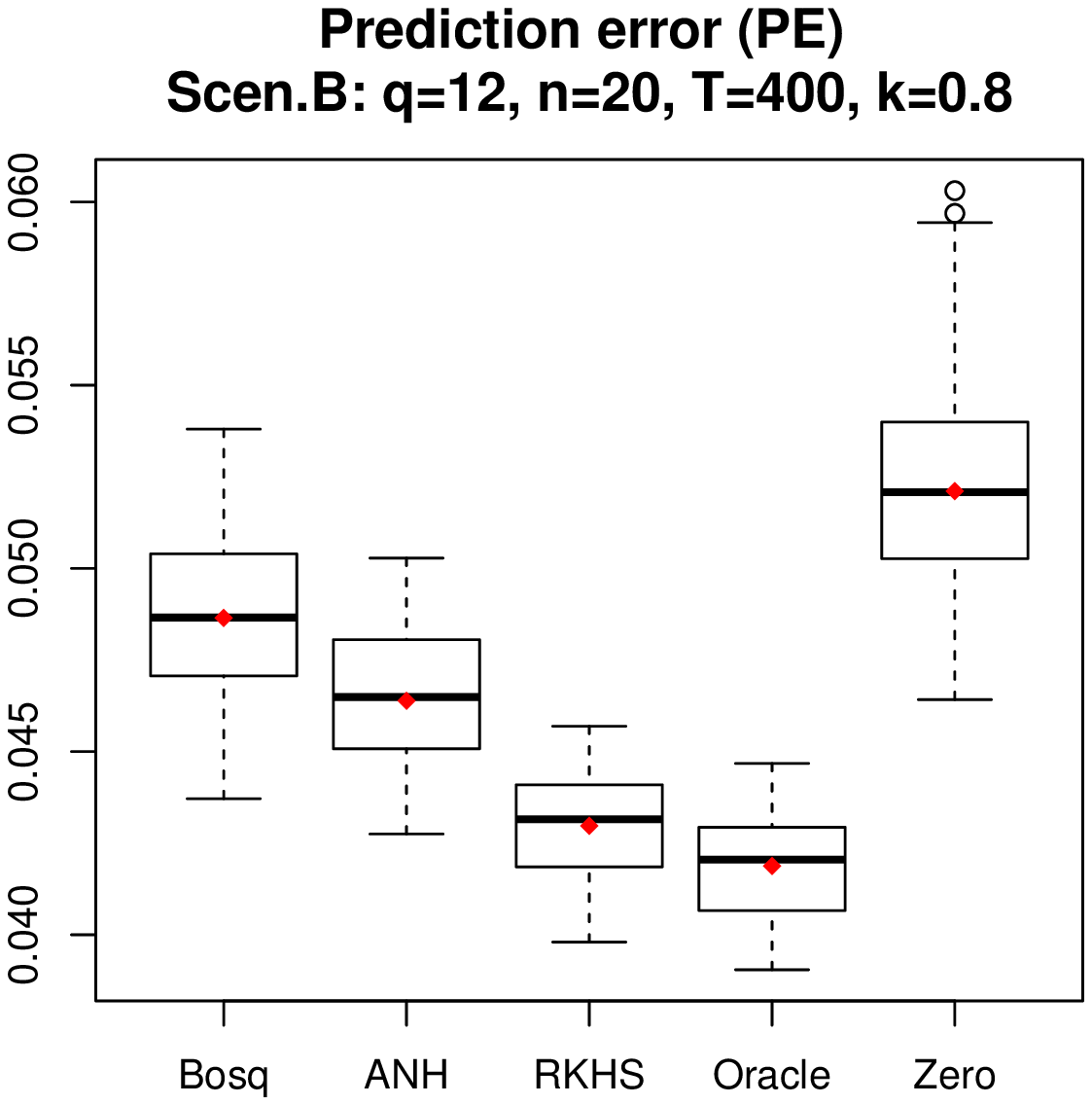}
		\vspace{-1cm}
	\end{subfigure}
	~
	\begin{subfigure}{0.32\textwidth}
		\includegraphics[angle=270, width=1.25\textwidth]{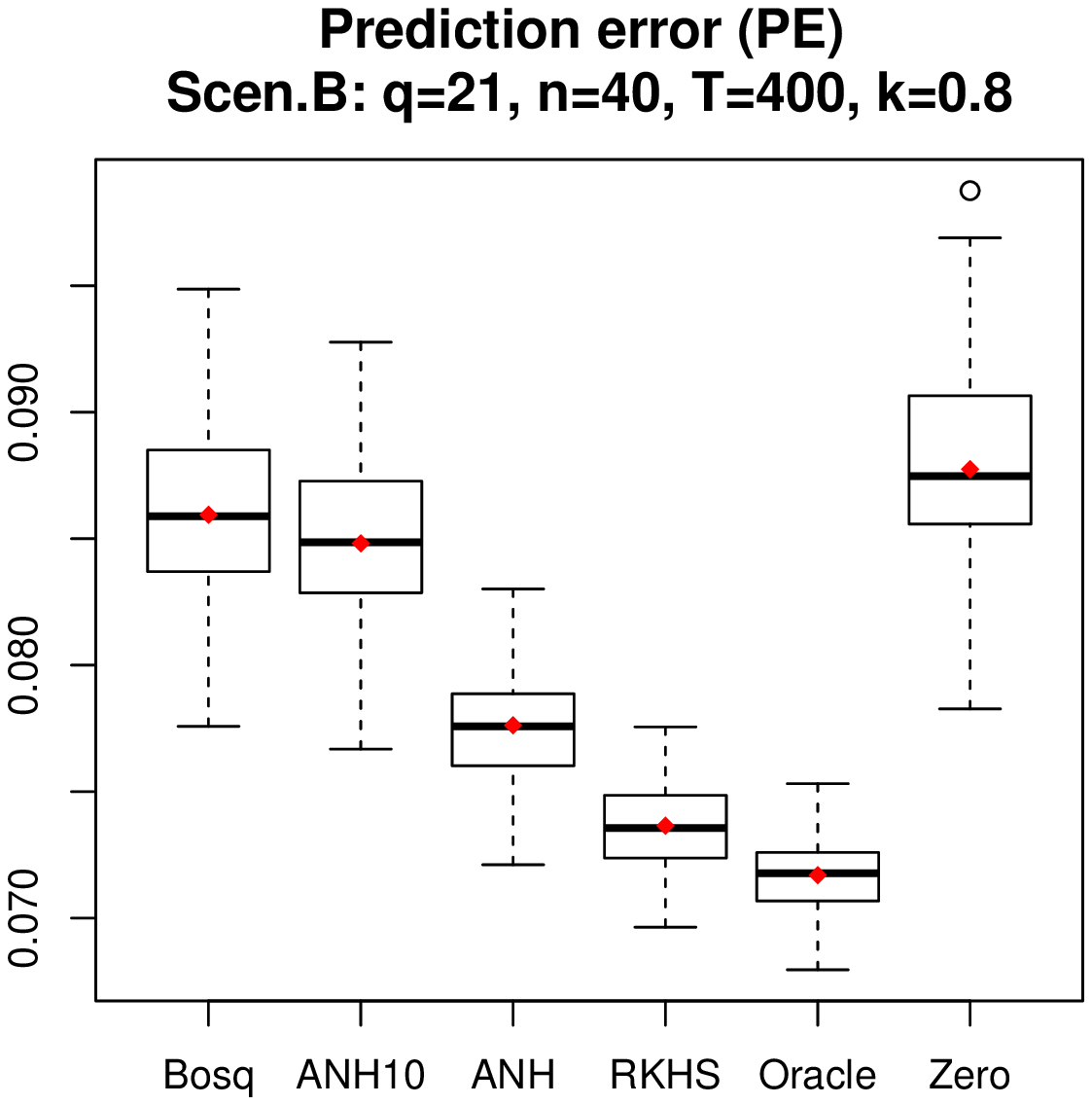}
		\vspace{-1cm}
	\end{subfigure}
	~
	\begin{subfigure}{0.32\textwidth}
		\includegraphics[angle=270, width=1.25\textwidth]{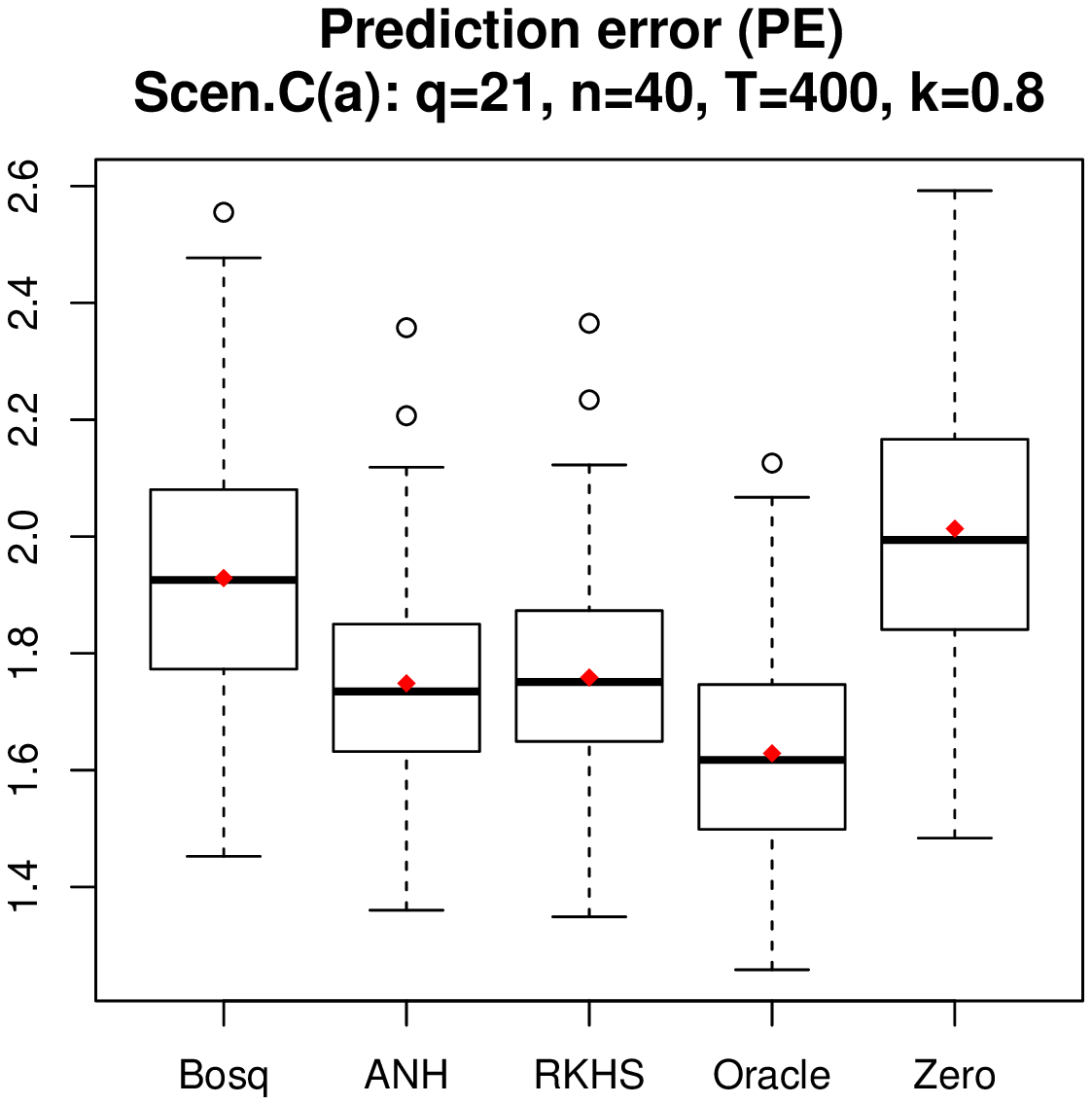}
		\vspace{-1cm}
	\end{subfigure}
	~
	\begin{subfigure}{0.32\textwidth}
		\includegraphics[angle=270, width=1.25\textwidth]{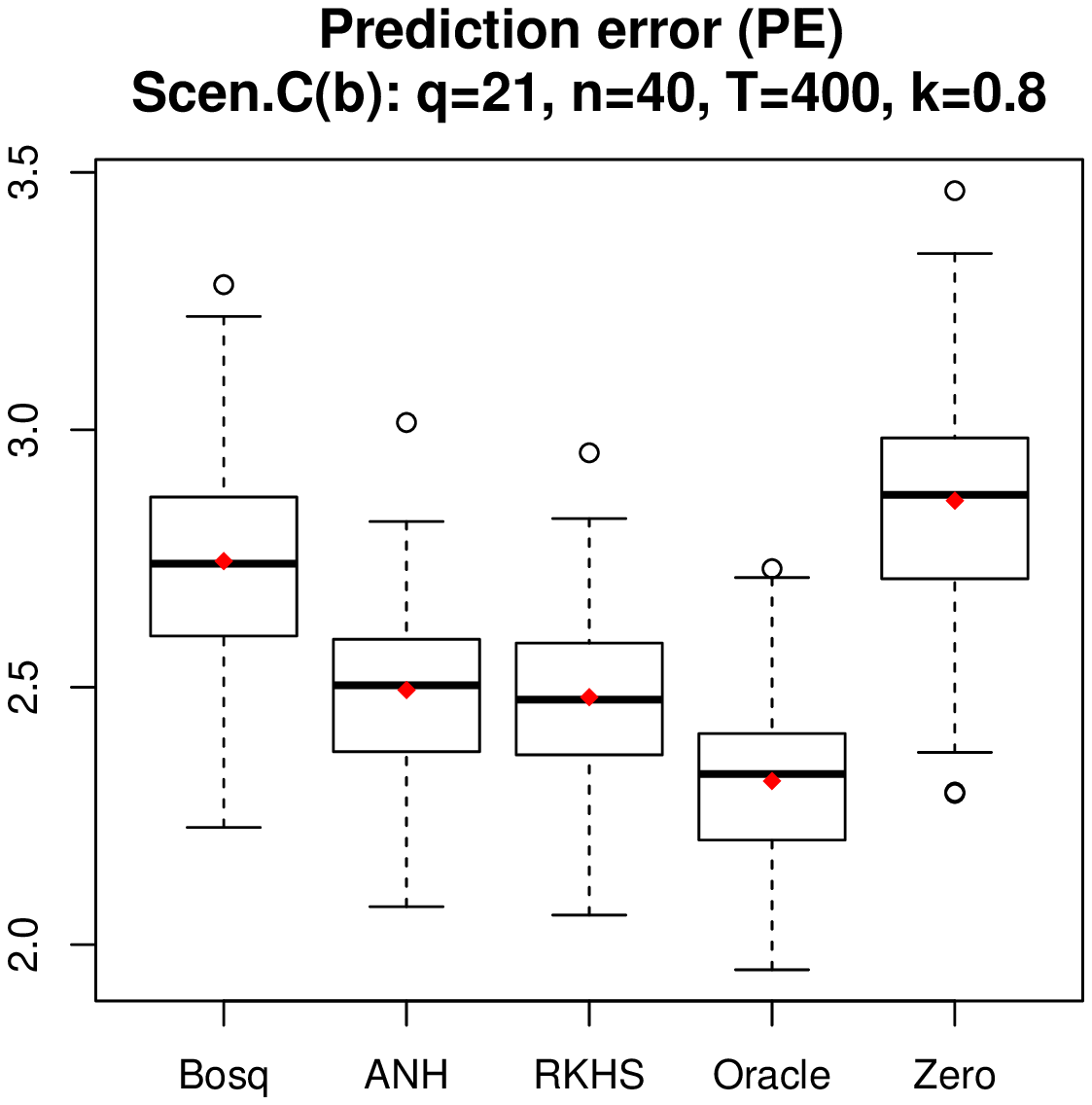}
		\vspace{-1cm}
	\end{subfigure}
	\caption{Boxplot of prediction error~(PE) for FAR(1) across 100 experiments with signal strength $\kappa=0.8.$ ANH10 stands for ANH based on 10 cubic B-splines under $q=21$.}
	\label{fig:FAR1_k8_PE}
\end{figure}

\begin{figure}[h]
	\begin{subfigure}{0.32\textwidth}
		\includegraphics[angle=270, width=1.25\textwidth]{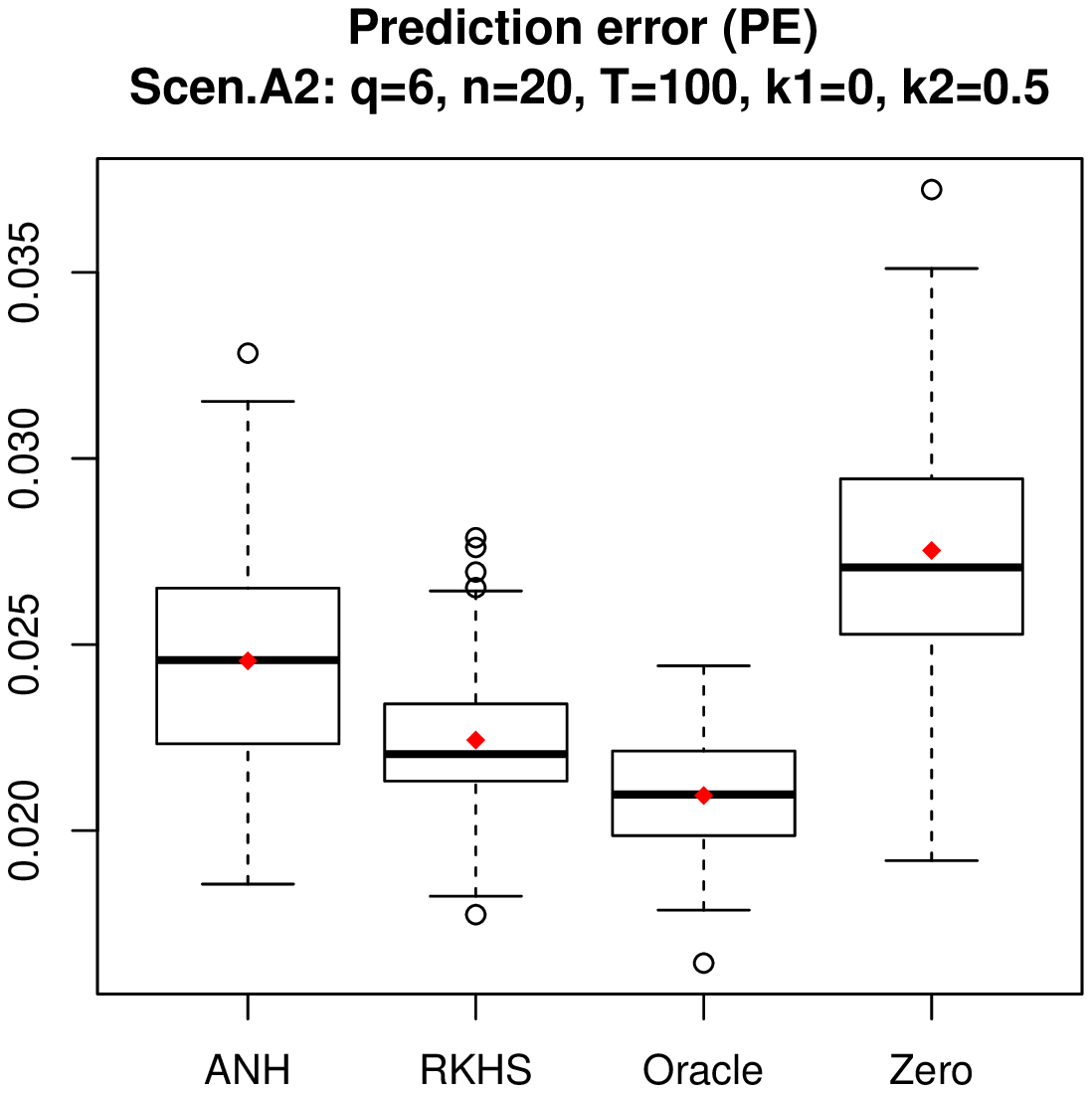}
		\vspace{-1cm}
	\end{subfigure}
	~
	\begin{subfigure}{0.32\textwidth}
		\includegraphics[angle=270, width=1.25\textwidth]{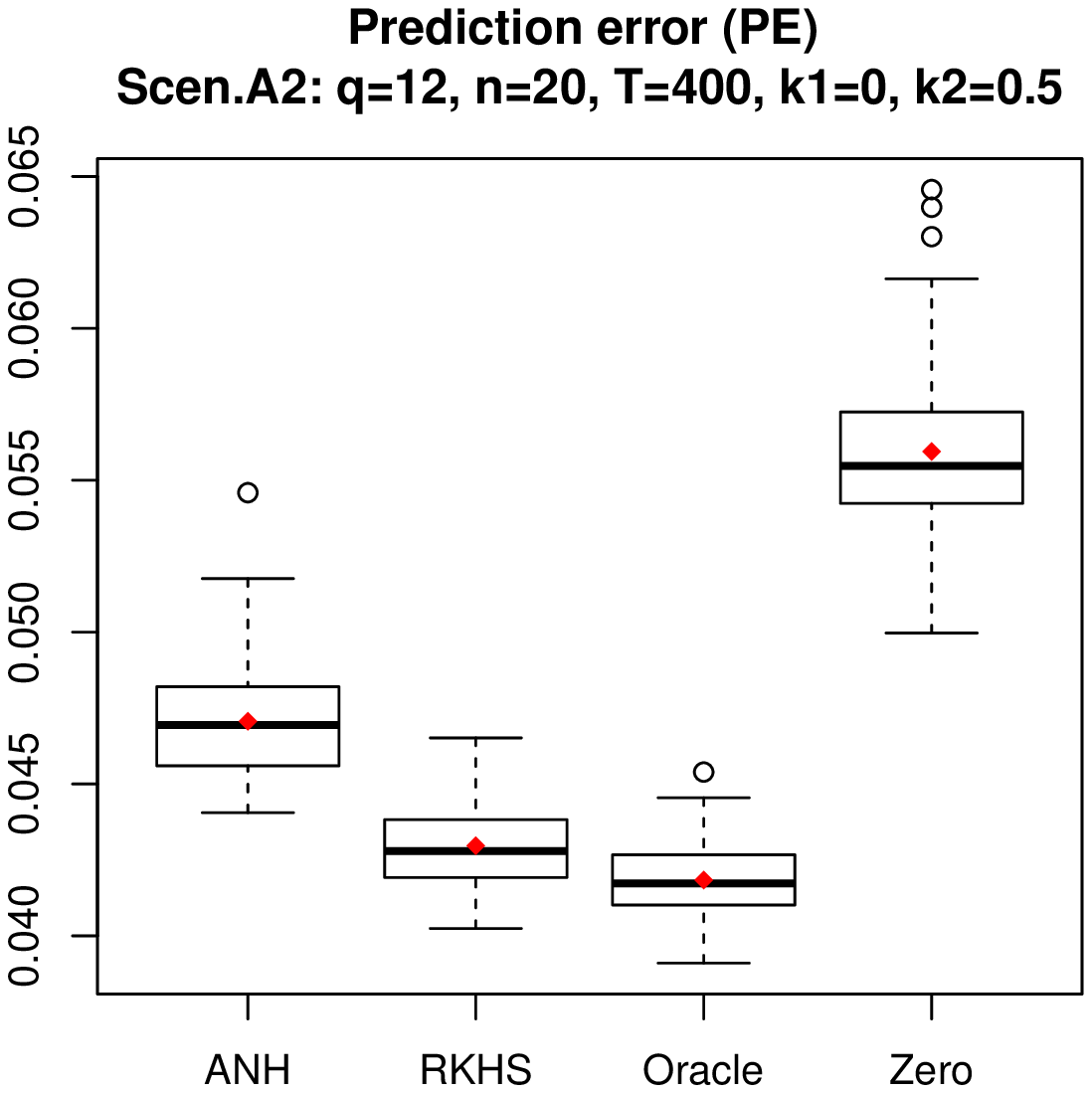}
		\vspace{-1cm}
	\end{subfigure}
	~
	\begin{subfigure}{0.32\textwidth}
		\includegraphics[angle=270, width=1.25\textwidth]{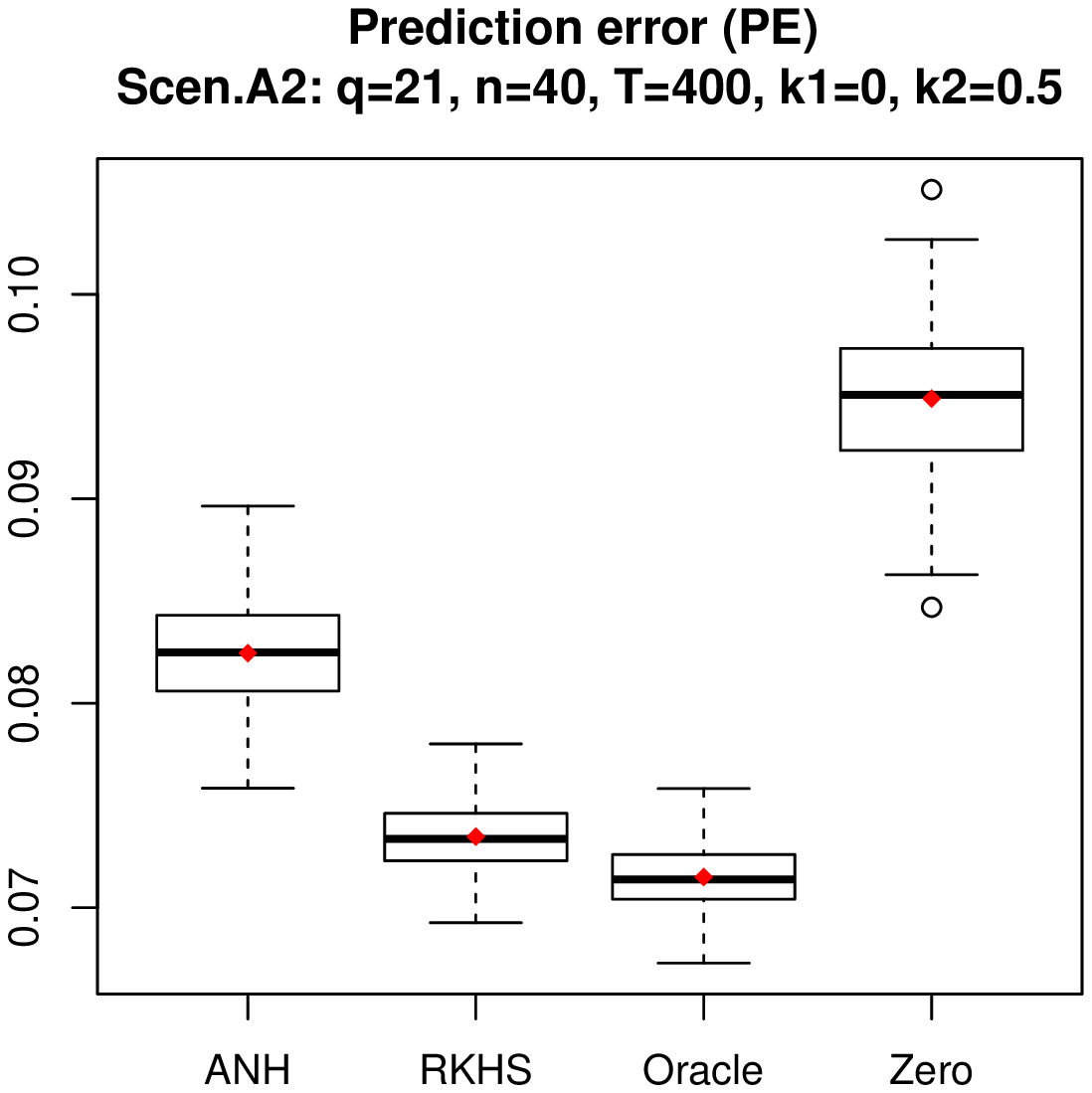}
		\vspace{-1cm}
	\end{subfigure}
	~
	\begin{subfigure}{0.32\textwidth}
		\includegraphics[angle=270, width=1.25\textwidth]{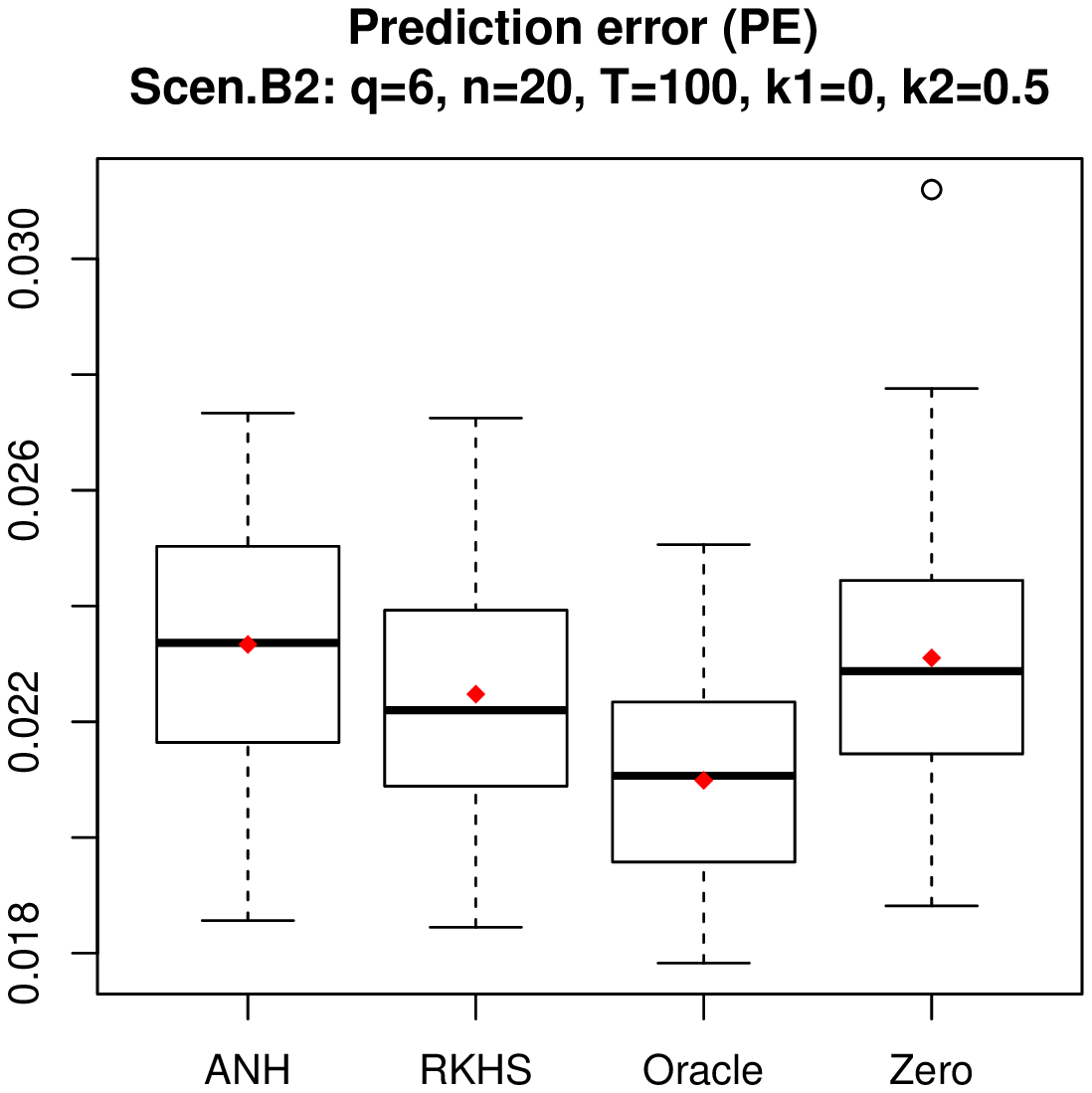}
		\vspace{-1cm}
	\end{subfigure}
	~
	\begin{subfigure}{0.32\textwidth}
		\includegraphics[angle=270, width=1.25\textwidth]{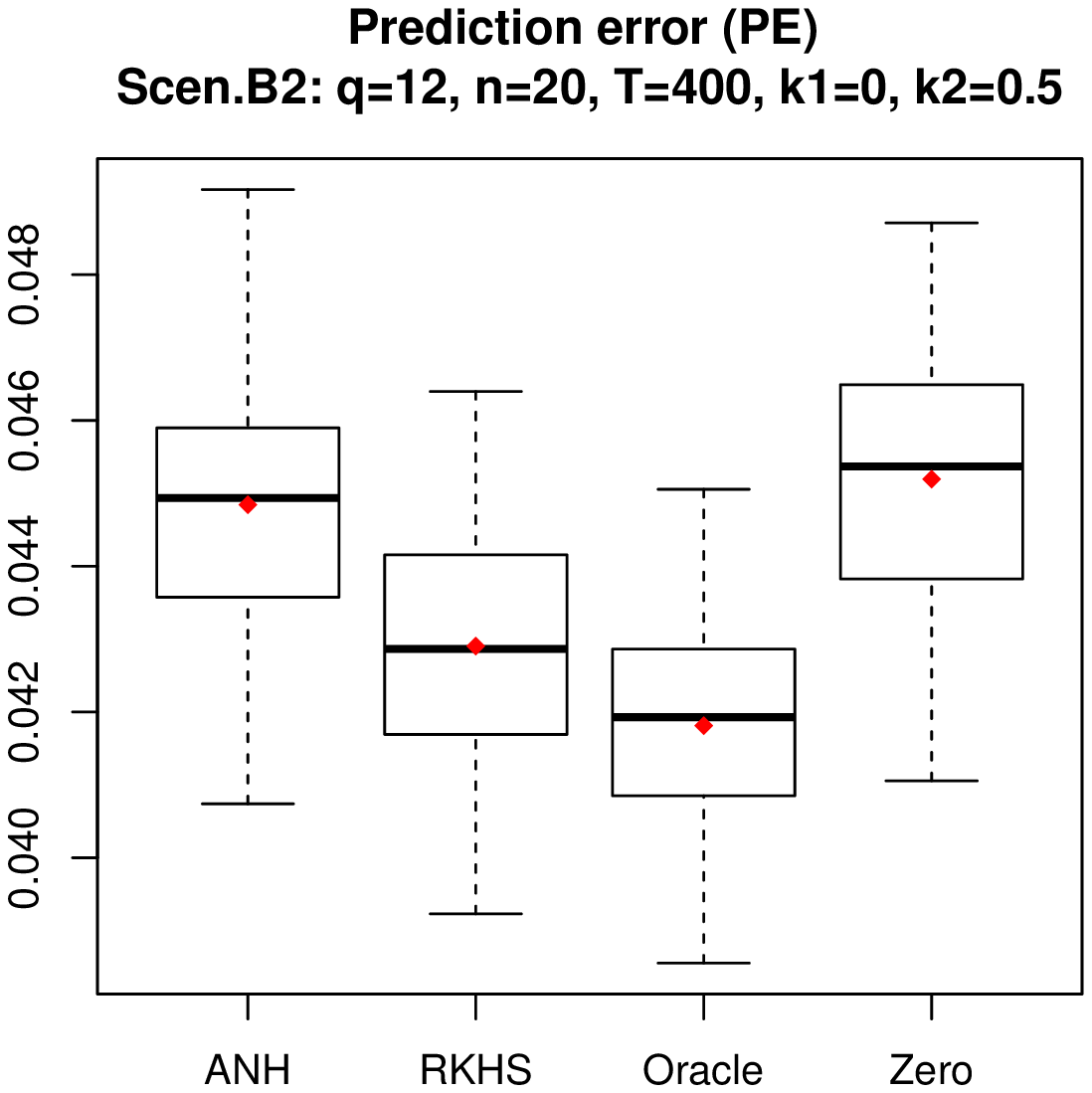}
		\vspace{-1cm}
	\end{subfigure}
	~
	\begin{subfigure}{0.32\textwidth}
		\includegraphics[angle=270, width=1.25\textwidth]{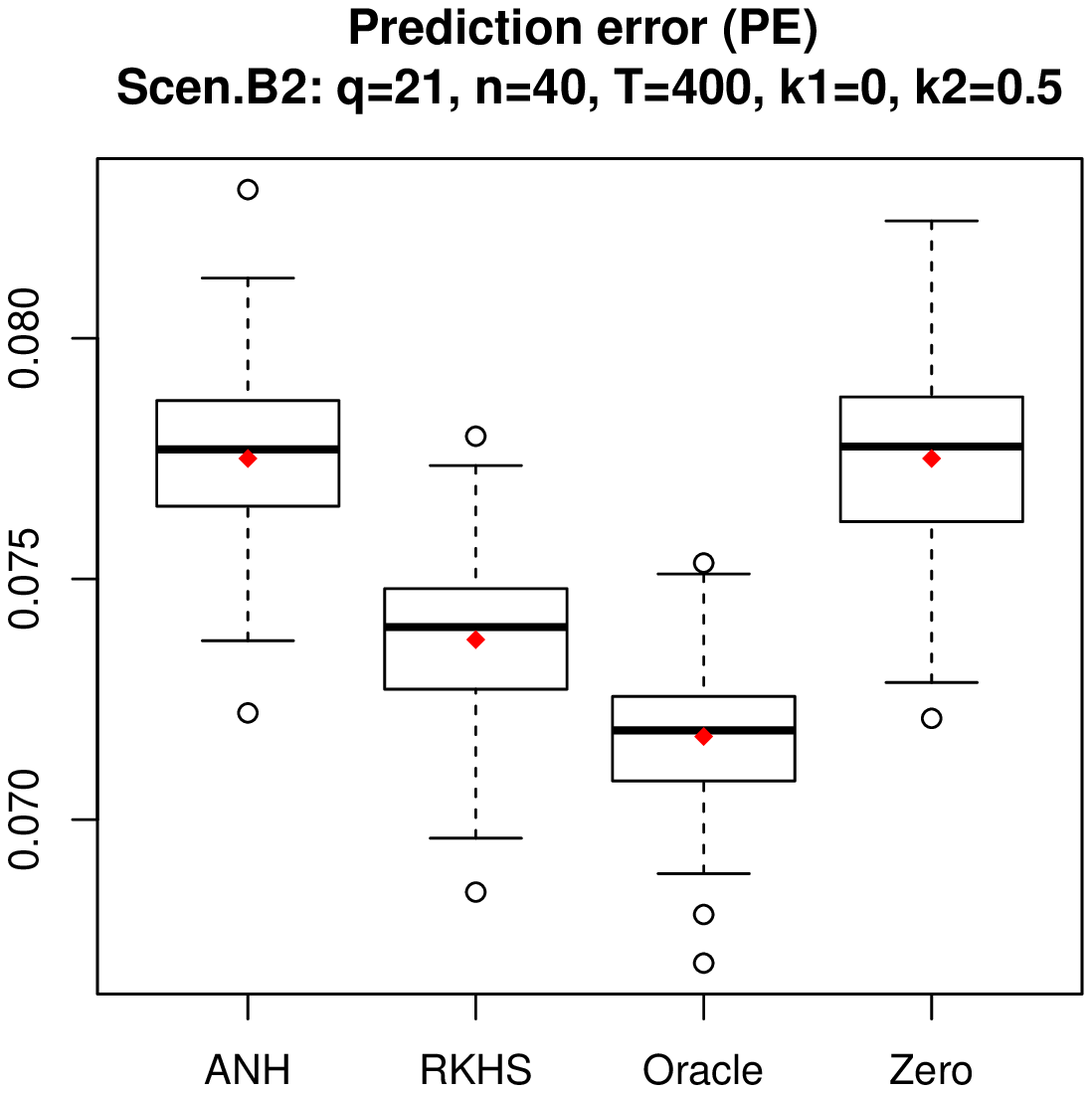}
		\vspace{-1cm}
	\end{subfigure}
	~
	\begin{subfigure}{0.32\textwidth}
		\includegraphics[angle=270, width=1.25\textwidth]{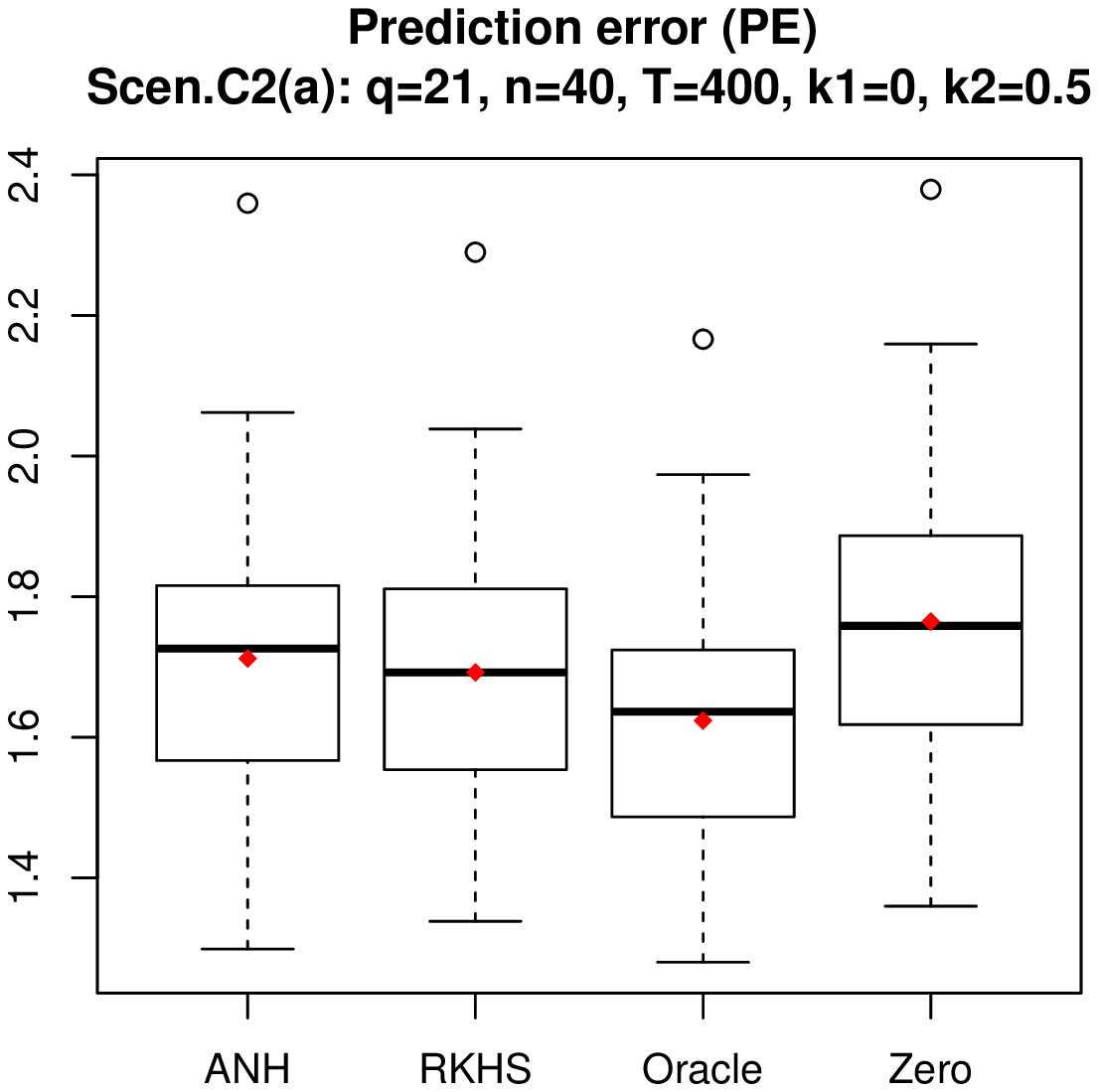}
		\vspace{-1cm}
	\end{subfigure}
	~
	\begin{subfigure}{0.32\textwidth}
		\includegraphics[angle=270, width=1.25\textwidth]{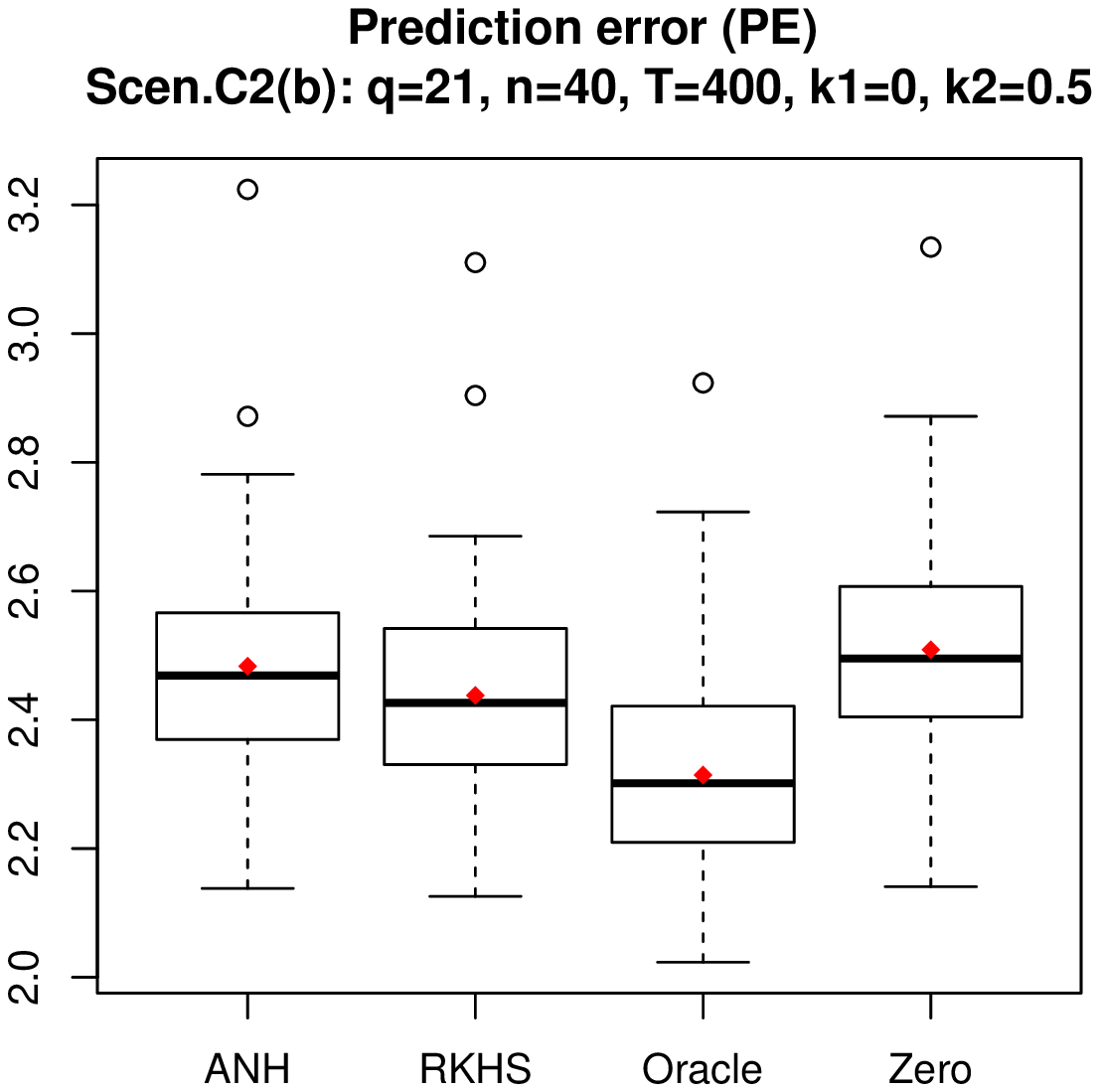}
		\vspace{-1cm}
	\end{subfigure}
	\caption{Boxplot of prediction error~(PE) for FAR with autoregressive order selection across 100 experiments with signal strength $\kappa_1=0, \kappa_2=0.5.$}
	\label{fig:FAR_orderselection_PE_k1_0_k2_5}
\end{figure}
  
\end{document}